%% file: main.tex
\pdfoutput=1 
\documentclass[acmsmall,screen,table]{acmart}
\setcopyright{cc}
\setcctype{by}
\acmDOI{10.1145/3720470}
\acmYear{2025}
\acmJournal{PACMPL}
\acmVolume{9}
\acmNumber{OOPSLA1}
\acmArticle{114}
\acmMonth{4}
\received{2024-10-16}
\received[accepted]{2025-02-18}

\usepackage{natbib}

\usepackage{booktabs}
\usepackage{subcaption}
\usepackage{siunitx}

\usepackage[linesnumbered,ruled,vlined]{algorithm2e}
\usepackage[utf8]{inputenc}
\usepackage{amsmath}
\usepackage{stmaryrd}
\usepackage{soul}
\usepackage{empheq}
\usepackage[breakable]{tcolorbox}
\usepackage{multirow}
\usepackage{wrapfig}
\usepackage{fancyvrb}
\usepackage{listings}
\usepackage{thmtools,thm-restate}
\usepackage{hhline}
\usepackage{setspace}
\usepackage[normalem]{ulem}
\usepackage{multicol}
\usepackage{graphics}

\SetCommentSty{mycommfont}
\newcommand{\mypar}[1]{\vspace{1mm}\textit{#1.}}

\newcommand{\bigex}{%
\mathop{\lower0.75ex\hbox{%
   \scalebox{1.7}{\ensuremath{\exists}}}}\limits}

\newcommand{\numf}[1]{\num[group-separator={,}, group-minimum-digits={3}]{#1}}

\xspaceaddexceptions{[]\{\}}
\newcommand{\rone}{(\emph{i})\xspace}
\newcommand{\rtwo}{(\emph{ii})\xspace}
\newcommand{\rthree}{(\emph{iii})\xspace}
\newcommand{\rfour}{(\emph{iv})\xspace}

\newcommand{\typeone}{Type I}
\newcommand{\typetwo}{Type II}

\let\oldnl\nl
\newcommand{\nonl}{\renewcommand{\nl}{\let\nl\oldnl}}

\newcommand{\Omit}[1]{}

\newcommand{\query}{\Psi}

\newcommand{\bestoverset}{\mathit{SCons}_{\lang}}
\newcommand{\bestunderset}{\mathit{WImpl}_{\lang}}

\newcommand{\querymod}{\query}
\newcommand{\queryrem}{\query_{rem}}
\newcommand{\queryremwpp}{\query_\wppre}
\newcommand{\lang}{\mathcal{L}}
\newcommand{\langover}{\mathcal{O}}
\newcommand{\langunder}{\mathcal{U}}
\newcommand{\langneg}{\overline{\lang}}
\newcommand{\langoverneg}{\overline{\mathcal{O}}}
\newcommand{\fls}{\bot}
\newcommand{\tru}{\top}

\newcommand{\lpropertyp}[1]{$#1$-property\xspace}
\newcommand{\loverpropertyp}[1]{$#1$-consequence\xspace}
\newcommand{\lunderpropertyp}[1]{$#1$-implicant\xspace}
\newcommand{\lproperty}{\lpropertyp{\lang}}
\newcommand{\loverproperty}{\loverpropertyp{\lang}}
\newcommand{\lunderproperty}{\lunderpropertyp{\lang}}
\newcommand{\lpropertiesp}[1]{$#1$-properties\xspace}
\newcommand{\loverpropertiesp}[1]{$#1$-consequences\xspace}
\newcommand{\lunderpropertiesp}[1]{$#1$-implicants\xspace}
\newcommand{\lproperties}{\lpropertiesp{\lang}}
\newcommand{\loverproperties}{\loverpropertiesp{\lang}}
\newcommand{\lunderproperties}{\lunderpropertiesp{\lang}}
\newcommand{\lconjunctionp}[1]{$#1$-conjunction\xspace}
\newcommand{\lconjunctionsp}[1]{$#1$-conjunctions\xspace}
\newcommand{\ldisjunctionp}[1]{$#1$-disjunction\xspace}
\newcommand{\ldisjunctionsp}[1]{$#1$-disjunctions\xspace}
\newcommand{\lconjunction}{\lconjunctionp{\lang}}
\newcommand{\lconjunctions}{\lconjunctionsp{\lang}}
\newcommand{\ldisjunction}{\ldisjunctionp{\lang}}
\newcommand{\ldisjunctions}{\ldisjunctionsp{\lang}}

\newcommand{\phiquery}{\exists \ecolor{h}.\, \phiprog(\ucolor{v}, \ecolor{h})}
\newcommand{\phiprog}{\psi}

\newcommand{\phiand}{\varphi_{\wedge}}
\newcommand{\phior}{\varphi_{\vee}}

\newcommand{\phiinit}{\varphi_{init}}
\newcommand{\philast}{\varphi_{last}}
\newcommand{\phil}{\varphi}

\newcommand{\hoaretriple}[3]{\{#1\}\,#2\,\{#3\}}

\newcommand{\incortriple}[3]{[#1]\,#2\,[#3]}

\newcommand{\wlpre}{\mathit{wlp}}
\newcommand{\wppre}{\mathit{wpp}}
\newcommand{\spost}{\mathit{spo}}
\newcommand{\wpost}{\mathit{wupo}}
\newcommand{\lwlp}{\lwlpparam{\lang}}
\newcommand{\lspo}{\lspoparam{\lang}}
\newcommand{\lwpo}{\lwpoparam{\lang}}
\newcommand{\lwpp}{\lwppparam{\lang}}
\newcommand{\lwlpparam}[1]{$#1$-weakest liberal precondition\xspace}
\newcommand{\lspoparam}[1]{$#1$-strongest postcondition\xspace}
\newcommand{\lwpoparam}[1]{$#1$-weakest under-approximate postcondition\xspace}
\newcommand{\lwppparam}[1]{$#1$-weakest possible precondition\xspace}

\newcommand{\wlprelang}{\wlprelangp{\lang}}
\newcommand{\wpprelang}{\wpprelangp{\lang}}
\newcommand{\spostlang}{\spostlangp{\lang}}
\newcommand{\wpostlang}{\wpostlangp{\lang}}
\newcommand{\wprelangp}{\wprelangp{\lang}}
\newcommand{\wlprelangp}[1]{\wlpre_{#1}}
\newcommand{\wpprelangp}[1]{\wppre_{#1}}
\newcommand{\spostlangp}[1]{\spost_{#1}}
\newcommand{\wpostlangp}[1]{\wpost_{#1}}
\newcommand{\progstate}{\sigma}

\newcommand{\ex}{e}
\newcommand{\posex}{\ex^+}
\newcommand{\negex}{\ex^-}
\newcommand{\eplus}{E^+}
\newcommand{\eminus}{E^-}

\newcommand{\eminusmay}{E^-_{may}}
\newcommand{\eplusmay}{E^+_{may}}

\newcommand{\overapprox}[1]{\textsc{IsCons}_{#1}\xspace}
\newcommand{\underapprox}[1]{\textsc{IsImpl}_{#1}\xspace}
\newcommand{\checkimpl}{\textnormal{\textsc{CheckImplication}}\xspace}
\newcommand{\cs}{\textnormal{{\textsc{CheckSoundness}}}\xspace}

\newcommand{\cp}{\textnormal{{\textsc{CheckPrecision}}}\xspace}
\newcommand{\cpover}{\textnormal{{\textsc{CheckStrongest}}}\xspace}
\newcommand{\cpunder}{\textnormal{{\textsc{CheckWeakest}}}\xspace}
\newcommand{\synth}{\textnormal{{\textsc{Synthesize}}}\xspace}

\newcommand{\synthproperty}{\textnormal{{\textsc{SynthStrongestConsequence}}}\xspace}
\newcommand{\synthoverproperty}{\textnormal{{\textsc{SynthStrongestConsequence}}}\xspace}
\newcommand{\synthunderproperty}{\textnormal{{\textsc{SynthWeakestImplicant}}}\xspace}
\newcommand{\synthoverproperties}{\textnormal{{\textsc{SynthStrongestConjunction}}}\xspace}
\newcommand{\synthunderproperties}{\textnormal{{\textsc{SynthWeakestDisjunction}}}\xspace}

\newcommand{\interp}[1]{\llbracket #1 \rrbracket}

\newcommand{\poscolor}[1]{\hlg{#1}}
\newcommand{\negcolor}[1]{\hlr{#1}}
\newcommand{\diffcolor}[1]{\textcolor{violet}{#1}}
\newcommand{\keyword}[1]{{\textbf{\texttt{#1}}}}

\newcommand{\cegqi}{\textnormal{\textsc{CEGQI}}\xspace}
\newcommand{\genex}{\textnormal{\textsc{GenCandidateNegEx}}\xspace}
\newcommand{\checkex}{\textnormal{\textsc{CheckCandidateNegEx}}\xspace}
\newcommand{\genspec}{\textnormal{\textsc{GenCandidateSpecNegEx}}\xspace}

\newcommand{\framework}{\textsc{loud}\xspace}
\newcommand{\name}{\textsc{aspire}\xspace}
\newcommand{\spyro}{\textsc{spyro}\xspace}
\newcommand{\sketch}{\textsc{sketch}\xspace}

\newcommand{\sygus}{\textsc{SyGuS}\xspace}

\newcommand{\synquid}{\textsc{Synquid}\xspace}

\newcommand{\exname}[1]{\texttt{#1}\xspace}

\newcommand{\x}{x}
\newcommand{\y}{y}

\newcommand{\vara}{a}

\newcommand{\varl}{l_{in}}

\newcommand{\varmodulus}{M}

\newcommand{\lout}{l_{out}}

\newcommand{\modhash}{\exname{modhash}}
\newcommand{\modhasharg}[3]{#1#3\bmod #2}
\newcommand{\remhash}{\exname{remhash}}

\newcommand{\isprime}{\exname{isPrime}}
\newcommand{\remainderop}{\exname{\%}}

\newcommand{\sort}{\exname{sort}\xspace}

\newcommand{\maxtwo}{\exname{max2}}
\newcommand{\maxthree}{\exname{max3}}
\newcommand{\maxfour}{\exname{max4}}
\newcommand{\maxfive}{\exname{max5}}

\newcommand{\diff}{\exname{diff}}

\newcommand{\arraytwo}{\exname{arrSearch2}}
\newcommand{\arraythree}{\exname{arrSearch3}}

\newcommand{\lappend}{\exname{append}}

\newcommand{\ldelete}{\exname{deleteFirst}}
\newcommand{\ldeleteall}{\exname{delete}}
\newcommand{\ldrop}{\exname{drop}}

\newcommand{\lreverse}{\exname{reverse}\xspace}

\newcommand{\qenqueue}{\exname{enqueue}}

\newcommand{\iaabs}[1]{\exname{abs{#1}}}
\newcommand{\ialinsum}[1]{\exname{linSum{#1}}}
\newcommand{\ianonlinsum}[1]{\exname{nonlinSum{#1}}}

\newcommand{\size}{\exname{len}}
\newcommand{\forallm}{\forall}
\newcommand{\existsm}{\exists}

\newcommand{\nondisj}{D}

\newcommand{\nonap}{AP}
\newcommand{\nonint}{I}

\definecolor{lgreen}{RGB}{248,255,248}
\definecolor{dgreen}{RGB}{0,128,0}
\definecolor{titlecol}{RGB}{170,199,250}
\definecolor{dred}{RGB}{200,0,0}
\definecolor{lyellow}{RGB}{255,255,40}
\definecolor{plum}{rgb}{0.56, 0.27, 0.52}
\definecolor{dblue}{RGB}{0, 0, 200}
\definecolor{bviolet}{RGB}{71, 57, 146}
\definecolor{dyellow}{RGB}{253,180,00}

\newcommand{\hlg}[1]{\colorbox{lime}{$#1$}}
\newcommand{\hlr}[1]{\colorbox{pink}{$#1$}}
\newcommand{\hln}[1]{\colorbox{gray!20}{$#1$}}

\newcommand{\ecolor}[1]{\textcolor{dblue}{#1}}
\newcommand{\ucolor}[1]{\textcolor{dred}{#1}}

\usepackage{tikz}
\usetikzlibrary{shapes,snakes}
\usepackage[framemethod=tikz]{mdframed}
\usepackage{pgfplots}

\newcommand{\inlinef}[1]{``#1''}

\newcommand{\isPrime}{\exname{isPrime}}

\newcommand{\shuff}[1]{\exname{shuffle{#1}}}
\newcommand{\wppBubble}[1]{\exname{bubble{#1}}}
\newcommand{\wppSwap}[1]{\exname{swap{#1}}}
\newcommand{\wppHash}[1]{\exname{hashcoll}{#1}}
\newcommand{\coin}{\exname{coin}}

\newcommand{\rsum}{\exname{rsum}}
\newcommand{\rsquaresum}{\exname{rsqsum}}
\newcommand{\rcubicsum}{\exname{rcubsum}}
\newcommand{\mergesort}{\exname{mergesort}}
\newcommand{\jain}[1]{\exname{jain{#1}}}

\newcommand{\race}[1]{\exname{race{#1}}}
\newcommand{\resource}[1]{\exname{rsrc{#1}}}
\newcommand{\philosopher}{\exname{philo}}
\newcommand{\obdet}{\exname{obdet}}

\newcommand{\remwpp}{\exname{remwpp}}
\newcommand{\remwupo}{\exname{remwupo}}
\newcommand{\incwpp}[1]{\exname{inc{#1}wpp}}
\newcommand{\incwupo}[1]{\exname{inc{#1}wupo}}
\newcommand{\arith}[1]{\exname{arit{#1}}}

\newcommand{\rggame}{\exname{rg}}
\newcommand{\numgame}[1]{\exname{num{#1}}}
\newcommand{\decision}{\exname{decision}\xspace}
\newcommand{\playgame}{\text{play}\xspace}
\newcommand{\playerA}{player-1\xspace}
\newcommand{\playerB}{player-2\xspace}
\newcommand{\gtrue}{T}
\newcommand{\gfalse}{F}
\newcommand{\strategyA}{\ucolor{\alpha}}
\newcommand{\strategyB}{\ecolor{\beta}}
\newcommand{\nimgame}[1]{\exname{nim{#1}}}
\newcommand{\tempgame}{\exname{temp}}

\newcommand{\atomcons}[1]{\ucolor{ac_{#1}}}
\newcommand{\atomp}{\textit{atom}}

\newcommand{\numSpec}{35\xspace}

\newcommand{\numSpecSyGuS}{7\xspace}
\newcommand{\numSpecImp}{6\xspace}

\newcommand{\numNondConc}{8\xspace}

\newcommand{\numILPaperWPP}{3\xspace}

\newcommand{\numConcProb}{3\xspace}

\usepackage[nosort]{cleveref}
\newcommand{\cachespd}{1.4x\xspace}

\usepackage{csquotes}
\usepackage{framed}
\usepackage{arydshln}
\usepackage{afterpage}
\usepackage{enumitem}

\newenvironment{mybox}[1][gray!20]{
	\begin{tcolorbox}[   
		breakable,
		left=0pt,
		right=0pt,
		top=0pt,
		bottom=-1pt,
		colback=#1,
		colframe=#1,
		width=\dimexpr\textwidth\relax,
		boxsep=2pt,
		arc=0pt,outer arc=0pt,
		]
	}{
\end{tcolorbox}
}

\newtheorem{example}{Example}[section]
\newtheorem{definition}{Definition}[section]
\setlist[itemize]{align=parleft,left=0pt..1em, topsep=2pt}
\setlist[description]{topsep=2pt}

\begin{document}

\title{\framework: Synthesizing Strongest and Weakest Specifications}

\author{Kanghee Park}
\authornote{Kanghee Park and Xuanyu Peng contributed equally to this work.}
\orcid{0009-0005-7983-233X}
\affiliation{%
  \institution{University of California San Diego}
  \country{USA}
}
\email{kap022@ucsd.edu}

\author{Xuanyu Peng}
\authornotemark[1]
\orcid{0000-0001-8613-3506}
\affiliation{%
  \institution{University of California San Diego}
  \country{USA}
}
\email{xup002@ucsd.edu}

\author{Loris D’Antoni}
\orcid{0000-0001-9625-4037}
\affiliation{%
  \institution{University of California San Diego}
  \country{USA}
}
\email{ldantoni@ucsd.edu}

\input{abstract.tex}

\ccsdesc[500]{Theory of computation~Program specifications}
\ccsdesc[500]{Software and its engineering~Automated static analysis}
\ccsdesc[500]{Theory of computation~Abstraction}
\ccsdesc[500]{Software and its engineering~Automatic programming}

\keywords{Program Specifications, Program Synthesis}

\maketitle

\input{1introduction.tex}
\input{2motivating_example.tex}
\input{3framework.tex}
\input{4algorithm.tex}

\input{5cegqi.tex}
\input{6implementation.tex}

\input{7evaluation.tex}
\input{8related_work.tex}
\input{9conclusion.tex}

\bibliographystyle{ACM-Reference-Format}
\bibliography{main.bib}

\newpage
\appendix
\input{11appendix-hoare}
\input{12appendix-alg}
\input{13appendix-proof}

\input{10appendix-eval}

\end{document}

%% file: abstract.tex
\begin{abstract}
This paper tackles the problem of synthesizing specifications for nondeterministic programs. 
For such programs, useful specifications can capture demonic properties, which hold for \emph{every} nondeterministic execution, but also
angelic properties, which hold for \emph{some} nondeterministic execution.
We build on top of a recently proposed framework by \citet{park2023specification} in which given 
\rone a \textit{quantifier-free} query $\query$ posed about a set of function definitions (i.e., the behavior for which we want to generate a specification), and 
\rtwo a language $\lang$ in which each extracted property is to be expressed (we call properties in the language \lproperties), the goal is to synthesize a conjunction $\bigwedge_i \phil_i$ of \lproperties such that each of the $\phil_i$ is a \emph{strongest \loverproperty} for $\query$:
$\phil_i$ is an over-approximation of $\query$ and there is no other \lproperty that over-approximates $\query$ and is strictly more precise than $\phil_i$.
This framework does not apply to nondeterministic programs for two reasons: it does not support existential quantifiers in queries (which are necessary to expressing nondeterminism) and it can only compute \loverproperties, i.e., it is unsuitable for capturing both angelic and demonic properties.

This paper addresses these two limitations and presents a framework, \framework, for synthesizing both \emph{strongest \loverproperties} and \emph{weakest \lunderproperties} (i.e., under-approximations of the query $\query$) for queries that can involve \textit{existential quantifiers}.
We devise algorithms for handling the quantifiers appearing in \framework queries and implement them in a solver, \name, for problems expressed in \framework which can be used to describe and identify sources of bugs in both deterministic and nondeterministic programs,
extract properties from concurrent programs,
and synthesize winning strategies in two-player games.
\end{abstract}

%% file: 1introduction.tex
\section{Introduction}
\label{Se:Introduction}

%
%
Specifications allow us to understand what programs do, but are often hard to write and maintain.
Writing specifications is especially hard for programs involving nondeterminism,
a construct necessary to model many practical applications such as concurrency, random execution, and games.
Part of what makes writing such specifications hard is that specifications for programs involving nondeterminism might capture two types of properties:
\textit{demonic properties}, which hold for \emph{every} nondeterministic execution,
and \textit{angelic properties}, which hold for \emph{some} nondeterministic execution.

Consider a program that nondeterministically shuffles the elements of a list.
A possible valid demonic property could state that the output list is always a permutation of the input list.
However, \citet{edsko2011reverse} argued that this specification alone is ``incomplete'';
a sorting function would also meet this specification, even though it would not behave correctly as a shuffle function.
A better specification would include an angelic property that states that the shuffle function can in fact produce all permutations of the input list---i.e., for any permutation of the input list, there exists \emph{some} nondeterministic execution that can generate it.

The goal of this paper is to devise a unified logical framework for synthesizing provably sound angelic and demonic specifications for nondeterministic programs.
Most approaches that automatically generate specifications from code~\cite{DBLP:journals/tse/ErnstCGN01,DBLP:journals/scp/ErnstPGMPTX07,smallbone2017quick,DBLP:journals/pacmpl/AstorgaSDWMX21} rely on a finite set of dynamically generated input executions and cannot guarantee soundness even when the programs are deterministic.
To our knowledge, \spyro~\cite{park2023specification} is the only framework for synthesizing specifications that are provably sound, but it is fundamentally limited to deterministic programs.
In this paper, we redesign the \spyro framework to support nondeterminism, and both angelic and demonic reasoning.

\paragraph{Limitations of \spyro}
In \spyro, the problem of synthesizing a specification is phrased as follows:
Given \rone a \textit{quantifier-free} query $\query$ posed about a set of function definitions, and \rtwo a domain-specific language $\lang$ in which each extracted property is to be expressed (we call properties in the language \lproperties),
the goal is to synthesize a conjunction $\bigwedge_i \phil_i$ of
logically incomparable \lproperties such that each $\phil_i$ is an \emph{over-approximation} of $\query$ and is a strongest in $\lang$---i.e., there is no other \lproperty that over-approximates $\query$ and that strictly implies $\phil_i$.
For example, for a query $\query := (\lout = \lreverse(\varl))$ describing a list-reverse function, and a language $\lang$ of arithmetic formulas over variables and their lengths,
the property $\phil := \size(\lout) \leq \size(\varl)$ is an \loverproperty of $\query$, but not a strongest one,
because $\phil_1 := \size(\lout) = \size(\varl)$ is a stronger \loverproperty of $\query$.

Because the query $\Psi$ and DSL $\lang$ can be provided by a user, the \spyro framework can be applied to many domains.
By setting $\lang$ to capture the syntax of restricted refinement types, \spyro has been used for extracting refinement types from data-structure transformations~\cite{polikarpova2016program}, whereas by setting $\lang$ to capture algebraic specification, \spyro could synthesize interfaces for software modules~\cite{DBLP:conf/fmcad/ParkJDR23}.

\spyro is very expressive, but it does not support nondeterministic programs for two reasons.
First, \spyro's synthesis algorithms are fundamentally limited to \textit{quantifier-free queries}.
Without existential quantifiers, \spyro's queries cannot model nondeterministic programs, concurrent programs, or uncertainty.
Second, \spyro is \textit{limited to synthesizing over-approximations} of the program behavior.
While over-approximations can capture what must happen for \emph{every} (nondeterministic) execution (the demonic properties), reasoning about the actual behaviors that \emph{some} (nondeterministic) execution can exhibit (the angelic properties) requires under-approximated specifications.

\paragraph{The \framework Framework}
This paper addresses the two limitations of \spyro and presents \framework, a general framework for solving the following problem:
\begin{mybox}
  Given an \textbf{existentially quantified query} $\query$ posed about a set of function definitions and a language $\lang$, find \rone a \textbf{strongest conjunction} of \loverproperties that is implied by $\query$, and
  \rtwo a \textbf{weakest disjunction} of \lunderproperties that implies $\query$.
\end{mybox}

The \framework framework is our key contribution. While \framework superficially looks like a small modification of \spyro~\cite{park2023specification}, the minimal extension of including existential quantifications in queries allows \framework to elegantly capture many new complex scenarios into a single unified logical framework.
Specifically, existentially quantified queries allow \framework to reason about both angelic and demonic properties of programs involving nondeterminism.
For example, consider the dining-philosophers problem where $n$ philosophers are arranged in a circle, each concurrently (and nondeterministically) acquiring and releasing contended resources placed on either of their sides.
We can model what combinations of actions and scheduling lead to a deadlock using an existentially quantified query such
as $\exists \ecolor{s}.\, \ucolor{dl} = \philosopher(\ecolor{s}, \ucolor{p_1}, \ldots, \ucolor{p_n})$, where $\ucolor{p_i} \in \{L, R\}$ indicates which resource the philosopher $\ucolor{p_i}$ tries to take first and
$\ucolor{dl}$ denotes that a deadlock has happened; $\ecolor{s}$ is the nondeterministic sequence of order in which threads are scheduled (detailed description in \Cref{se:eval-concurrency}).

Supporting both over- and under-approximate reasoning (i.e., computing both \loverproperties and \lunderproperties) enables new applications.
Let's say we are interested in understanding what philosophers' actions \textit{may} lead to a deadlock for \textit{some} possible schedule.
When given an appropriate language $\lang$, a possible under-approximation (i.e., \lunderproperty) of the query $\exists \ecolor{s}.\, \ucolor{dl} = \philosopher(\ecolor{s}, \ucolor{p_1}, \ldots, \ucolor{p_n})$ is $\ucolor{dl} \wedge \ucolor{p_1} = \cdots = \ucolor{p_n}$, which states that deadlock can happen when all the philosophers prefer the same direction.
On the other hand, if we instead are interested in what philosophers' choices  \textit{must} prevent deadlocks for \textit{any} possible schedule, we can resort to over-approximation. A possible over-approximation (i.e. \loverproperty) in the \framework framework is $\ucolor{p_1} \neq \ucolor{p_2}  \Rightarrow \neg \ucolor{dl}$, which states that deadlock will not happen when processes $\ucolor{p_1}$ and $\ucolor{p_2}$ disagree on their fork choice.
The above two example properties show that, for nondeterministic programs, consequences hold for \textit{every possible} nondeterministic choice (the demonic perspective), whereas implicants hold for \textit{at least one} nondeterministic choice (the angelic perspective).

Thanks to its generality,
\framework can also capture reasoning capabilities of Hoare logic~\cite{hoare69axiomatic} (e.g., computing weakest liberal precondition and strongest postcondition) and incorrectness logic~\cite{Peter2019Incorrectness,edsko2011reverse} (e.g., computing weakest possible precondition and weakest under-approximate postcondition).

\paragraph{New Synthesis Algorithms in \framework}

Existentially quantified queries and the ability to synthesize both over- and under-approximations make the \framework framework more expressive than \spyro, but also make synthesis more challenging, thus requiring new algorithmic insights.


\citet{park2023specification} presented a counterexample-guided synthesis (CEGIS) algorithm for solving problems in the \spyro framework.
%
The algorithm accumulates positive and negative examples of possible program behaviors with respect to the given query and synthesizes \lproperties consistent with them.
A primitive called \cs checks if a candidate property is indeed sound and, if not, it produces a new positive example that the property fails to accept.
%
To ensure \loverproperties are strongest, a primitive \cp checks if the current \loverproperty is strongest;
if it is not, \cp returns a new \lproperty that accepts all positive examples, rejects all negative examples, and rejects one more negative example (which is also returned).
By alternating calls to these primitives, the algorithm eventually finds a strongest \loverproperty.

Key contributions and innovations in how we algorithmically solve \framework problems include
\rone generalized \cs and \cp primitives so that each operation has a dual form that can be used to synthesize both \loverproperties and \lunderproperties, and
\rtwo how we implement these primitives in the presence of existential quantifiers.
%
Specifically, proving that an \loverproperty is strongest and proving that an \lunderproperty is sound require solving a constraint with quantifier alternations of the form $\exists \ucolor{\ex}.\, \forall \ecolor{h}.\, \neg \phiprog(\ucolor{\ex}, \ecolor{h}) \land \phil(\ucolor{\ex})$.
To perform this check, we integrate a counterexample-guided quantifier instantiation algorithm (CEGQI) that operates in tandem with the overall CEGIS algorithm.
The CEGIS algorithm accumulates examples that approximate the behavior of the query, while the CEGQI algorithm accumulates instances of the quantified variable $\ecolor{h}$ that show if an example is positive or negative.
To our knowledge, our algorithm is the first one to combine CEGIS and CEGQI to deal with multiple nested quantifiers.

We implement a tool, called \name, to solve the synthesis problems in the \framework framework. \name can describe and identify sources of bugs in both deterministic and nondeterministic programs,
extract properties from concurrent programs, and synthesize winning strategies in two-player games. Because \name is built on the top of the \sketch program synthesizer~\cite{DBLP:journals/sttt/Solar-Lezama13}, it is only sound for programs in which inputs, recursion, and loops are bounded.
In the future, this limitation can be lifted by considering more general (though less efficient) program synthesizers~\cite{MessySemGuSTool-2022}.

\mypar{Contributions} Our work makes the following contributions:
\begin{itemize}
  \item
        A unified logical framework, \framework, for the problem of synthesizing strongest \loverproperties and weakest \lunderproperties for existentially quantified queries (\S\ref{se:framework}).
  \item
        Algorithms for solving \framework problems using four simple well-defined primitives: \synth, \checkimpl, \cpover and \cpunder (\S\ref{se:algorithm}).
  \item
        An algorithm that combines CEGQI and CEGIS to efficiently implement the primitives \checkimpl and \cpover for existentially quantified queries (\S\ref{se:cegqi}).
  \item
        A tool that implements our framework, called \name (\S\ref{se:implementation}).
  \item
        Multiple instantiations of \framework, showing its capability across a wide range of applications, e.g., reasoning about nondeterministic/concurrent programs and synthesizing game strategies (\S\ref{se:evaluation}).
\end{itemize}
\S\ref{se:related-work} discusses related work.
\S\ref{se:conclusion} concludes.
In the appendix,
\S\ref{app:relation-to-program-logics} relates \framework to program logics;
\S\ref{app:alg-allproperties} contains further details about algorithms;
\S\ref{app:proof} contains proofs;
and \S\ref{App:eval} contains further evaluation details.

%% file: 2motivating_example.tex
\section{Motivating Examples}
\label{se:motivating-example}

In this section, we illustrate how the \framework framework can be used to synthesize useful
over-approximated (\S\ref{se:example-reachability-over}) and under-approximated (\S\ref{se:example-reachability-under}) 
properties of programs.

\input{code-figure/group-2-1-1}

Consider the parametric hash function shown in \Cref{fig:modhash}, where $\x$ is an integer input and $\vara$ and $\varmodulus > 0$ are possible parameters---i.e., $\vara$ and $\varmodulus$ are fixed in a specific implementation of $\modhash$. 
Intuitively, \modhash~can be viewed as a family of hash functions where the variable $\x$ is the input.

\input{code-figure/modhash}
In \framework, to reason about the behavior of a program, one provides a logical query they are interested in over- or under-approximating with properties in a given language.
For example, one may provide the query $\querymod$ in \Cref{fig:modhash-query},
which is equivalent to the existentially quantified formula 
$\existsm \ecolor{\x}.\, \ucolor{\y} = \ucolor{\vara}\ecolor{\x} \bmod{\ucolor{\varmodulus}}$.
In this example, our goal is to identify which choices of parameters $\ucolor{\vara}$ and $\ucolor{\varmodulus}$ will make \modhash~surjective onto $\mathbb{Z}_{\ucolor{\varmodulus}} = \{0, 1, \ldots, \ucolor{\varmodulus} - 1\}$. Specifically, we seek properties that capture the relationship between the output variable $\ucolor{\y}$ and the parameters $\ucolor{\vara}$ and $\ucolor{\varmodulus}$ (colored in red). 
To do so, we treat the input $\ecolor{\x}$ as a nondeterministic parameter (existentially quantified and thus colored in blue).

\subsection{Over-approximate Reasoning} 
\label{se:example-reachability-over}

We start with properties that are consequences (over-approximations) of the query in \Cref{fig:modhash-query}. 
That is, we want formulae $\phil(\ucolor{\y}, \ucolor{\vara}, \ucolor{\varmodulus})$ such that
$\forallm \ucolor{\y}, \ucolor{\vara}, \ucolor{\varmodulus}.\, 
(\existsm \ecolor{\x}.\, \ucolor{\y} = \ucolor{\vara}\ecolor{\x} \bmod{\ucolor{\varmodulus}})
\Rightarrow \phil(\ucolor{\y}, \ucolor{\vara}, \ucolor{\varmodulus})$.

As argued by \citet{park2023specification}, different applications---e.g., generating type judgments or generating algebraic specifications---require formulae to adhere to a specific syntactic fragment.
Thus in our framework, users are in charge of providing a DSL $\lang$ as an input to express properties they are interested in; we call properties expressible in this DSL \lproperties. 
We note that the DSL $\lang$ might contain user-given functions for which an implementation should be provided.
Furthermore, we say an \lproperty $\phil$ is an \loverproperty if it is a consequence of the given query formula.
One goal of our framework is to synthesize a set of incomparable \emph{strongest} \loverproperties.

In our example, the user provides the DSL $\langover$ shown in \Cref{fig:hash-over-grammar}, which includes the constant $0$, all free variables and comparison operations between them, and also the user-defined function \isprime~(together with its implementation) that is potentially related to the problem.
An incomparable set of strongest \loverpropertiesp{\langover} for the query $\querymod$ is shown in \Cref{fig:modhash-over-prop}.
Since conjoining two consequences of $\query$ results in a stronger consequence of $\query$, 
we interpret the set of properties as their conjunction and thus call the set an \lconjunctionp{\langover}.
The properties in \Cref{fig:modhash-over-prop} give us insights into the behavior of the function \modhash---e.g., that setting the value of $\ucolor{\vara}$ to be equal to $\ucolor{\varmodulus}$ or $-\ucolor{\varmodulus}$ is probably not a good idea as it would result in a function that always returns $\ucolor{\y}=0$.
For our discussion, we focus our attention on the first two properties, which imply that the output $\ucolor{\y}$ falls within $0 \leq \ucolor{\y} < \ucolor{\varmodulus}$.
A well-designed hash function with a set $S$ as range should be surjective onto the set $S$, 
meaning that for every value $v$ in $S$, there should be inputs that yield $v$ as output. 
However, because the formulae in \Cref{fig:modhash-over-prop} are over-approximations, we are not guaranteed that all the values in $0 \leq \ucolor{\y} < \ucolor{\varmodulus}$ are indeed possible outputs of \modhash.

\subsection{Under-approximate Reasoning}
\label{se:example-reachability-under}

Over-approximation alone cannot capture whether a specific program behavior \emph{can} occur---i.e., is reachable.
For a formula $\phil(\ucolor{\y}, \ucolor{\vara}, \ucolor{\varmodulus})$ to define a reachability condition (i.e., a behavior that \textit{must} happen) of $\ucolor{\y}$ from some input $\ecolor{\x}$, the formula $\phil(\ucolor{\y}, \ucolor{\vara}, \ucolor{\varmodulus})$ must be an implicant (under-approximation) of the query $\querymod$, which formally can be stated as follows:
$\forallm \ucolor{\y}, \ucolor{\vara}, \ucolor{\varmodulus}.\, \phil(\ucolor{\y}, \ucolor{\vara}, \ucolor{\varmodulus}) \Rightarrow \exists \ecolor{\x}.\, \ucolor{\y} = \ucolor{\vara}\ecolor{\x} \bmod{\ucolor{\varmodulus}}$

We say an \lproperty $\phil$ is an \lunderproperty if it is an implicant of the given query formula.
Another goal of our framework is to synthesize a set of incomparable \emph{weakest} \lunderproperties.

In our example, the user provides the DSL $\langunder$ using the rules for \texttt{\nonap} and \texttt{\nonint} as shown in \Cref{fig:hash-over-grammar}, but replaces disjunction rules (nonterminal \texttt{\nondisj}) with the conjunction rules \texttt{C -> /\textbackslash[AP, 0..6]}.
Throughout the paper we will use $\langover$ to denote the language from \Cref{fig:hash-over-grammar} in examples involving over-approximations, and $\langunder$ for the language described here in examples involving under-approximation.
For instance, $\ucolor{\vara} = 0 \land \ucolor{\y} = 0$ is a \lunderpropertyp{\langunder} for query $\querymod$, but not a weakest one, as it strictly implies a \lunderpropertyp{\langunder} $\ucolor{\y} = 0$.

\input{code-figure/property-under}
A mutually incomparable set of weakest \lunderpropertiesp{\langunder} for the query $\querymod$ is shown in \Cref{fig:modhash-under-prop}. 
We interpret the set of properties as their disjunction and refer to the set as a \ldisjunctionp{\langunder}.
Each formula provides a sufficient condition for reachability of the output $\ucolor{\y}$---that is, if $\ucolor{\y}$, $\ucolor{\vara}$ and $\ucolor{\varmodulus}$ satisfy any formula in \Cref{fig:modhash-under-prop}, 
then there exists an input $\ecolor{\x}$ such that $\ucolor{\y} = \ucolor{\vara}\ecolor{\x} \bmod{\ucolor{\varmodulus}}$.
Crucially, the last formula provides a sufficient condition for \modhash~to be surjective onto $\mathbb{Z}_{\ucolor{\varmodulus}} = \{0, 1, \ldots, \ucolor{\varmodulus} - 1\}$---i.e., for a prime value of $\ucolor{\varmodulus}$ and non-zero value of $\ucolor{\vara}$ selected from the range $- \ucolor{\varmodulus} < \ucolor{\vara} < \ucolor{\varmodulus}$, 
all values of $\ucolor{\y}$ in $\mathbb{Z}_{\ucolor{\varmodulus}}$ are attainable from some choice of input $\ecolor{\x}$.

While our primary motivation is to model nondeterminism, 
the generality of our framework enables a variety of other applications. 
In particular, existential quantifiers can be used to model both forward and backward reasoning in the style of Hoare and incorrectness program logics for both deterministic and nondeterministic programs. 
For example, given a program and a predicate over its output, one can synthesize preconditions on what input may produce (over-approximate) or must produce (under-approximate) an output that satisfies the given predicate.
%

A detailed discussion of various applications of our framework is provided in \Cref{se:evaluation} and \Cref{app:relation-to-program-logics}.

%% file: code-figure/group-2-1-1.tex
\begin{figure}[t!]    
    \centering
    \begin{subfigure}{0.28\textwidth}
        \centering
        \input{code-figure/query}
        \vspace{-10pt}
        \caption{Query $\querymod$}
        \label{fig:modhash-query}
    \end{subfigure}
    \hspace{1em}\vrule\hspace{0.5em}
    \begin{subfigure}{0.31\textwidth}
        \centering
        \input{code-figure/dsl-over}
        \vspace{-10pt}
        \caption{DSL $\langover$}
        \label{fig:hash-over-grammar}
    \end{subfigure}
    \hspace{1em}\vrule\hspace{0.5em}
    \begin{subfigure}{0.28\textwidth}
        \centering
        \input{code-figure/property-over}
        \caption{\loverpropertiesp{\langover}}
        \label{fig:modhash-over-prop}
    \end{subfigure}
    \vspace{-2mm}
    \caption{
    (a) A query $\querymod$ for identifying properties of \modhash that hold for any choice of input $\ecolor{\x}$. 
    Users declare variables and label existentially quantified ones with the keyword \texttt{\textbf{exist}}.
    (b) A DSL $\langover$ for over-approximations. \texttt{\textbackslash/[AP, 0..6]} is a shorthand for the disjunction of 0 to 6 atomic propositions. 
    (c) \loverpropertiesp{\langover} by our tool \name when given the query $\querymod$ and the DSL $\langover$.
    We write $p \Rightarrow q$ instead of $\neg p \vee q$ for readability.}
\end{figure}

%% file: code-figure/query.tex
\begin{lstlisting}[ tabsize=3, 
    basicstyle= \tt \footnotesize, 
    keywordstyle=\color{black}\bfseries, 
    commentstyle=\color{gray}, 
    xleftmargin=-2.5em, 
    escapeinside=``,
    language = C,
    morekeywords = {Variables, Query, exist },
    emph={a, M, y}, emphstyle=\color{dred},  
    emph={[2]x}, emphstyle={[2]\color{dblue}}, 
    numbers = none
    ]
    Variables { 
        int a, M, y; 
        exist int x; 
    }
    Query { 
        y = modhash(a, M, x);
    }
\end{lstlisting}

%% file: code-figure/dsl-over.tex
    \begin{lstlisting}[ tabsize=3, 
    basicstyle= \tt \footnotesize, 
    keywordstyle=\color{black}\bfseries, 
    commentstyle=\color{gray}, 
    xleftmargin=-2.5em, 
    escapeinside=``,
    language = C,
    morekeywords = {Language, bool},
    emph={a, M, y}, emphstyle=\color{dred},  
    emph={[2]x}, emphstyle={[2]\color{dblue}}, 
    numbers = none
    ]
    Language { 
        D -> \/[AP,0..6];
        AP -> I {<=|<|==|!=} I
            | isPrime(M) 
            | !isPrime(M)
        I -> 0 | a | y | M | -M
    }
\end{lstlisting}

%% file: code-figure/property-over.tex
\begin{lstlisting}[tabsize=3, 
    basicstyle= \tt \footnotesize, 
    keywordstyle=\color{black}\bfseries, 
    commentstyle=\color{gray}, 
    xleftmargin=-2.5em, 
    escapeinside=``,
    language = C,
    morekeywords = {Language, bool},
    emph={a, M, y}, emphstyle=\color{dred},  
    emph={[2]x}, emphstyle={[2]\color{dblue}}, 
    numbers = none,
    ]
    C1: 0 <= y
    C2: y < M
    C3: (a == 0) => (y == 0)  
    C4: (a == M) => (y == 0)
    C5: (a == -M) => (y == 0)
\end{lstlisting}

%% file: code-figure/modhash.tex
\begin{wrapfigure}{r}{0.33\textwidth}
\vspace{-4mm}
    \centering
    \begin{lstlisting}[ tabsize=3, 
    basicstyle= \tt \footnotesize, 
    keywordstyle=\color{black}\bfseries, 
    commentstyle=\color{gray}, 
    escapeinside=``,
    language = C,
    morekeywords = {mod},
    numbers = none
    ]
int modhash (int a, M, x) {
    return a * x mod M; }
\end{lstlisting}
\vspace{-4mm}
\caption{\texttt{modhash} function}
\label{fig:modhash}
\vspace{-4mm}    
\end{wrapfigure} 

%% file: code-figure/property-under.tex
\begin{wrapfigure}{r}{0.43\textwidth}
\vspace{-8mm}
    \centering
    \begin{lstlisting}[ tabsize=3, 
    basicstyle= \tt \footnotesize, 
    keywordstyle=\color{black}\bfseries, 
    commentstyle=\color{gray}, 
    xleftmargin=-2.5em, 
    escapeinside=``,
    language = C,
    morekeywords = {Language, bool},
    emph={a, M, y}, emphstyle=\color{dred},  
    emph={[2]x}, emphstyle={[2]\color{dblue}}, 
    numbers = none
    ]
    I1: y == 0
    I2: 0 <= a /\ a < M /\ a == y
    I3: 0 <= y /\ y < M /\ -M < a 
        /\ a < M /\ a != 0 /\ isPrime(M)
    
\end{lstlisting}
\vspace{-4mm}
\caption{Synthesized \lunderpropertiesp{\langunder}. }
\vspace{-4mm}
\label{fig:modhash-under-prop} 
\end{wrapfigure}

%% file: 3framework.tex
\section{Strongest Consequences and Weakest Implicants}
\label{se:framework}

In this section, we describe our framework, which extends the strongest \loverproperties synthesis framework of \citet{park2023specification} in two key ways: 
\rone allowing existentially quantified variables, and
\rtwo enabling the synthesis of both the strongest \loverproperties and weakest \lunderproperties.

We describe what inputs a user of the framework provides, and what they obtain as output.

\noindent
\textbf{Input 1: Query.} 
The \textit{query} $\query$ is a first-order formula of the form $\phiquery$, where $\phiprog$ is a quantifier-free formula $\phiprog$.
%
The inclusion of the existentially quantified variables $\ecolor{h}$ in the query is a key novelty of this paper: it enables many new applications such as reasoning about nondeterministic programs, and forward and backward reasoning in program logics (\Cref{se:eval-incorrectness}).

We use the symbol $\ecolor{h}$ (for hidden) to represent existentially quantified variables and the symbol $\ucolor{v}$ (for visible) to represent free variables. In practice, both $\ecolor{h}$ and $\ucolor{v}$ can be tuples and denote multiple variables.
In our motivating examples, queries are given in \Cref{fig:modhash-query}.


\noindent\textbf{Input 2: Grammar of \lproperties.}
The grammar of the DSL $\lang$ in which the synthesizer is to express properties for the query.
Each formula $\phil$ in the DSL $\lang$ is a predicate defined over the free variables $\ucolor{v}$ of the query $\phiquery$.
An example of a DSL is the language $\langover$ in \Cref{fig:hash-over-grammar}.

\noindent\textbf{Input 3: Semantics of the program and operators.}
Semantics of the function symbols that appear in query $\query$ (e.g., $\modhash$) and in the DSL $\lang$ (e.g., \isPrime).
In our implementation, semantic definitions are given as a program in the $\sketch$ language~\cite{DBLP:journals/sttt/Solar-Lezama13}, which automatically transforms them into first-order formulas.
For example, one may provide a C-style function written in \sketch that checks whether a number between $2$ and $\sqrt{n}$ is a divisor of $n$ as the semantic of \isprime.
We discuss in \Cref{se:implementation} how \sketch works and examine the limitations of this approach.

\noindent\textbf{Output: Strongest \loverproperties and weakest \lunderproperties.}
Our goal is to synthesize a set of incomparable \emph{strongest} \loverproperties and a set of incomparable \emph{weakest} \lunderproperties of query $\query$. Ideally, both strongest \loverproperty and weakest \lunderproperty would be the formula that is exactly equivalent to $\phiquery$, but in general, the DSL $\lang$ might not be expressive enough to do so. 
As argued by \citet{park2023specification}, this feature is actually a desired one as it allows for the application of our framework to various use cases, as demonstrated in \Cref{se:evaluation}.
Because in general there might not be an \loverproperty and an \lunderproperty that are equivalent to $\phiquery$, the goal becomes instead to find \lproperties that tightly approximate $\phiquery$.

We denote the set of models (over the free variables of $\query$) of a formula $\phil$ as $\interp{\phil}$.
For example in \Cref{se:example-reachability-over}, $\interp{\ucolor{\y} \geq 0}$ represents the set of models $\{ (\ucolor{\y}, \ucolor{\vara}, \ucolor{\varmodulus}) \mid \ucolor{\y} \geq 0 \}$. 
We say $\phil$ is stronger than $\phil'$ (or $\phil'$ is weaker than $\phil$) when $\interp{\phil} \subseteq \interp{\phil'}$, 
and $\phil$ is strictly stronger than $\phil'$ 
when $\interp{\phil} \subset \interp{\phil'}$.


%
\setlength{\columnseprule}{0.4pt}
\begin{multicols}{2}
\begin{definition}[A strongest \loverproperty]
\label{dfn:a-best-over}
An \lproperty $\phil$ is \emph{a strongest \loverproperty} for a query $\query$
if and only if 

\noindent \rone $\phil$ is a \emph{consequence} of the query $\query$: 
\[
\overapprox{\query}(\phil) := \forall \ucolor{v}.\, [\exists \ecolor{h}.\, \phiprog(\ucolor{v}, \ecolor{h}) \Rightarrow \phil(\ucolor{v})]
\]
\rtwo $\phil$ is \emph{strongest} with respect to $\query$ and $\lang$: 
\[
\neg\exists \phil' \in \lang.\, \overapprox{\query}(\phil') \land \interp{\phil'} \subset \interp{\phil}
\]
%
\end{definition}
%
\columnbreak
\begin{definition}[A weakest \lunderproperty]
\label{dfn:a-best-under}
An \lproperty $\phil$ is \emph{a weakest \lunderproperty} for a query $\query$
if and only if 

\noindent \rone $\phil$ is an \emph{implicant} of the query $\query$: 
\[
\underapprox{\query}(\phil) := \forall \ucolor{v}.\, [\phil(\ucolor{v}) \Rightarrow \exists \ecolor{h}.\, \phiprog(\ucolor{v}, \ecolor{h})]
\]
\rtwo $\phil$ is \emph{weakest} with respect to $\query$ and $\lang$: 
\[
\neg\exists \phil' \in \lang.\,
\underapprox{\query}(\phil') \land \interp{\phil'} \supset \interp{\phil}
\]
%
\end{definition}
\end{multicols}

Throughout the paper, we also use the term \emph{most-precise} to mean strongest for \loverproperties and weakest for \lunderproperties.
We use $\bestoverset(\query)$ and $\bestunderset(\query)$ to denote the set of all strongest \loverproperties and the set of all best \lunderproperties for $\query$, respectively.
Because $\lang$ may not be closed under conjunction (and disjunction), strongest \loverproperties (and weakest \lunderproperties) may not be semantically unique.
In \Cref{se:example-reachability-over}, both formulae $0 \leq \ucolor{\y}$ and $\ucolor{\y} < \ucolor{\varmodulus}$ are strongest \loverproperties of query $\querymod$, and neither implies the other.

The goal of our framework is to find a semantically strongest conjunction of incomparable strongest \loverproperties and a weakest disjunction of incomparable weakest \lunderproperties.

\begin{multicols}{2}
\begin{definition}[Best \lconjunction]
\label{dfn:best-conjunctions}
A potentially infinite set of \loverproperties $\Pi = \{\phil_i\}$ forms a \emph{best \lconjunction} $\phiand=\bigwedge_{i} \phil_i$ for query $\query$ if and only if

\noindent \rone
    each  $\phil_i\in \Pi$ is a strongest \loverproperty of $\query$;
    
\noindent \rtwo
    every distinct $\phil_i, \phil_j \in \Pi$ are \emph{incomparable}---i.e.,  $\interp{\phil_i}\setminus \interp{\phil_j}\neq \emptyset$ and $\interp{\phil_j}\setminus \interp{\phil_i}\neq \emptyset$;
    
\noindent \rthree
    the set is
    \emph{semantically minimal}---i.e.,
    for every strongest \loverproperty $\phil$ we have $\interp{\phiand}\subseteq \interp{\phil}$.
\end{definition}
%
\columnbreak
\begin{definition}[Best \ldisjunction]
\label{dfn:best-disjunctions}
A potentially infinite set of \lunderproperties $\Pi = \{\phil_i\}$ forms a \emph{best \ldisjunction} $\phior=\bigvee_{i} \phil_i$ for query $\query$ if and only if

\noindent \rone
    each  $\phil_i\in \Pi$ is a weakest \lunderproperty of $\query$;
    
\noindent \rtwo
    every distinct $\phil_i, \phil_j \in \Pi$ are \emph{incomparable}---i.e.,  $\interp{\phil_i}\setminus \interp{\phil_j}\neq \emptyset$ and $\interp{\phil_j}\setminus \interp{\phil_i}\neq \emptyset$; 
    
\noindent \rthree
    the set is
    \emph{semantically maximal}---i.e.,
    for every weakest \lunderproperty $\phil$ we have $\interp{\phior}\supseteq \interp{\phil}$.
\end{definition}
%
\end{multicols}

Best \lconjunctions and best \ldisjunctions are not necessarily unique, but they are all logically equivalent.
Specifically, a best \lconjunction is equivalent to the conjunction of all possible strongest \loverproperties, and a best \ldisjunction is equivalent to the disjunction of all possible weakest \lunderproperties.
Note that \textit{best} means more than semantic optimality because a strongest \lconjunction is not necessarily a best \lconjunction; predicates $\phil_1(\x) := \x \geq 0$ and $\phil_2(\x) := \x \geq 1$ could form a strongest \lconjunction, but it is not a best one because $\phil_1$ is strictly stronger than $\phil_2$--i.e., $\phil_1$ and $\phil_2$ are comparable.

\begin{restatable}[Semantic Optimality]{theorem}{bestset}
\label{thm:bestset}
If $\phiand$ is a best \lconjunction, then its interpretation coincides with the conjunction of all possible strongest \loverproperties:
$\interp{\phiand}=\interp{\bigwedge_{\phil\in \bestoverset(\query)} \phil}$.
If $\phior$ is a best \ldisjunction, then its interpretation coincides with the disjunction of all possible best \lunderproperties:
$\interp{\phior}=\interp{\bigvee_{\phil\in \bestunderset(\query)} \phil}$.
\end{restatable}

We are now ready to state our problem definition:
\begin{definition}[Best \lconjunction and \ldisjunction Synthesis]
  Given query $\query$,
  the concrete semantics for the function symbols in $\query$, and a domain-specific language $\lang$ with its corresponding semantic definition, synthesize a best \lconjunction and \ldisjunction for $\query$. 
\end{definition}

\paragraph{Practical remarks.}

The algorithm presented in \Cref{se:algorithm} computes finite \lconjunctions and \ldisjunctions, but it is possible that a solution to a given \framework problem instance requires an infinite number of conjuncts (or disjuncts).
Even in this case, because the algorithm incrementally computes incomparable strongest \loverproperties (or weakest \lunderproperties), one can stop it at any point and all computed \loverproperties (or \lunderproperties) will form a valid \lconjunction (or \ldisjunction), just not a best one.
The benchmarks considered in \Cref{se:evaluation} do not encounter this problem since they consider finite (though very large) languages.

Note that \Cref{dfn:best-conjunctions} and \ref{dfn:best-disjunctions} only guarantee that each \textit{pair} of synthesized properties is incomparable.
However, ensuring that each $\phil_i$ is incomparable to the entire set of properties $\Pi \setminus \{\phil_i\}$ only requires checking a single implication for each $\phil_i$, which can be trivially done in a postprocessing phase.
In practice, it is quite rare to encounter a redundant conjunct or disjunct. 

%% file: 4algorithm.tex
\section{Counterexample-Guided Inductive Specification Synthesis}
\label{se:algorithm}

In this section, we present algorithms for synthesizing best \lconjunctions and best \ldisjunctions.
We follow the example-guided approach proposed by \citet{park2023specification} that synthesizes strongest \loverproperties.
In tandem, we present a dual algorithm to synthesize weakest \lunderproperties, and we extend both algorithms to allow existentially quantified query formulas.

We first present the primitives necessary to instantiate the synthesis algorithms (\Cref{se:examples}).
Then, we present the algorithms for synthesizing a single \loverproperty or \lunderproperty that is incomparable to all the ones synthesized so far (\Cref{se:oneproperty}).
Finally, we present the algorithms for iteratively synthesizing the properties forming an \lconjunction or an \ldisjunction (\Cref{se:allproperties}).

\subsection{Synthesis from Positive and Negative Examples}
\label{se:examples}

The algorithms for synthesizing strongest \loverproperties and weakest \lunderproperties maintain two sets of examples: a set of positive examples $\eplus$, which should be accepted by the synthesized predicates, and negative examples $\eminus$, which should be rejected by the synthesized predicates.

\begin{definition}[Examples]
    Given a query $\query := \phiquery$ and a model $\ucolor{\ex}$ over the free variable $\ucolor{v}$ of query $\query$,
    we say that $\ucolor{\ex}$ is a \emph{positive example} if $\phiprog(\ucolor{e}, \ecolor{h})$ holds true for \emph{some} value of $\ecolor{h}$ (i.e., $\ucolor{\ex} \in \interp{\query}$) and a \emph{negative example} if $\phiprog(\ucolor{e}, \ecolor{h})$ does not hold for \emph{all} values of $\ecolor{h}$ (i.e., $\ucolor{\ex} \notin \interp{\query}$).
\end{definition}

%
%
\begin{example}[Positive and Negative Examples]
    \label{ex:pos-neg}
    Given the query $\query := \exists \ecolor{\x}.\, \ucolor{\y} = \modhasharg{\ucolor{\vara}}{\ucolor{\varmodulus}}{\ecolor{\x}}$,
    the model that assigns $\ucolor{\y}$ to the integer $1$, $\ucolor{\vara}$ to the integer $6$, and $\ucolor{\varmodulus}$ to the integer $5$ is a positive example, because the choice of value $\ecolor{x} = 1$ makes the equation $1 = 6 \bmod{5}$ holds true.
    For brevity, we represent such example as $\poscolor{(1, 6, 5)}$, where it denotes a valuation to the tuple $(\ucolor{\y}, \ucolor{\vara}, \ucolor{\varmodulus})$.
    The following examples are negative ones because no value of $\ecolor{\x}$ satisfies $\ucolor{\y} = \modhasharg{\ucolor{\vara}}{\ucolor{\varmodulus}}{\ecolor{\x}}$: $\negcolor{(-1, 1, 3)}$, $\negcolor{(3, 1, 3)}$, $\negcolor{(3, 2, 6)}$.
\end{example}

Because the DSL $\lang$ might not be expressive enough to capture the exact behavior of the query $\query$,
in general there is no predicate capable of accepting all the positive examples and rejecting all the negative examples.
Intuitively, a strongest \loverproperty must accept all positive examples while also excluding as many negative examples as possible.
\begin{example}[Examples and \loverproperties]
    \label{ex:pos-neg-over}
    Consider again the query $\query := \exists \ecolor{\x}.\, \ucolor{\y} = \modhasharg{\ucolor{\vara}}{\ucolor{\varmodulus}}{\ecolor{\x}}$ and the set of strongest \loverpropertiesp{\langover}
    $\{ 0 \leq \ucolor{\y}, \ucolor{\y} < \ucolor{\varmodulus},\,
        \ucolor{\vara} = 0 \Rightarrow \ucolor{\y} = 0,\,
        \ucolor{\vara} = \ucolor{\varmodulus} \Rightarrow \ucolor{\y} = 0,\,
        \ucolor{\vara} = -\ucolor{\varmodulus} \Rightarrow \ucolor{\y} = 0 \}$ from \Cref{fig:modhash-over-prop}.
    While a positive example $(\ucolor{\y}, \ucolor{\vara}, \ucolor{\varmodulus}) = \poscolor{(1, 6, 5)}$ is accepted by all \loverpropertiesp{\langover}, the negative example $\negcolor{(3, 2, 6)}$ is not rejected by any of them.
    In fact, the \loverpropertiesp{\langover} in \Cref{fig:modhash-over-prop} form a best \lconjunction, so $\negcolor{(3, 2, 6)}$ cannot be rejected by \emph{any} strongest \loverpropertyp{\langover}.
\end{example}

As illustrated by \Cref{ex:pos-neg-over} when attempting to synthesize \loverproperties, we can consider positive examples as hard constraints but need to treat negative examples as soft constraints.

For \lunderproperties, the role of positive and negative examples is inverted.
A weakest \lunderproperty must reject all negative examples while also accepting as many positive examples as possible.

%
%
\begin{example}[Examples and \lunderproperties]
    \label{ex:pos-neg-under}
    Consider again the query $\query := \exists \ecolor{\x}.\, \ucolor{\y} = \modhasharg{\ucolor{\vara}}{\ucolor{\varmodulus}}{\ecolor{\x}}$ and the set of weakest \lunderpropertiesp{\langunder}
    $\{ \ucolor{\y} = 0,\, 0 \leq \ucolor{\vara} < \ucolor{\varmodulus} \land \ucolor{\vara} = \ucolor{\y},\,
        0 \leq \ucolor{\y} < \ucolor{\varmodulus} \land -\ucolor{\varmodulus} < \ucolor{\vara} < \ucolor{\varmodulus} \land \ucolor{\vara} \neq 0 \land \isprime(\ucolor{\varmodulus}) \}$
    from \Cref{fig:modhash-under-prop}.
    While the negative example $(\ucolor{\y}, \ucolor{\vara}, \ucolor{\varmodulus}) = \negcolor{(3, 2, 6)}$ is rejected by all \lunderpropertiesp{\langunder}, the positive example $\poscolor{(1, 6, 5)}$ is not accepted by any of them.
    The \lunderpropertiesp{\langunder} in \Cref{fig:modhash-under-prop} form a best \ldisjunction, so $\poscolor{(1, 6, 5)}$ must be rejected by \emph{every} weakest \lunderpropertiesp{\langunder}.
\end{example}

We are now ready to introduce the generalizations of the key operations used by \citet{park2023specification} to synthesize strongest \loverproperties: \synth(\Cref{se:synth}), \checkimpl(\Cref{se:CheckingImpl}) and \cpover(\Cref{se:CheckingPrecision}).
Additionally, we introduce \cpunder(\Cref{se:CheckingPrecision}), an operation used alongside \synth and \checkimpl to synthesize weakest \lunderproperties.

\subsubsection{Synthesis from Examples}
\label{se:synth}

While strongest \loverproperties and weakest \lunderproperties can effectively treat some of the examples as soft constraints, as we will show in \Cref{se:oneproperty}, our synthesis algorithm can find such properties by iteratively calling a synthesis primitive $\synth$ that treats a carefully chosen set of examples as hard constraints.
Avoiding soft constraints was one of the key innovations of \citet{park2023specification} with respect to prior work~\cite{DBLP:journals/pacmpl/KalitaMDRR22}.

Given a set of positive examples $\eplus$ and a set of negative examples $\eminus$, the procedure $\synth(\eplus, \eminus)$ returns an \lproperty $\phil$ that accepts all the positive examples in $\eplus$ and rejects all the negative examples in $\eminus$, if such an \lproperty $\phil$ exists. If no such \lproperty exists, then $\synth(\eplus, \eminus)$ returns $\bot$.
Given a set of examples $E$, we write $\phil(E)$ to denote the conjunction $\bigwedge_{\ex\in E} \phil(\ex)$ and $\neg\phil(E)$ to denote the conjunction $\bigwedge_{\ex\in E} \neg\phil(\ex)$.
The operation $\synth(\eplus, \eminus)$ can be expressed as the formula $\exists \phil.\, \phil(\eplus) \wedge \neg\phil(\eminus)$.
\begin{example}[\synth]
    \label{ex:synth}
    \Cref{ex:pos-neg-over} showed there can be a negative example that no \loverproperties can reject.
    With the DSL $\langover$ defined in \Cref{fig:hash-over-grammar},
    if $\eplus = \{\poscolor{(1, 6, 5)}\}$ and $\eminus = \{\negcolor{(3, 2, 6)}\}$, then $\synth(\eplus, \eminus)$ can return the formula $\phil(\ucolor{\y}, \ucolor{\vara}, \ucolor{\varmodulus}) := \ucolor{\y} < \ucolor{\vara}$, which is not a consequence of the query $\query := \exists \ecolor{\x}.\, \ucolor{\y} = \modhasharg{\ucolor{\vara}}{\ucolor{\varmodulus}}{\ecolor{\x}}$.
    In this case, once more positive examples are added to $\eplus$ (which is something our synthesis algorithm automatically takes care of), \synth will return $\bot$.
    For example, if $\eplus$ is augmented to the set $\{ \poscolor{(1, 6, 5)}, \poscolor{(1, 1, 5)}, \poscolor{(1, -4, 5)}, \poscolor{(6, 2, 8)} \}$, then $\synth(\eplus, \eminus)$ returns $\bot$---i.e., the negative example $\negcolor{(3, 2, 6)}$ cannot be rejected by any \loverpropertiesp{\langover} and our synthesis algorithm will later remove it from $\eminus$.
\end{example}


\subsubsection{Checking Implication}
\label{se:CheckingImpl}

The $\checkimpl$ primitive described in this section allows us to check whether a formula is valid \lunderproperty or a valid \loverproperty.

Given two predicates $\phil$ and $\phil'$, the primitive $\checkimpl(\phil, \phil')$ checks whether $\phil$ is an implicant of $\phil'$ (or dually, whether $\phil'$ is consequence of $\phil$).
In logical terms, $\checkimpl(\phil, \phil')$  checks whether there does not exist an example $\ex$ that is accepted by $\phil$ but rejected by $\phil'$; it returns $\top$ or an example if the check fails. 
This check can be expressed as $\neg \exists e.\, \neg \phil'(e) \land \phil(e)$.

$\checkimpl(\query, \phil)$ returns $\top$ if a predicate $\phil$ is a consequence of $\query$.
Similarly, $\checkimpl(\phil, \query)$  returns $\top$ if a predicate $\phil$ is an implicant of $\query$.

\begin{example}[\checkimpl]
    \label{ex:check-impl}
    Consider again the query $\query := \exists \ecolor{\x}.\, \ucolor{\y} = \modhasharg{\ucolor{\vara}}{\ucolor{\varmodulus}}{\ecolor{\x}}$.
    Because the formula $\phil(\ucolor{\y}, \ucolor{\vara}, \ucolor{\varmodulus}) := \ucolor{\y} < \ucolor{\vara}$ is not a \emph{consequence} of $\query$, the primitive
    $\checkimpl(\query, \phil)$ would return a positive example that is rejected by $\phil$, such as $\poscolor{(1, -4, 5)}$.
    On the other hand, calling $\checkimpl(\query, \phil')$ on the formula $\phil'(\ucolor{\y}, \ucolor{\vara}, \ucolor{\varmodulus}) := \ucolor{\y} \geq 0$ would instead return $\top$ because the formula $\phil'$ is indeed a \emph{consequence} of $\query$.

    Similarly, for the formula $\phil(\ucolor{\y}, \ucolor{\vara}, \ucolor{\varmodulus}) := \ucolor{\y} < \ucolor{\vara}$, which is also not a \emph{implicant} of $\query$,
    running $\checkimpl(\phil, \query)$ (where this time the query $\query$ is the second parameter)  would return a negative example that is accepted by $\phil$, such as $\negcolor{(-1, 1, 3)}$.
    On the other hand, running $\checkimpl(\phil', \query)$ on the formula $\phil'(\ucolor{\y}, \ucolor{\vara}, \ucolor{\varmodulus}) := \ucolor{\y} = 0$ would instead return $\top$ because the formula $\phil'$ is an \emph{implicant} of $\query$.
\end{example}

The \checkimpl procedure can be implemented using a constraint solver.
However, the presence of quantifiers in implicants can result in a constraint with alternating quantifiers, making the check computationally harder, and most importantly, outside the capabilities of solvers that do not support quantifiers.
We discuss a practical procedure for performing this check in \Cref{se:cegqi}.

\subsubsection{Checking Precision}
\label{se:CheckingPrecision}

Checking precision---i.e., whether an \loverproperty is strongest or whether an \lunderproperty is weakest---requires more sophisticated queries than the one described above.
Specifically, one cannot simply ask whether there exists a negative example that is accepted by $\phil$ to check whether $\phil$ is a strongest \loverproperty, because, as shown in \Cref{ex:pos-neg-over}, there might be some negative example that must be accepted by every strongest \loverproperty.

In theory, to prove or disprove that an \loverproperty $\phil$ is strongest one needs to check whether there exists an \lproperty $\phil'$ that is \rone a consequence of the query $\query$ (i.e., $\phil'$ accepts all the positive examples) and \rtwo strictly stronger than $\phil$ (i.e., $\phil'$ rejects at least one more negative example while rejecting all the negative examples that were already rejected by $\phil$).
Such a check would be too expensive as it effectively asks one to synthesize a provably sound \loverproperty.

Our algorithm does not require such a powerful primitive, and instead approximates \rone and \rtwo using a set of positive examples $\eplus$ accepted by $\phil$ and a set of negative examples $\eminus$ rejected by $\phil$.
By combining implication and precision checks in a counterexample-guided fashion, our algorithm improves the approximation over time and is thus sound.

Given an \loverproperty $\phil$, a set of positive examples $\eplus$ accepted by $\phil$, a set of negative examples $\eminus$ rejected by $\phil$,
a query $\query$,
and the \loverproperties $\phiand$ we have already synthesized,
$\cpover(\phil, \phiand, \query,  \eplus, \eminus)$ checks 
if there do not exist an \lproperty $\phil'$ and an negative example $\ex \notin \interp{\query}$ satisfying $\phiand$ such that:
\rone $\phil'$ accepts all the positive examples in $\eplus$;
\rtwo $\phil'$ rejects $\ex$ and all the negative examples in $\eminus$, whereas $\phil$ accepts $\ex$.
In our algorithm, the formula $\phiand$ is used to ensure that the example produced by \cpover is not already rejected by best \loverproperties we have already synthesized.
The above check can be logically stated as follows:
\begin{equation}
    \label{eqn:cpover}
    \cpover(\phil, \hspace{-1pt}\phiand, \hspace{-2pt}\query\hspace{-1pt},  \hspace{-1pt}\eplus\hspace{-2pt}, \hspace{-1pt}\eminus)
    \hspace{-1pt}=\hspace{-1pt}
    \neg \exists \phil'\hspace{-2pt}, \ex.
    \hln{\neg \query(\ex) \hspace{-1pt} \land \hspace{-1pt}
        \phiand(\ex) \hspace{-1pt} \land \hspace{-1pt}
        \phil(\ex) \hspace{-1pt} \land \hspace{-1pt}
        \neg\phil'(\ex)} \hspace{-1pt} \land \hspace{-1pt}
    \phil'(\eplus) \hspace{-2pt} \land \hspace{-2pt}
    \neg\phil'(\eminus)
\end{equation}
The highlighted part of the formula is what changes when checking if the formula is weakest.

\begin{example}[\cpover]
    \label{ex:check-strongest}
    Consider again the query $\query := \exists \ecolor{\x}.\, \ucolor{\y} = \modhasharg{\ucolor{\vara}}{\ucolor{\varmodulus}}{\ecolor{\x}}$, and an \loverpropertyp{\langover} $\phil(\ucolor{\y}, \ucolor{\vara}, \ucolor{\varmodulus}) := \ucolor{\y} \neq \ucolor{\varmodulus}$, which is \emph{not} a strongest one.

    $\cpover(\phil, \top, \query, \{\poscolor{(1, 6, 5)}\}, \{\negcolor{(3, 1, 3)}\})$ can return a strictly stronger \loverpropertyp{\langover}
    $\phil'_1(\ucolor{\y}, \ucolor{\vara}, \ucolor{\varmodulus}) := \ucolor{\y} < \ucolor{\varmodulus}$ with a negative example $\negex_1 = \negcolor{(6, 1, 5)}$.

    However, because $\cpover$ only considers whether the formula $\phil$ is strongest with respect to the examples $\{\poscolor{(1, 6, 5)}\}, \{\negcolor{(3, 1, 3)}\}$, it can also alternatively return a property that is not an actual \loverpropertyp{\langover}---e.g., $\phil'_2(\ucolor{\y}, \ucolor{\vara}, \ucolor{\varmodulus}) := \ucolor{\y} < \ucolor{\vara}$ with a negative example $\negex_2 = \negcolor{(6, 1, 5)}$.
    The formula $\phil'_2$ is not a \loverpropertyp{\langover} of $\query$ as it rejects the positive example $\poscolor{(1, -4, 5)}$.
\end{example}

When checking if an \lunderproperty $\phil$ is a weakest one, 
we can perform a dual check and ask if there does not exist an \lproperty that can accept one more example than the current formula. That is,
$\cpunder(\phil, \phior, \query, \eplus, \eminus)$ checks whether 
there do not exist an \lproperty $\phil'$ and
a positive example $\ex \in \interp{\query}$ satisfying $\phior$ such that:
\rone $\phil'$ accepts all the positive examples in $\eplus$;
\rtwo $\phil'$ accepts $\ex$ and all the positive examples in $\eplus$, whereas $\phil$ rejects $\ex$.
This check can be logically stated as
\begin{equation}
    \label{eqn:cpunder}
    \cpunder(\phil,\hspace{-1pt} \phior,\hspace{-2pt}  \query\hspace{-1pt}, \eplus\hspace{-2pt}, \eminus)
    \hspace{-1pt}=\hspace{-1pt}
    \neg \exists \phil',\hspace{-2pt} \ex.
    \hln{\query(\ex) \hspace{-1pt} \land \hspace{-1pt}
        \neg \phior(\ex) \hspace{-1pt} \land \hspace{-1pt}
        \neg\phil(\ex) \hspace{-1pt} \land \hspace{-1pt}
        \phil'(\ex)} \hspace{-1pt} \land \hspace{-1pt}
    \phil'(\eplus) \hspace{-1pt} \land \hspace{-1pt}
    \neg\phil'(\eminus)
\end{equation}

\begin{example}[\cpunder]
    \label{ex:check-weakest}
    Consider again the query $\query := \exists \ecolor{\x}.\, \ucolor{\y} = \modhasharg{\ucolor{\vara}}{\ucolor{\varmodulus}}{\ecolor{\x}}$, and a \lunderpropertyp{\langunder} $\phil(\ucolor{\y}, \ucolor{\vara}, \ucolor{\varmodulus}) := \ucolor{\y} = 0 \wedge \ucolor{\vara} = 0$, which is \emph{not} a weakest \lunderpropertyp{\langunder}.
    \cpunder$(\phil, \bot, \query, \{\poscolor{(0, 0, 5)}\}, \{\negcolor{(3, 2, 6)}\})$ can return a strictly weaker \lunderpropertyp{\langunder}
    $\phil'_1(\ucolor{\y}, \ucolor{\vara}, \ucolor{\varmodulus}) := \ucolor{\y} = 0$ with a positive example $\posex_1 = \poscolor{(0, 1, 5)}$.
    However, for the same reasons outlined in \Cref{ex:check-strongest}, the returned property may not be a \lunderpropertyp{\langunder}.
\end{example}

The \cpover and \cpunder procedures are effectively solving a synthesis problem---i.e., they are looking for a formula---and implementing them requires a form of example-based synthesis.
Again, the presence of quantifiers in the negation of the query $\neg \query$ for \cpover can result in constraint \eqref{eqn:cpover} containing alternating quantifiers, thus bringing us outside of the capability of many program synthesizers.
We discuss a practical procedure for sidestepping the quantifier-alternation problem and performing this check in \Cref{se:cegqi}.

\subsection{Synthesizing One Strongest \loverproperty and One Weakest \lunderproperty}
\label{se:oneproperty}

We are now ready to describe our main procedures: \synthoverproperty (Algorithm~\ref{alg:SynthesizeProperty}) and \synthunderproperty (Algorithm~\ref{alg:SynthesizeUnderProperty}).
We first recall the description of \synthoverproperty by \citet{park2023specification}, the algorithm that synthesizes a strongest \loverproperty that is incomparable with the \loverproperties we already synthesized (the algorithm will be used in \Cref{se:allproperties} to synthesize one \loverproperty at a time).
We then describe one of the contributions of this paper, i.e., how the algorithm changes for its under-approximated dual \synthunderproperty.

\subsubsection{Synthesizing One Strongest \loverproperty}
\label{se:synth-over-prop}

Given a query formula $\query$ and a conjunction of \loverproperties we have already synthesized $\phiand$, the procedure $\synthoverproperty$ synthesizes a strongest \loverproperty $\phil$ for the query $\query$ that is incomparable to the already synthesized formulas in $\phiand$.
We say an \loverproperty $\phil$ for the query $\query$ is strongest \emph{with respect to} $\phiand$ if there does not exist an \loverproperty $\phil'$ for $\query$ such that $\phil' \land \phiand$ is strictly stronger than $\phil \land \phiand$---i.e., the \loverproperty $\phil$ is incomparable to all the \loverproperties in $\phiand$.
\begin{figure}[!t]
    \begin{minipage}{0.50\linewidth}
        \begin{algorithm}[H]
            \footnotesize
            {\it
            \caption{\fontsize{7.8}{10}\selectfont\newline$\synthoverproperty(\query, \phiand, \phiinit, \eplus, \eminus)$}
            \label{alg:SynthesizeProperty}
            \DontPrintSemicolon
            $\phil, \philast \gets \phiinit\text{; } \eminusmay \gets \emptyset$

            \While{true}{
            \label{Li:LoopStart}

            $\posex \gets \checkimpl(\query, \phil)$
            \label{Li:CallCheckSoundness}

            \eIf {$\posex \neq \top$} {
            $\eplus \gets \eplus \cup \{\posex\}$
            \label{Li:UpdatePlus}

            $\phil' \gets \synth(\eplus, \eminus \cup \eminusmay)$
            \label{Li:CallSynthesize}

            \eIf {$\phil' \neq \bot$} {
                $\phil \gets \phil'$
            }{
                $\phil \gets \philast\text{; } \quad\eminusmay \gets \emptyset$
                \label{Li:RevertLastSound}
            }
            }{
            $\eminus \gets \eminus \cup \eminusmay\text{; } \quad\eminusmay \gets \emptyset$
            \label{Li:MergeMayIntoMust}

            $\philast \gets \phil$
            \label{Li:StorePhiLast}


            $\negex\hspace{-3pt}, \phil' {\gets} \cpover(\phil,\phiand, \query, \hspace{-1pt}\eplus, \hspace{-1pt}\eminus)\hspace{-8pt}$
            \label{Li:CallCheckPrecision}

            \eIf {$\negex \neq \top$} {
                $\eminusmay \gets \{\negex\}$
                \label{Li:ReinitializeMinusMay}

                $\phil \gets \phil'$
                \label{Li:ReinitializePhi}
            }{
                \Return $\phil, \eplus, \eminus$
                \label{Li:SynthPropReturn}
            }
            }
            }
            }
        \end{algorithm}
    \end{minipage}
    \begin{minipage}{0.49\linewidth}
        \begin{algorithm}[H]
            \footnotesize
            {\it
            \caption{\newline\fontsize{7.8}{10}\synthunderproperty($\query, \diffcolor{\phior}, \phiinit, \eplus, \eminus$)}
            \label{alg:SynthesizeUnderProperty}
            \DontPrintSemicolon
            $\phil, \philast \gets \phiinit\text{; } \diffcolor{\eplusmay} \gets \emptyset$
            %


            \While{true}{
            {$\diffcolor{\negex} \gets \checkimpl(\diffcolor{\phil}, \diffcolor{\query})$}

            \eIf {$\diffcolor{\negex} \neq \top$} {
            {$\diffcolor{\eminus} \gets \diffcolor{\eminus \cup \{\negex\}}$}

            $\phil' \gets \synth(\diffcolor{\eplus \cup \eplusmay}, \diffcolor{\eminus})$

            \eIf {$\phil' \neq \bot$} {
                $\phil \gets \phil'$
            }{
                $\phil \gets \philast\text{; } \quad\diffcolor{\eplusmay} \gets \emptyset$
            }
            }{
            $\diffcolor{\eplus} \gets \diffcolor{\eplus \cup \eplusmay}\text{; } \quad \diffcolor{\eplusmay} \gets \emptyset$

            $\philast \gets \phil$


            $\diffcolor{\posex}, \phil' {\gets} \diffcolor{\cpunder}(\phil, \diffcolor{\phior}, \query, \eplus, \eminus)$

            \eIf {$\diffcolor{\posex} \neq \top$} {
                $\diffcolor{\eplusmay} \gets \diffcolor{\{\posex\}}$

                $\phil \gets \phil'$
            }{
                \Return $\phil, \eplus, \eminus$ 
            }
            }
            }
            }
        \end{algorithm}
    \end{minipage}
    \vspace{-4mm}
\end{figure}

In each iteration, \synthoverproperty performs two steps.
First, it uses \checkimpl to check whether the current candidate $\phil$ is a consequence of $\query$ (line \ref{Li:CallCheckSoundness}).
Second, if the candidate $\phil$ is a consequence of $\query$, it uses \cpover to check whether $\phil$ is strongest with respect to $\phiand$ (line \ref{Li:CallCheckPrecision}).
The algorithm terminates once a formula passes both checks (line \ref{Li:SynthPropReturn}).

If the current candidate $\phil$ is not a consequence of $\query$, \checkimpl returns a positive example $\posex$ (line \ref{Li:CallCheckSoundness}). The algorithm then adds $\posex$ to the set of positive examples $\eplus$ and uses it to \synth a new candidate \lproperty (lines \ref{Li:UpdatePlus} and \ref{Li:CallSynthesize}).

If the current candidate $\phil$ is a consequence of $\query$ but there is an \lproperty $\phil'$ that aligns with the current set of positive and negative examples, $\eplus$ and $\eminus$, and can reject one more negative example $\negex$, \cpover returns this property $\phil'$ along with the negative example $\negex$.
The example $\negex$ is then temporarily stored in $\eminusmay$ without immediately updating $\eminus$ (line \ref{Li:ReinitializeMinusMay}).
Updating $\eminus$ is delayed because $\phil'$ may not be a consequence of the query $\query$ (\Cref{ex:check-strongest}), and in the worst case, there might not exist an \loverproperty that rejects $\negex$ (\Cref{ex:pos-neg-over}).
The example stored in $\eminusmay$ can be safely added to $\eminus$ when \checkimpl verifies that \lproperty $\phil$ returned by \cpover is indeed a consequence of the query $\query$ (line \ref{Li:MergeMayIntoMust}); at this point we are certain that the example in $\eminusmay$ can be rejected by at least one \loverproperty, as witnessed by $\phil$.\footnote{
    Delaying the update of negative examples enables treating all negative examples in $\eminus$ as hard constraints.
    This is a key innovation by \citet{park2023specification} over prior work \cite{DBLP:journals/pacmpl/KalitaMDRR22}.}

The candidate \lproperty returned by \cpover in line \ref{Li:CallCheckPrecision} is only guaranteed to be consistent with the examples, and therefore must be checked again by \checkimpl in line \ref{Li:CallCheckSoundness}.
If the candidate fails to pass \checkimpl,
the algorithm keeps adding more positive examples until either \rone it finds a \loverproperty that rejects all the negative examples in $\eminus \cup \eminusmay$, or \rtwo \synth in line \ref{Li:CallSynthesize} fails to find an \lproperty.
In the latter case, \synthoverproperty concludes that the example in $\eminusmay$ cannot be rejected by any \loverproperty and thus restarts after discarding the example in $\eminusmay$ (line \ref{Li:RevertLastSound}).
For efficiency, whenever $\eminus$ is updated in line \ref{Li:MergeMayIntoMust}, the algorithm stores the current \loverproperty $\phil$ that rejects all the negative examples in $\eminus$ in a variable $\philast$ (line \ref{Li:StorePhiLast}).
In this way, the algorithm can revert to $\philast$ when $\eminusmay$ is discarded (line \ref{Li:RevertLastSound}).

\begin{example}[\Cref{alg:SynthesizeProperty} Run]
    \label{ex:synth-strongest-consequence}
        Consider the query $\query := \exists \ecolor{\x}.\, \ucolor{\y} = \modhasharg{\ucolor{\vara}}{\ucolor{\varmodulus}}{\ecolor{\x}}$.
        %
        %
        %
        The table below shows the last 4 iterations in a possible execution of \synthoverproperty$(\query, \top, \top, \emptyset, \emptyset)$. 
        Specifically, it shows the value of $\phil$, $\eplus$, and $\eminus$ at the start of each iteration (\Cref{Li:LoopStart}). 
        In the table, $\posex_1 = (\ucolor{y}, \ucolor{a}, \ucolor{M}) = \poscolor{(1,6,5)}, \posex_2 = \poscolor{(1,1,5)}, \posex_3 = \poscolor{(1,-4,5)}, \posex_4 = \poscolor{(6,2,8)}$, $\negex_1 = \negcolor{(8,1,8)}$, $\negex_2 = \negcolor{(3,2,6)}$, $\negex_3 = \negcolor{(6,2,5)}$.

\noindent
        \noindent\textul{Iteration $n-3$.} So far the algorithm has computed an \loverpropertyp{\langover} $\ucolor{y} \neq \ucolor{M}$ that is not a strongest one.
        Therefore, \checkimpl passes but \cpover fails, and the execution reaches \Cref{Li:ReinitializeMinusMay}.
        Then $\phil$ is set to a new candidiate $\isprime(\ucolor{M})$, and a negative example $\negex_2$ is added to $\eminusmay$.
        
\noindent
        \noindent\textul{Iteration $n-2$.} The property $\isprime(\ucolor{M})$ is not an \loverpropertyp{\langover}, so \checkimpl fails. The execution reaches \Cref{Li:UpdatePlus}, and a positive example $\posex_4$ is added to $\eplus$.
        As discussed in \Cref{ex:synth}, $\negex_2$ cannot be rejected by any \loverpropertyp{\langover}, so \synth fails, $\phil$ is reverted to $\ucolor{y} \neq \ucolor{M}$, and $\eminusmay$ is cleared.

\noindent
\begin{minipage}{.50\textwidth}
\vspace{6pt}
        \noindent\textul{Iteration $n-1$.} Similar to iteration $n-3$ but with a new positive example $\posex_4$ added, \cpover returns a new candidiate $\ucolor{y}<\ucolor{M}$ and a new negative example $\negex_3$, which is added to $\eminusmay$.

        \noindent\textul{Iteration $n$.} The property $\ucolor{y}<\ucolor{M}$ is indeed a strongest \loverpropertyp{\langover}, so it passes both the \checkimpl and \cpover and is finally returned.
\end{minipage}
\begin{minipage}{.49\textwidth}
\centering
\small
\begin{tabular}{c|c@{\hskip 4pt}c@{\hskip 3pt}c@{\hskip 3pt}c}  
        \hline
        \hline
        \textbf{Iter.} & $\phil$                      & $\eplus$                                     & $\eminus$      & $\eminusmay$   \\ \hline
        $n-3$              & $\ucolor{y} \neq \ucolor{M}$ & $\{\posex_1, \posex_2, \posex_3\}$           & $\{\negex_1\}$ & $\emptyset$    \\ 
        $n-2$              & $\isprime(\ucolor{M})$       & $\{\posex_1, \posex_2, \posex_3\}$           & $\{\negex_1\}$ & $\{\negex_2\}$ \\ 
        $n-1$              & $\ucolor{y} \neq \ucolor{M}$ & $\{\posex_1, \posex_2, \posex_3, \posex_4\}$ & $\{\negex_1\}$ & $\emptyset$    \\ 
        $n$                & $\ucolor{y} < \ucolor{M}$    & $\{\posex_1, \posex_2, \posex_3, \posex_4\}$ & $\{\negex_1\}$ & $\{\negex_3\}$ \\ 
        \hline
        \hline
    \end{tabular}
\end{minipage}

\end{example}


Once \synthoverproperty terminates, it returns a strongest \loverproperty for $\query$ (with respect to $\phiand$) that accepts all the examples in $\eplus$ and rejects all the examples in $\eminus$.
Informally, if \checkimpl returns $\top$, then the property is an \loverproperty, and if \cpover returns $\top$, then the property is a strongest property.

\synthoverproperties is also guaranteed to terminate for a finite DSL, when every call to primitives \synth, \checkimpl and \cpover terminates.
Informally, \checkimpl at \Cref{Li:CallCheckSoundness} can only return a counterexample (i.e., $\phil$ is not a consequence) finitely many times, and \cpover at \Cref{Li:CallCheckPrecision} can only return a counterexample (i.e., $\phil$ is a consequence but not a strongest one) finitely many times between successive instances where \checkimpl returns $\top$.
The full proof is in \Cref{app:proof}.


\subsubsection{Synthesizing One Weakest \lunderproperty}

\label{se:synth-under-prop}
Given a query formula $\query$, a disjunction of \lunderproperties we already synthesized $\phior$, the goal of the procedure $\synthunderproperty$ is to synthesize a strongest \lunderproperty $\phil$ for the query $\query$ that is incomparable to the already synthesized formulas in $\phior$.

We say an \lunderproperty $\phil$ for the query $\query$ is weakest \emph{with respect to} $\phior$ if there does not exist an \lunderproperty $\phil'$ for $\query$ such that $\phil' \lor \phior$ is strictly weaker than $\phil \lor \phior$---i.e., the \lunderproperty $\phil$ is incomparable to the already synthesized \lunderproperties.

\synthunderproperty solves the dual of the problem solved by \synthoverproperty, and the two algorithms share the same structure.
Due to the duality
\rone the roles of positive and negative examples are inverted;
\rtwo \checkimpl in line \ref{Li:CallCheckSoundness} checks whether $\phil$ is an \emph{implicant} of $\query$,
instead of checking that $\phil$ is a \emph{consequence} of $\query$; and
\rthree precision is checked by \cpunder instead of \cpover.
These changes are highlighted in \diffcolor{violet} in Algorithm~\ref{alg:SynthesizeUnderProperty}.



\subsection{Synthesizing a Best \lconjunction and \ldisjunction}
\label{se:allproperties}

We conclude by briefly recalling how \synthoverproperties works (as described by \citet{park2023specification}) and present the dual algorithm \synthunderproperties.
These two algorithms use \synthoverproperty and \synthunderproperty to synthesize a best \lconjunction and \ldisjunction, respectively.
The detailed algorithms are illustrated in \Cref{app:alg-allproperties}.


The algorithm \synthoverproperties iteratively synthesizes incomparable strongest \loverproperties.
At each iteration, \synthoverproperties keeps track of the conjunction of synthesized strongest \loverproperties $\phiand$,
and calls \synthoverproperty to synthesize a strongest \loverproperty for $\query$ with respect to $\phiand$.

If \synthoverproperty returns an \loverproperty $\phil$ that does not reject any example that was not already rejected by $\phiand$, the formula $\phiand$ is a best \lconjunction, and thus \synthoverproperties returns the set of synthesized \loverproperties.

If \synthoverproperty returns an \loverproperty $\phil$ that rejects some example that was not rejected by $\phiand$,
\synthoverproperties needs to further strengthen $\phil$ to a strongest \loverproperty for $\query$ with respect to examples that might already be rejected by $\phiand$.
Without this step the returned \loverproperty may be imprecise for examples that were not considered by \synthoverproperty because they were outside of $\phiand$.
To achieve this further strengthening, \synthoverproperties makes another call to \synthoverproperty with the example sets $\eplus$ and $\eminus$ returned by the previous call to \synthoverproperty together with $\phil$ and $\phiinit := \phil$, but with $\phiand := \tru$.

Again, because \synthunderproperties solves the dual problem of the one solved by \synthoverproperties, the two algorithms share the same structure.
\synthunderproperties uses \synthunderproperty in a similar manner, but it maintains the disjunction of synthesized weakest \lunderproperties instead of conjunction. For weakening, it also makes another call to \synthunderproperty, but $\phiand := \tru$ is replaced by $\phior := \bot$.

Using the argument of \citet{park2023specification} for \synthoverproperties, assuming all the primitives always terminate, our algorithms satisfy the following soundness and completeness theorems,
Specifically, the number of iterations in \Cref{alg:SynthesizeProperty} and \ref{alg:SynthesizeUnderProperty} is bounded by the number of properties in the DSL $\lang$ and the number of examples in the domain, whichever is smaller.
%
\begin{restatable}[Soundness]{theorem}{soundnessconj}
    \label{thm:soundness}
    If \synthoverproperties terminates, $\phiand$ is a best \lconjunction for $\query$.
    If \synthunderproperties terminates, $\phior$ is a best \ldisjunction for $\query$.
\end{restatable}
\begin{restatable}[Relative Completeness]{theorem}{finitecompleteness}
    \label{thm:finite-completeness}
    Suppose that either $\lang$ contains finitely many formulas,
    or the example domain is finite.

    If \synth, \checkimpl and \cpover are decidable on $\lang$,
    then \synthoverproperty and \synthoverproperties always terminate.

    If \synth, \checkimpl and \cpunder are decidable on $\lang$,
    then \synthunderproperty and \synthunderproperties always terminate.
\end{restatable}
Part of the proof is stated in the end of \Cref{se:synth-over-prop}. 
The full proof is in \Cref{app:proof}.

Also, our procedure for solving problems expressed in the \framework framework iteratively synthesizes strongest \loverproperties (resp. \lunderproperties) that can strengthen (resp. weaken) the current \lconjunction (resp. \ldisjunction) until no further strengthening (resp. weakening) is possible.
Thanks to this iterative approach, even when the procedure does not terminate, one can output intermediate results, which are properties that are guaranteed to be strongest (resp. weakest), though they might not form a best \lconjunction (resp. \ldisjunction).

%% file: 5cegqi.tex
\section{Counterexample-Guided Quantifier Instantiation}
\label{se:cegqi}

In \Cref{se:algorithm}, when we discussed the main specification-synthesis loop, we assumed we were given implementations of all the needed primitives.
In this section, we explain how each primitive can be implemented for existentially quantified queries.

%
\cpover (line \ref{Li:CallCheckPrecision}) in \synthoverproperty and  \checkimpl (line \ref{Li:CallCheckSoundness}) in \synthunderproperty check for the existence of a new negative example (along with additional constraints).
However, when dealing with an existentially quantified query $\query := \phiquery$,
a negative example $\ucolor{e}$ must be such that the formula $\neg\phiprog(\ucolor{e}, \ecolor{h})$ is valid for \textit{all} values of the existentially quantified variable $\ecolor{h}$.
Therefore, checking the existence of a negative example $\ucolor{e}$ requires solving a formula that has alternating quantifiers.
To handle these primitives involving quantifier alternation, we propose a CounterExample-Guided Quantifier Instantiation (CEGQI) algorithm similar to the one by \citet{Reynolds2015CEQI}, which can implement the primitives that require finding negative counterexamples using only existentially-quantified formulas.

\subsubsection*{Counterexample-Guided Quantifier Instantiation for Weakest \lunderproperty}
\label{se:cegqi-under}

We start with the simpler of the two queries, \checkimpl in \synthunderproperty (line \ref{Li:CallCheckSoundness}), which requires solving a formula with alternating quantifiers of the following form (by negating formulas in \S\ref{se:CheckingImpl}):
\begin{equation}
    \label{eq:alternating-quant-neg}
    \exists \ucolor{\ex}.\, \forall \ecolor{h}.\, \neg \phiprog(\ucolor{\ex}, \ecolor{h}) \land \phil(\ucolor{\ex})
\end{equation}

The CEGQI algorithm for solving \Cref{eq:alternating-quant-neg} iteratively builds a set $H$ of possible values for $\ecolor{h}$ and finds a value of $\ucolor{\ex}$ that is consistent with the finite set of values $H$.
The set $H$ is updated by repeating the following two operations until a solution that holds for all values of $\ecolor{h}$ is found.

\mypar{Generating Candidate Negative Example}
Given formulae $\phil$, $\phiprog$, and a finite set $H$ of values the existentially-quantified variable $\ecolor{h}$ can take, $\genex(\phil, \phiprog, H)$ generates an example $\ucolor{\ex} \in \interp{\phil}$ such that the formula $\phiprog(\ucolor{\ex}, \ecolor{h})$ does not hold for all the values $\ecolor{h}$ in the set $H$, if such an example $\ucolor{\ex}$ exists.
If no such example exists, $\genex(\phil, \phiprog, H)$ returns $\bot$.
%
Formally:
\begin{equation}
    \label{eq:genex}
    {\textstyle
        \genex(\phil, \phiprog, H)=\exists \ucolor{\ex}.\, \bigwedge_{\ecolor{h} \in H} \neg\phiprog(\ucolor{\ex}, \ecolor{h}) \land \phil(\ucolor{\ex})
    }
\end{equation}

\mypar{Checking Candidate Negative Example}
Given a formula $\phiprog$, and a candidate negative example $\ucolor{\ex}$, the function $\checkex(\phiprog, \ucolor{\ex})$ checks if there does not exist a value for the existentially quantified variable $\ecolor{h}$ such that $\phiprog(\ucolor{\ex}, \ecolor{h})$ holds true (i.e., whether there exists a value of $\ecolor{h}$ that makes the example actually positive); it returns $\top$ or the value of $\ecolor{h}$ if the check fails. Formally:
\begin{equation}
    \label{eq:checkex}
    \checkex(\phiprog, \ucolor{\ex})= 
    \neg \exists \ecolor{h}.\, \phiprog(\ucolor{\ex}, \ecolor{h})
\end{equation}

\mypar{Counterexample-Guided Quantifier Instantiation}
The CEGQI algorithm for \checkimpl (\Cref{alg:CEGIS-Loop}) iteratively generates candidate negative examples using \genex and checks whether they are actually negative using \checkex.
Across iterations, it maintains the set of values of $\ecolor{h}$ returned by \checkex in $H$, and uses \genex to find an example $\ex$ that behaves well for all the values in $H$ discovered so far---i.e., $\ucolor{\ex}$ satisfies $\neg\phiprog(\ucolor{\ex}, H) \land \phil(\ucolor{\ex})$ (line \ref{Li:CallGenEx}).

If \genex fails to find an example, it means that there is no example $\ucolor{\ex}$ satisfying $\neg\phiprog(\ucolor{\ex}, H) \land \phil(\ucolor{\ex})$, thereby a stronger condition in \Cref{eq:alternating-quant-neg} also cannot be satisfied. Therefore, \cegqi returns $\top$ (line \ref{Li:ReturnBot})---i.e., there does not exist a valid negative example.

If \genex returns an example $\ucolor{\ex}$, the example is tested by $\checkex$ to check whether $\phiprog(\ucolor{\ex}, \ecolor{h})$ does not hold for every possible value of $\ecolor{h}$ and not only for values found so far (line \ref{Li:CallCheckEx}).
The algorithm returns the example $\ucolor{\ex}$ once it passes the check (line \ref{Li:ReturnEx}), but if it fails the check, a new counterexample $\ecolor{h}$ returned by \checkex is added to the set $H$, and the algorithm restarts at line \ref{Li:CallGenEx}.
Note that the set of instances $H$ can be cached and reused across different calls to \cegqi.

\noindent
\begin{minipage}{0.73\textwidth}
\vspace{6pt}
\begin{example}[\Cref{alg:CEGIS-Loop} Run]
    \label{ex:cegqi}
    Consider again the query $\query :=  \exists \ecolor{\x}.\, \ucolor{\y} = \modhasharg{\ucolor{\vara}}{\ucolor{\varmodulus}}{\ecolor{\x}}$ and the formula $\phil := 0 \le \ucolor{y} < \ucolor{M}$.
        The table on the right shows a possible execution of how $\checkimpl(\phil, \query) := \forall \ucolor{\ex}.\, \phil(\ucolor{\ex}) \Rightarrow \exists \ecolor{h}.\, \phiprog(\ucolor{\ex}, \ecolor{h})$ can find a negative example and prove that $\phil$ is not a \lunderpropertyp{\langunder}.
        Specifically, it shows the values of $\ucolor{e} = (\ucolor{y},\ucolor{a},\ucolor{M})$ and $\ecolor{h}= \ecolor{x}$ at the end of each iteration (\Cref{Li:UpdateH}) within \Cref{alg:CEGIS-Loop}.
        For each iteration, $\phil(\ucolor{e})$ is always true for $\ucolor{e} = (\ucolor{y}, \ucolor{a}, \ucolor{M})$ ,
        while $\phiprog(\ucolor{e}, \ecolor{h})$ is false for all previous $\ecolor{h} = \ecolor{x}$ but true for $\ecolor{h}$ in the current iteration.
\end{example}
\end{minipage}
\begin{minipage}{0.26\textwidth}
    \centering
{\footnotesize
    \begin{tabular}{c|ccc|c}
        \toprule
         & \multicolumn{3}{c|}{$\ucolor{e}$} & $\ecolor{h}$ \\
        \textbf{Iter}. & $\ucolor{y}$ & $\ucolor{a}$ & $\ucolor{M}$ & $\ecolor{x}$ \\
        \midrule
        \textbf{1} & 2 & 2 & 4 & 1 \\
        \textbf{2} & 3 & 1 & 4 & 3 \\
        \textbf{3} & 0 & 2 & 4 & 2 \\
        \textbf{4} & 0 & 3 & 4 & 0 \\
        \textbf{5} & 3 & 2 & 4 & $\top$ \\
        \bottomrule
    \end{tabular}
}
\end{minipage}

\begin{figure}[!t]
    \begin{minipage}{0.44\textwidth}
        \begin{algorithm}[H]
            \footnotesize{
                \it
                \caption{\fontsize{7.8}{10}$\checkimpl(\phil, \query)$}
                \label{alg:CEGIS-Loop}
                \DontPrintSemicolon
                assume $\query = \exists h.\, \phiprog(v, h)$\\
                $H \gets \emptyset$ \label{Li:InitializeH}\\
                \While{$\tru$}{
                    $\ex \gets \genex (\phil, \phiprog, H)$
                    \label{Li:CallGenEx}
                    \\

                    \eIf{$\ex = \bot$}
                    {
                        \Return $\top$
                        \label{Li:ReturnBot}\\
                    }
                    {
                        $h \hspace{-1pt}\gets\hspace{-1pt} \checkex(\phiprog,\hspace{-1pt} \ex)$
                        \label{Li:CallCheckEx}\\
                        \eIf{$h = \top$}
                        {
                            \Return $\ex$
                            \label{Li:ReturnEx}\\
                        }
                        {
                            $H \gets H \cup \{h\}$
                            \label{Li:UpdateH}
                        }
                    }
                }
            }
        \end{algorithm}
    \end{minipage}
    \begin{minipage}{0.54\textwidth}
        \begin{algorithm}[H]
            \footnotesize{
                \it
                \caption{\fontsize{7.8}{10}$\cpover(\phil, \phiand, \query,  \eplus, \eminus)$}
                \label{alg:CEGQI-Loop-Over}
                \DontPrintSemicolon
                assume $\query = \exists h.\, \phiprog(v, h)$\\
                $H \gets \emptyset$\\
                \While{$\tru$}{
                    $\ex$,\hspace{-1pt} \diffcolor{$\phil'$ \hspace{-2pt}}$\gets$\hspace{-2pt}
                    \diffcolor{\genspec} ($\phil$,\hspace{-1pt} $\phiprog$,\hspace{-1pt}
                    \diffcolor{$\phiand$,\hspace{-1pt} $\eplus$\hspace{-1pt},\hspace{-1pt} $\eminus$\hspace{-2pt}},\hspace{-1pt} $H$)
                    \\

                    \eIf{$\ex = \bot$}
                    {
                        \Return $\top$ \\
                    }
                    {
                        $h \hspace{-1pt}\gets\hspace{-1pt} \checkex(\phiprog,\hspace{-1pt} \ex)$ \\
                        \eIf{$h = \top$}
                        {
                            \Return $\ex$\\
                        }
                        {
                            $H \gets H \cup \{h\}$
                        }
                    }
                }
            }
        \end{algorithm}
    \end{minipage}
\end{figure}


\subsubsection*{Counterexample-Guided Quantifier Instantiation for Strongest \loverproperty}
\label{se:cegqi-over}

The call to \cpover in \synthoverproperty (line \ref{Li:CallCheckPrecision}) requires solving a formula that has alternating quantifiers
and the following form (by negating appropriate formulas in \S\ref{se:CheckingPrecision}):
\begin{equation}
    \label{eq:alternating-quant-neg-with-formula}
    \exists \ucolor{\ex}, \phil'.\, \forall \ecolor{h}.\,
    \neg\phiprog(\ucolor{\ex}, \ecolor{h}) \land
    \phiand(\ex) \land
    \phil(\ex) \land
    \neg\phil'(\ex) \land
    \phil'(\eplus) \land
    \neg\phil'(\eminus)
\end{equation}

This formula looks more complicated due to the presence of the existential variable $\phil'$.
However, a similar \cegqi approach to the one presented in \Cref{se:cegqi-under} can also be used to solve \Cref{eq:alternating-quant-neg-with-formula}, by finding a negative example and a formula in tandem.

The only change in the CEGQI algorithm for solving \cpover
(\Cref{alg:CEGQI-Loop-Over}) is that \genex is replaced by a new operation,
\genspec, defined as follows.
Given formulae $\phil$, $\phiprog$, $\phiand$, set of examples $\eplus$, $\eminus$, and a finite set $H$ of values the existentially-quantified variable $\ecolor{h}$ can take, $\genspec (\phil, \phiprog, \phiand, \eplus, \eminus, H)$ generates an \lproperty $\phil'$ and an example $\ucolor{\ex}$ satisfying $\phiand$ such that
\rone the formula $\phiprog(\ucolor{\ex}, \ecolor{h})$ does not hold for all the values $\ecolor{h}$ in the set $H$
\rtwo $\phil'$ accepts all the positive examples in $\eplus$;
\rthree $\phil'$ rejects $\ucolor{\ex}$ and all the negative examples in $\eminus$, whereas $\phil$ accepts $\ucolor{\ex}$,
if such an example $\ucolor{\ex}$ and formula $\phil'$ exist.
If no such example exists, then $\genspec (\phil, \phiprog, \phiand, \eplus, \eminus, H)$ returns $\bot$.
%
Stated formally:
\begin{multline}
    \label{eq:genspec}
    \genspec (\phil, \phiprog, \phiand, \eplus, \eminus, H)= \\
    {\textstyle
    \exists \ucolor{\ex}, \phil'. \bigwedge_{\ecolor{h} \in H} \neg\phiprog(\ucolor{\ex}, \ecolor{h}) \land
    \phiand(\ex) \land
    \phil(\ex) \land
    \neg\phil'(\ex) \land
    \phil'(\eplus) \land
    \neg\phil'(\eminus)
    }
\end{multline}

Because the variable $\ecolor{h}$ only appears in the constraint $\neg\phiprog(\ucolor{e}, \ecolor{h})$, whether $\ucolor{e}$ is indeed a negative example can still be tested using \checkex~\eqref{eq:checkex}.

Similar to the \Cref{alg:CEGIS-Loop}, if \genex fails to find an example, it means that there is no example $\ucolor{\ex}$ satisfying  \Cref{eq:genspec}, thereby a stronger condition in \Cref{eq:alternating-quant-neg-with-formula} also cannot be satisfied.
The example is only returned after it has been tested by \checkex to ensure that $\phiprog(\ucolor{\ex}, \ecolor{h})$ does not hold for every possible value of $\ecolor{h}$.

\subsubsection*{Correctness}
The above observations are summarized as the following soundness theorem.

\begin{restatable}[Soundness of CEGQI]{theorem}{soundnesscegqi}
    \label{thm:soundness-cegqi}
    \rone If \checkimpl terminates with an example $\ucolor{\ex}$, the example $\ucolor{\ex}$ is a valid solution to the existential quantifier in \Cref{eq:alternating-quant-neg}.
    If \checkimpl terminates with $\bot$, there is no example $\ucolor{\ex}$ that satisfies \Cref{eq:alternating-quant-neg}.
    \rtwo If \cpover terminates with an example $\ucolor{\ex}$ and $\phil'$, the example $\ucolor{\ex}$ and $\phil'$ are valid solution to the existential quantifier in \Cref{eq:alternating-quant-neg-with-formula}.
    If \cpover terminates with $\bot$, there is no example $\ucolor{\ex}$ and $\phil'$ that satisfy \Cref{eq:alternating-quant-neg-with-formula}.
\end{restatable}

Because \Cref{alg:CEGIS-Loop} and \ref{alg:CEGQI-Loop-Over} monotonically increases the size of the set $H$, as long as the domain of one of the variables $\ucolor{\ex}$ and $\ecolor{h}$ is finite, both algorithms always terminate.
\begin{restatable}[Completeness of CEGQI]{theorem}{finite-completeness-cegqi}
    \label{thm:finite-completeness-cegqi}
    Suppose at least one of the domains of the variables $\ucolor{\ex}$ or $\ecolor{h}$ is finite.
    If \genex and \checkex are decidable for $\phiprog$ and $\phil$, then \checkimpl always terminates.
    If \genspec and \checkex are decidable for $\phiprog$ and $\phil$, then \cpover always terminates.
\end{restatable}

Therefore, when the domain is finite, the specification synthesis for an existentially quantified query can be solved using only calls with the quantifier free part of the query.

Note that, in the worst case, \cegqi can enumerate the entire domain of $\ecolor{h}$. As we demonstrate in our evaluation, this exhaustive enumeration (which is common for CEG-style algorithms \cite{armandothesis}) is practically rare and a small number of examples are usually sufficient to solve the problem.

%% file: 6implementation.tex
\section{Implementation}
\label{se:implementation}

We implemented our algorithms for solving synthesis problems in the \framework framework in a tool called \name. 
\name is implemented in Java, on top of the \sketch program synthesizer (v.1.7.6)~\cite{DBLP:journals/sttt/Solar-Lezama13}.

Following \Cref{se:framework}, \name takes the following four inputs:
\rone A query $\query$ for which \name is to find \loverproperties or \lunderproperties where each variable in $\query$ should be labeled either as free or existentially quantified.
\rtwo The context-free grammar of the DSL $\lang$ in which properties are to be expressed.
\rthree A piece of code in the \sketch programming language that expresses the concrete semantics of the function symbols in $\query$ and $\lang$.
\rfour The bounded domain of each variable in the query $\query$---i.e., each variable is assigned a range of possible input values.

Take the motivating problem in \Cref{se:example-reachability-over} as an example: input \rone corresponds to the blocks \keyword{Variables} and \keyword{Query} in \Cref{fig:modhash-query}, input \rtwo corresponds to the block \keyword{Language} in \Cref{fig:hash-over-grammar}, and input \rthree are \sketch implementations of \modhash and \isprime---i.e., simple imperative functions. 
Input \rfour is provided via the \keyword{Examples} block in \Cref{fig:modhash-example}, which we will discuss next.

\mypar{From Grammars to \sketch Generators}
As synthesis needs to be performed over properties in the DSL $\lang$, the context-free grammar for $\lang$ is automatically translated to a \sketch generator. 
A \sketch generator is a construct that allows one to describe a recursively defined search space of programs. 
In a generator, one is allowed to use holes (denoted with \texttt{??}) to allow the synthesizer to make choices about what terms to output. 
In our setting, holes are used to select which production is expanded at each node in a recursively defined derivation tree. 
%
%

\name also uses grammars, which in turn are translated to generators, to specify the values each variable in the query $\query$ can assume. 
%
\Cref{fig:modhash-example} shows how the user specifies the bounded domain of each variable for the problem in \Cref{se:example-reachability-over}. 
The nonterminal \texttt{IG} is translated to a generator that can produce an integer from $[-15, 15]$ (the notation \texttt{??(4)} denotes a 4-bit hole), and is used to define the domain for variables \ucolor{\texttt{a}}, \ucolor{\texttt{y}}, and \ecolor{\texttt{x}}.
%
Similarly, the nonterminal \texttt{PosIG} defines the domain of the positive integer variable \ucolor{\texttt{M}} to be the range $[1,16]$.
\name also supports inductive datatypes, e.g. a generator of lists of integers in the range $[-15, 15]$ can be defined as 
\texttt{[ list LG -> nil() | cons(IG, LG) ]}.
\input{code-figure/example}
%
%

     

\mypar{Synthesis Primitives in \sketch}
The primitives \synth, \checkimpl in \Cref{alg:SynthesizeProperty}, and \cpunder in \Cref{alg:SynthesizeUnderProperty} are implemented as calls to the \sketch synthesizer. 
Typically, a \sketch program contains 3 elements: \rone a harness procedure that defines what should be synthesized, \rtwo holes associated with a corresponding generator, and \rthree assertions. 
The harness procedure is the entry point of the \sketch program, and together with the assertion it serves as a specification for what values the holes can assume to form a valid solution to the synthesis problem. 
Multiple harnesses in one \sketch program are also allowed, where the goal of the \sketch synthesizer is to find the same assignment to shared holes that make all assertions pass.
For example, when encoding \synth, each example is implemented as a harness with assertions indicating that it should be positive or negative.
%
Both \cpover in \synthoverproperty (\Cref{alg:SynthesizeProperty}) and \checkimpl in \synthunderproperty (\Cref{alg:SynthesizeUnderProperty}) are implemented using the CEGQI approach described in \Cref{se:cegqi} (\Cref{alg:CEGQI-Loop-Over,alg:CEGIS-Loop}, respectively).
These algorithms are implemented as separate procedures where each call to  \genex and \checkex
only has an existential quantifier, and can thus be implemented as a single call to the \sketch synthesizer.
%

\mypar{Bounds} 
\sketch allows one to provide bounds for recursion and loops to make synthesis tractable. 
In \name, we need to consider two kinds of bounds. 

First, one has to bound the depth of each recursive generator. 
Concretely, this bound means that \name can only support DSLs where the derivation trees have bounded height (\name allows one to specify the bound for each DSL).
As recursive generators are used to produce examples for inductive datatypes---e.g. list---one also has to bound the height of such examples. 
Second, one has to bound how many times a loop can be unrolled/executed. 
We will discuss in \Cref{se:evaluation} what benchmarks are in theory affected by these bounds.
Additionally, these bounds limit DSLs and example domains to a finite size, ensuring that our algorithms in theory terminate (\Cref{thm:finite-completeness}).


\mypar{Timeout} We use a timeout of 20 minutes, after which \name returns the current \loverproperties (or \lunderproperties). Although it might not be the strongest \loverproperties (or weakest \lunderproperties), each individual \loverproperty (or \lunderproperty) is a strongest (or a weakest) one.  


%% file: code-figure/example.tex
\begin{wrapfigure}{r}{0.35\textwidth}
\vspace{-1mm}
    \centering
    \begin{lstlisting}[ tabsize=3, 
    basicstyle= \tt \footnotesize, 
    keywordstyle=\color{black}\bfseries, 
    commentstyle=\color{gray}, 
    xleftmargin=-2.5em, 
    escapeinside=``,
    language = C,
    morekeywords = {Examples},
    emph={a, M, y}, emphstyle=\color{dred},  
    emph={[2]x}, emphstyle={[2]\color{dblue}}, 
    numbers = none
    ]
    Examples { 
        int IG -> ??(4) | -??(4)
            for a, y, x; 
        int PosIG -> ??(4) + 1;
            for M;
    }
\end{lstlisting}
    \vspace{-4mm}
    \caption{The bounded domains of variables in problem \Cref{se:example-reachability-over} }
\label{fig:modhash-example}
\vspace{-4mm}    
\end{wrapfigure}

%% file: 7evaluation.tex
\section{Evaluation}
\label{se:evaluation}

Our evaluation of \name consists of two parts.
The first part consists of simple deterministic and nondeterministic programs (from \cite{park2023specification, Chatterjee2020Polynomial, SVCOMP24}) for which we could manually inspect whether \name computed the correct properties (\S\ref{se:eval-testset}).
The second part showcases the capabilities of \name and the flexibility of the \framework framework through three case studies:  
forward/backward reasoning for incorrectness logic (\S\ref{se:eval-incorrectness}), synthesizing properties of concurrent programs (\S\ref{se:eval-concurrency}), and solving two-player games (\S\ref{se:eval-game}).
The evaluation highlights that the programmable queries and customized DSLs in the \framework framework enable it to express a wide range of problems.
All our case studies involve underapproximation and/or existential quantifiers.
%
For each case study, we describe how we model the problem in \framework, how we collected the benchmarks, how we designed the DSLs, and we present an analysis of the running time and effectiveness of \name and of the quality of synthesized \loverproperties and \lunderproperties. 

We ran all experiments on an Apple M1 8-core CPU with 8GB RAM. Each benchmark was run 3 times (timeout 20 minutes), with different random seeds for the \sketch solver.
\footnote{Our synthesis primitives are nondeterministic---i.e., \sketch can return \textit{any} possible valid counterexample or property. 
The different random seeds will result in the \sketch solver selecting different such examples/formulas. 
\name can therefore generate different sets of properties with different seeds, but as stated in \Cref{thm:bestset}, all best \lconjunctions (or \ldisjunctions) are semantically equivalent.} 
All results in this section are for the median of three runs (by synthesis time).

\subsection{Testing \name with Simple Programs}
\label{se:eval-testset}

We conducted two experiments on simple programs for which we could manually check whether the synthesized properties were sound and strongest/weakest. 

We provide a brief description of each experiment, but details about each problem, running times, concrete DSLs, 
and synthesized properties are discussed in \Cref{App:eval-specmining} (Test set I) and \Cref{App:eval-non-deter} (Test set II).

\mypar{Test Set I: Mining Under-Approximation Specifications}
We used \name to compute \lunderpropertiesp{\langunder} for the 35 programs for which \citet{park2023specification} computed \loverpropertiesp{\langover} using \spyro.
These programs include integer functions, data structure manipulations, and small imperative programs.
To get the dual DSL $\langunder$ of each DSL $\langover$ used in the original benchmarks, we replaced every top-level production of the form $S \to \lor[AP, 0..n] $ with a production $S \to \land[AP, 0..n]$. Note that the queries for all the \spyro benchmarks \textit{do not} contain existential quantifiers.

\name could synthesize properties for 35/35 benchmarks, and guaranteed that all of them were best \lunderpropertiesp{\langunder} (i.e., \cpunder succeeded). 
Each benchmark takes at most 7 minutes (the largest grammar contained $1.48\cdot 10^{13}$ properties). 
Beause the DSL was originally designed for overapproximations, the synthesized underapproximation properties were often trivial. 
For example, for the list reverse function, \name only tells us any singleton list can be obtained as output when providing the same one as input, but provides no properties describing the behavior for lists of lengths greater than 1. 
%
While the synthesized properties were not informative, this simple benchmark set allowed us to test \name's ability to synthesize \lunderpropertiesp{\langunder}.

\mypar{Test set II: Nondeterministic programs}
To test \name on problems involving existential quantifiers, we designed 15 simple imperative programs where nondeterministic values serve as operands or array indices---e.g., a nondeterministic sorting algorithm.
To model the sequence of nondeterministic choices taken by a program, we use an existentially quantified array of values $\ecolor{h}$ in the query. 
Whenever the program execution reaches a non-deterministic command, the command takes the next value of the array $\ecolor{h}$.

\name synthesizes both \loverpropertiesp{\langover} and \lunderpropertiesp{\langunder} for 14/15 benchmarks, taking less than 400 seconds for each benchmark and guaranteeing that the synthesized properties are best ones. 
All benchmarks that terminated for one run terminated for all 3 runs, with fastest and slowest runs differing by at most 2x.
The synthesized properties provide intuition for both the demonic and angelic behaviors of these programs. 
For example, consider a program that nondeterministically swaps elements of an array of length 4 to sort it: \name can tell us which initial arrays \textit{may} be sorted within $n$ swaps (for a given $n$), and also identifies what kind of arrays will never be sorted in less than $n$ swaps.
The one timeout benchmark models a merge sort that computes the number of inverse pairs in an array of nondeterministic values; \name fails due to the complexity of the program, which involves nested recursions and loops.

\subsection{Application 1: Incorrectness Reasoning}
\label{se:eval-incorrectness}

Thanks to the support for both over- and under-approximation, some forms of forward/backward reasoning for both Hoare logic~\cite{hoare69axiomatic} and incorrectness logic~\cite{Peter2019Incorrectness}
can be captured in the \framework framework. 
Because there has been a lot of research and there are many tools on precondition/postcondition inference of Hoare triples, we only discuss the relation between the \framework framework and incorrectness logic in this subsection, along with an evaluation. 
A complete formalization of the relation between the \framework framework and Hoare/incorrectness logic can be found in \Cref{app:relation-to-program-logics}.

\subsubsection{Relation to  Incorrectness Logic}
\label{se:rel-incorrectness}

An \textit{incorrectness triple} $\incortriple{P}{s}{Q}$ consists of a presumption $P$, a statement $s$, and a result $Q$, and it has the following meaning: every final state satisfying $Q$ is reachable by executing program $s$ starting from \textit{some} state that satisfies presumption $P$:
\begin{equation}
\label{eq:incorrectness-validity}
\forall \progstate'.\, Q(\progstate') \Rightarrow \exists \progstate.\, [P(\progstate) \wedge \interp{s}(\progstate, \progstate')]
\end{equation}

\mypar{Forward Reasoning: Weakest Under-approximate Postcondition}
Given a program $s$ and a presumption $P$, the \emph{weakest under-approximate postcondition} $\wpost(s, P)$ is the weakest predicate $Q$ such that the triple $\incortriple{P}{s}{Q}$ holds. 
We use $\wpostlang(s, P)$ to denote the weakest under-approximation postcondition expressible as a \textit{disjunction} of predicates in the DSL $\lang$.
From \Cref{eq:incorrectness-validity}, $\wpostlang(s, P)$ can be obtained by synthesizing weakest \lunderproperties for the query 
$\exists \ecolor{\progstate}.\, P(\ecolor{\progstate}) \wedge \interp{s}(\ecolor{\progstate}, \ucolor{\progstate'})$.

\input{code-figure/remhash}
Let's say we are interested in reasoning about the possible behaviors of the (incorrect) implementation of a modular hash function \remhash shown in \Cref{fig:remhash}, 
where $\remainderop$ is the remainder operator (instead of the modulus).
The $\remainderop$ operator is often misused when implementing a modular operation, as $a \,\remainderop\, b$ may yield a negative output when either $a$ or $b$ is negative.

\input{code-figure/incorrect-property-under} 
A summary of the behaviors of \remhash can be identified by under-approximating the query $\queryrem:=(\existsm \ecolor{\x}.\, \ucolor{\y} = \ucolor{\vara}\ecolor{\x}\,\remainderop\,\ucolor{\varmodulus})$.
From the perspective of incorrectness logic, under-approximating $\queryrem$ corresponds to performing forward reasoning to find
results for $\ucolor{\y}$ when no presumption on $\ecolor{x}$ is given---i.e., the presumption $P(\ecolor{\x})$ is $\tru$.
For capturing under-approximations of the query $\queryrem$, we reuse the DSL $\langunder$ from the example in \Cref{se:example-reachability-under}.
A mutually incomparable set of weakest \lunderpropertiesp{\langunder} for query $\queryrem$ is shown in \Cref{fig:remhash-under-prop},
which shows that \remhash can indeed yield negative values for some choice of parameters $\ucolor{\vara}$ and $\ucolor{\varmodulus}$, as evidenced by the occurrence of a state in both the second and last formulae where $\ucolor{\y}$ is negative.
In other words, we recognize that some choices of input value can result in incorrect outcomes---i.e., negative numbers---but we do not know which ones. 
As we will show next, the inputs that lead to incorrect behaviors can be identified by backward reasoning.

\input{code-figure/group-2-2-2}
\mypar{Backward Reasoning: Weakest Possible Precondition} 
Surprisingly, backward predicate transformers for incorrectness logic do not always exist because valid presumptions may not exist.
For example, there is no predicate $P$ making the triple $\incortriple{P}{\y = \modhasharg{\vara}{\varmodulus}{\x}}{\y = -1}$ true because no values of $\vara$, $\varmodulus$ and $\x$ satisfy $\modhasharg{\vara}{\varmodulus}{\x} = -1$.
To address this shortcoming \citet{Peter2019Incorrectness} suggests using the weakest possible precondition $\wppre(s, Q)$, which is termed by \citet{hoare78properties} as ``possible correctness''.
Intuitively, $\wppre(s, Q)$ captures the set of initial states from which it is \textit{possible} to execute $s$ and terminate in a state that satisfies $Q$.
When considering the \remhash function from \Cref{fig:remhash}, if $Q$ encodes that the output is negative, $\wppre(s, Q)$ will show which input possibly leads to a negative output.
Formally, $\wppre(s, Q)$ is the weakest $P$ satisfying
\begin{equation}
\label{eq:hoare-wpp-cond}
\forall \ucolor{\progstate}.\, P(\ucolor{\progstate}) \Rightarrow [\exists \ecolor{\progstate'}.\, Q(\ecolor{\progstate'}) \wedge \interp{s}(\ucolor{\progstate}, \ecolor{\progstate'})]
\end{equation}

Note that $P = \wppre(s, Q)$ forms neither a Hoare nor an incorrectness triple with the program $s$ and the postcondition $Q$.
As proposed by \citeauthor{Peter2019Incorrectness}, we can use $P \hspace{-1pt}=\hspace{-1pt} \wppre(s,\hspace{-1pt} Q)$ to compute a new postcondition $Q' = \wpost(s, P)$ and obtain a valid incorrectness triple $\incortriple{P}{s}{Q'}$.

We use $\wpprelang(s, Q)$ to denote the weakest possible precondition expressible as a \textit{disjunction} of predicates in the DSL $\lang$. From \Cref{eq:hoare-wpp-cond}, $\wpprelang(s, Q)$ can be obtained by synthesizing weakest \lunderproperties of the query 
$\exists \ecolor{\progstate'}.\, Q(\ecolor{\progstate'}) \wedge \interp{s}(\ucolor{\progstate}, \ecolor{\progstate'})$.

For the \remhash function, we have shown the existence of a bug through forward reasoning---i.e., the output $\ecolor{\y}$ can be negative.
Now we want to compute the weakest possible precondition (expressible in a DSL) that leads to the error state $Q(\ecolor{\y}) := \ecolor{\y} < 0$.
Looking at \Cref{eq:hoare-wpp-cond}, 
we can spell out that a weakest possible precondition of $Q(\ecolor{\y})$ for $\remhash$ is a weakest implicant of the formula $\exists \ecolor{\y}.\, [\ecolor{\y} = \ucolor{\vara}\ucolor{\x}\,\remainderop\,\ucolor{\varmodulus} \wedge Q(\ecolor{\y})]$; one can provide to \framework the corresponding query $\queryremwpp$ in \Cref{fig:remhash-wpp-query}.

To capture implicants of the query $\queryremwpp$, we define the DSL $\langunder_{\wppre}$ in \Cref{fig:remhash-wpp-grammar} by substituting every occurrence of $\y$ in $\langunder$ with $\x$.
An incomparable set of weakest \lunderpropertiesp{\langunder_\wppre} for query $\queryremwpp$ is shown in \Cref{fig:remhash-wpp-prop},
where each formula states sufficient conditions under which $\ucolor{\vara}\ucolor{\x}\,\remainderop\,\ucolor{\varmodulus}$ produces a negative output---i.e., when either $\ucolor{\vara}$ or $\ucolor{\x}$ falls within the interval $(-\ucolor{\varmodulus}, 0)$ and the other falls within the interval $(0, \ucolor{\varmodulus})$.

\subsubsection{Evaluation on Incorrectness Reasoning}    
We collected a total of 14 benchmarks: 
\rone the 2 example problems \remwupo and \remwpp from \Cref{se:rel-incorrectness}.
\rtwo 6 simple illustrative examples from the incorrectness logic paper~\cite{Peter2019Incorrectness}, and 
\rthree 6 more complicated problems we crafted to illustrate how \name's handling of incorrectness reasoning differs from incorrectness logic.
Among these benchmarks, 7 are about \lwpo and the other 7 are about \lwpp.
It takes \name less than 4 seconds to solve each benchmark from~\cite{Peter2019Incorrectness} and less than 50 seconds to solve each benchmark we crafted. Evaluation details are shown in \Cref{tab:benchmarks-main}.

\input{wraptable}

\mypar{Analysis of Benchmarks from~\cite{Peter2019Incorrectness}} 
We collected all the 3 triples $[P]s[Q]$ from the examples used in \cite{Peter2019Incorrectness} where $s$ is a nondeterministic program.
However, in these triples, there was no guarantee that $Q$ was the weakest under-approximate postcondition of $P$, or $P$ was the weakest possible precondition of $Q$.
We used \name to synthesize  $\wpostlang(s, P)$ and $\wpprelang(s, Q)$  (the DSL $\lang$  contained the same primitives appearing in the examples in \cite{Peter2019Incorrectness}), thus 3+3=6 benchmarks.
We examined that each of 3 synthesized $\wpostlang(s, P)$ by \name was indeed a subset of $\wpost(s, P)$, and for 2 cases the two were equal. 
For the one that is not equal, $\wpost(s, P)$ is the set of all perfect squares numbers, whereas $\wpostlang(s, P)$ is the perfect squares numbers lower than a bound (this difference was due to our query limiting the sample space of each variable).
%
%
%
The \numILPaperWPP  $\wpprelang(s, Q)$ synthesized by \name are equal to $\wppre(s, Q)$.

\mypar{More Complex Benchmarks}
The 6 more complex benchmarks for which we performed incorrectness are \arith{1-wupo}, \arith{1-wpp}, \arith{2-wupo}, \arith{2-wpp}, \wppHash{}, and \coin.


The benchmarks \arith{1} and \arith{2} model two arithmetic functions  \inlinef{$x' = \texttt{ite}(\ecolor{h_0}, x,-x)$} and \inlinef{$x' = (\ecolor{h_1} + 1)\cdot x + \ecolor{h_2}$}, where each $\ecolor{h_i} \in \{0, 1\}$ is a nondeterministic value. 
For both cases, we set $a\le x \le b$ as a precondition $P$ (or $a\le x' \le b$ as a postcondition $Q$) to synthesize $\wpprelang(s, Q)$  (or $\wpostlang(s, P)$), and thus get 4 benchmarks in total.
To use \name, we need to mark $x$ as existentially quantified variables when synthesizing $\wpostlang(s, P)$, whereas mark $x'$ as existentially quantified variables when synthesizing $\wpprelang(s, P)$.
Given a DSL containing basic arithmetic and comparison operators, \name synthesizes $\wpprelang(s,Q)$ and $\wpostlang(s,P)$ that are equal to $\wppre(s,Q)$ and $\wpost(s,P)$.

For example, to synthesize $\wpprelang(s,Q)$ for \arith{1}, one can construct a query \inlinef{$\exists \ecolor{x'}, \ecolor{h_0}. ~ \ecolor{x'} = \texttt{ite}(\ecolor{h_0}, \ucolor{x},-\ucolor{x}) \land \ucolor{a}\le \ecolor{x'} \le \ucolor{b}$} and supply the DSL in \Cref{fig:arit1-grammar}. 
\name will synthesize the \lunderproperties $\{ -\ucolor{b} \le \ucolor{x}, \ucolor{x}  \le -\ucolor{a}, \ucolor{a}\le \ucolor{x}, \ucolor{x}  \le \ucolor{b}\}$.


\input{code-figure/arit1-grammar}
We briefly summarize the findings on other benchmarks.
The \exname{coin} benchmark models the values one can produce using two coins that have co-prime denominations; \name can identify a lower bound above which all possible values can be produced using these coins.
The \wppHash{} benchmark models a parametric hash function; \name can synthesize the condition that possibly causes a hash collision.
More details are discussed in \Cref{App:eval-incorrectness}.

\subsection{Application 2: Reasoning about Concurrent Programs} 
\label{se:eval-concurrency}

We show how \name can be used to reason about bugs in concurrent programs by considering \numNondConc variants of \numConcProb concurrency problems by \citet{ArpaciDusseau23-Book} (2 problems related to deadlocks, and 1 to race conditions), and one benchmark \obdet that requires synthesizing hyperproperties. 
Similar to how we model nondeterminism, we introduce an array $\ecolor{h}$ to represent the order in which threads are scheduled. 
%

In the \philosopher benchmark, we show how \name can synthesize conditions under which deadlock can be reached or avoided for the dining-philosophers problem, where $N$ processes arranged in a circle contend $N$ resources that are shared by neighboring processes.
A deadlock happens when no process can access both of their Left and Right resources indefinitely.
\name models this problem with a query $\exists \ecolor{h}.~\ucolor{dl} = schedule(\ucolor{o_1}, \cdots, \ucolor{o_N}, \ecolor{h})$, where
$\ucolor{o_i}\in \{L,R\}$ indicates which resource the process $i$ always takes first; $\ucolor{dl}$ denotes that a deadlock has happened. 

For the case involving three processes/philosophers ($N=3$), when given a DSL $\langover$  in \Cref{fig:philosopher-grammar} that contained predicates of the form $\ucolor{o_i} = \{L|R\}$, \name synthesizes the \loverpropertiesp{\langover} in \Cref{fig:philosopher-property-over}, which state that deadlock can be prevented by having two of the processes disagree on their fork choice.
For the same $N$, and a dual DSL $\langunder$ in \Cref{fig:philosopher-grammar}, \name synthesizes the \lunderpropertiesp{\langunder} in \Cref{fig:philosopher-property-under}, which 
exactly characterize the two cases that lead to a deadlock (first two properties) and also capture that there always exists an execution that does not lead to a deadlock (last property).



Whereas similar tools for concurrent programs only deal with properties over a single schedule~\cite{Wang2017concurrentbugs, Kumar2011concurrentmining}, 
the next benchmark \obdet demonstrates how \name can also synthesize properties that involve multiple schedules (i.e., hyperproperties).
The \obdet benchmark models two threads 
\texttt{p\_o := p\_i + s | p\_o := p\_i - s} where \texttt{s} is a secret variable and \texttt{p\_o} and \texttt{p\_i} are public variables.
We say that the system is observational deterministic~\cite{zdancewic2003determinism} if the observations made by a public observer (i.e., one that can only observe \texttt{p\_o} and \texttt{p\_i}) are deterministic, regardless of scheduling orders and secret input data (i.e., the values of the variables \texttt{s}).
The query in the \obdet benchmark is \inlinef{$\exists \ecolor{h_1}, \ecolor{h_2}. ~ \ucolor{p_{o1}} = schedule(\ucolor{p_{i1}}, \ucolor{s_1}, \ecolor{h_1}) \land \ucolor{p_{o2}} = schedule(\ucolor{p_{i2}}, \ucolor{s_2}, \ecolor{h_2})$} where  $\ecolor{h_1}$ and  $\ecolor{h_2}$ describe 2 different schedules.
%
Given a DSL $\lang$ containing only public variables, \name synthesizes the \loverproperty $\ucolor{p_{i1}} = \ucolor{p_{i2}} \Rightarrow \ucolor{p_{o1}} = \ucolor{p_{o2}}$, ensuring the system is observationally deterministic.



\input{code-figure/group-7-3}
Each of the 4 \resource{} benchmarks describes a simple resource allocator; \name synthesizes properties describing the minimum number of resources that must (or may) cause a deadlock. Each of the 3 \race{} benchmarks describes two threads; \name can discover the necessary (or sufficient) ways to place a critical section to prevent race conditions. 
Details are shown in \Cref{App:eval-concurrency}.


\subsection{Application 3: Solving Two-Player Games} 
\label{se:eval-game}

In this section, we show how \name can be used to synthesize generalized strategies for solving two-player games.
In particular, by carefully designing the DSL, \name can synthesize sets of winning strategies expressed in a symbolic form, rather than a single concrete strategy.

%
We illustrate the idea using an example by \citet{Bloem2018GraphGames}, called \rggame (for request/grant).
The two players take on the roles of client and server, and in each round, the server decides whether to grant ($g$) or not ($\bar g$) the request for that round, and then the client decides whether to send ($r$) or not ($\bar r$) a request in that round. 
To win the game, the server must grant every request in the same or next round.
\citet{Bloem2018GraphGames} show the server player can be in 3 possible states: \rone $q_0$: no ungranted request \rtwo $q_1$: an ungranted request in the last round \rthree $q_2$: ungranted requests 2 or more rounds ago. 
The server should prevent entering state $q_2$.

One of the winning strategies from the server side is to always grant on both state $q_0$ and $q_1$. We denote such a strategy as  $\strategyA[q_0] = g \land  \strategyA[q_1] = g$---i.e.,  $\strategyA[q] = a$ denotes that strategy $\strategyA$ chooses action $a$ when in state $q$. 
We can find winning strategies by modeling the \rggame~ game as a query \inlinef{$\exists \strategyB .~\ucolor{w} = \playgame(\strategyA, \strategyB)$}, where the client's strategy $\strategyB$ is existentially quantified and $\playgame(\strategyA, \strategyB)$~is the game controller that takes the strategy of both players and produces a Boolean value $\ucolor{w}$ denoting whether the server wins after playing the game.
Note that the way we model 2-player games can also be extended to multi-player games by introducing a set of opponent strategies $\{\strategyB_1, \cdots, \strategyB_k\}$.

The generality of the \framework framework allows \name to solve two-player games using the following queries: 
\rone \textit{Must-win strategy}: what strategy $\strategyA$ can guarantee a win for \textit{any} strategy $\strategyB$ (\Cref{eq:game-formalization-over})? 
\begin{equation}
\label{eq:game-formalization-over}
    \forall \ucolor{\alpha}, \ucolor{w}. ~\big(\exists \strategyB.~ \ucolor{w} = \playgame(\strategyA, \strategyB)\big) \Rightarrow \big(P_{must}(\strategyA) \Rightarrow \ucolor{w} = \gtrue \big)
\end{equation}
 and
\rtwo \textit{May-win strategy}: what strategy $\strategyA$ can win for \textit{at least one} strategy $\strategyB$ (\Cref{eq:game-formalization-under})? 
\begin{equation}
\label{eq:game-formalization-under}
    \forall \ucolor{\alpha}, \ucolor{w}.~ \big(P_{may}(\strategyA) \land \ucolor{w} = \gtrue \big) \Rightarrow \big(\exists \strategyB.~ \ucolor{w} = \playgame(\strategyA, \strategyB)\big)
\end{equation}

Looking at \Cref{eq:game-formalization-over}, if we provide a DSL $\langover$ that expresses formulas in the form \inlinef{$P_{must}(\strategyA) \Rightarrow \ucolor{w} = \gtrue$}, we can extract the must-win strategy in the $P_{must}$ part of the synthesized formulas. 
By replacing $\gtrue$ with $\gfalse$ we can get the must-lose strategy.
For the \rggame game, \name synthesized the following \loverpropertiesp{\langover}, which tells us that the server will always win if they grant requests in either of states $q_0$ and $q_1$---i.e., \name finds ``a set of'' winning strategies.
\begin{equation}
\label{eq:rq-over-result}
\begin{array}{c@{\hspace{4.0ex}}c}
    \strategyA[q_0] = g \Rightarrow \ucolor{w} = \gtrue &
    \strategyA[q_1] = g \Rightarrow \ucolor{w} = \gtrue
\end{array}
\end{equation}

When provided with the dual DSL $\langunder$ \name also synthesized the following \lunderpropertiesp{\langunder}:
\begin{equation}
\label{eq:rq-under-result}
\begin{array}{c@{\hspace{4.0ex}}c}
    \strategyA[q_0] = \bar g \land \strategyA[q_1] = \bar g \land \ucolor{w} = \gfalse &
    \ucolor{w} = \gtrue
\end{array}
\end{equation}
The first \lunderpropertyp{\langunder} states that the server may lose if they do not grant requests at both states  $q_0$ and $q_1$, whereas the second  \lunderpropertyp{\langunder} states that whatever strategy the server uses there exists a strategy of the requester (i.e., the one that never issues requests) that causes the server to win.

\mypar{Other benchmarks} 
We consider a total of 5 benchmarks:
\rggame (discussed above), \nimgame{2} (the Nim game), and \tempgame (a temperature controller), which are adapted from linear reachability games by \citet{Frazan2018Strategy}\footnote{All other games studied by \citeauthor{Frazan2018Strategy} cannot be modeled in \name due to the restricted features of \sketch languages, such as limited support to floating point numbers.}, and \numgame{1} and \numgame{2}, which are games designed by us in which two players manipulate an integer where one player's goal is to keep the integer in a certain range.
Because of the implementation bounds discussed in section \ref{se:implementation}, we stipulate that player 1 (typically the player that needs to stay in safe states) wins, after a finite number (we set as 15) of rounds of play.

It takes \name less than 85 seconds to synthesize must/may strategies for each benchmark. 
Compared to the work by \citet{Frazan2018Strategy}, \name synthesizes \rone not only \emph{must}- but also \emph{may}-strategies, \rtwo properties on desired strategies instead of a concrete strategy, and \rthree general strategies that work for games with parameters (e.g. the initial number of pebbles in \nimgame{2}).

Details of DSL design and synthesized properties of benchmarks are provided in \Cref{App:eval-game-benchmarks}.

%% file: code-figure/remhash.tex
\begin{wrapfigure}{r}{0.33\textwidth}
\vspace{-3mm}
    \centering
    \begin{lstlisting}[ tabsize=3, 
    basicstyle= \tt \footnotesize, 
    keywordstyle=\color{black}\bfseries, 
    commentstyle=\color{gray}, 
    xleftmargin=-2.5em, 
    escapeinside=``,
    language = C,
    morekeywords = {mod},
    numbers = none
    ]
    int remhash (int a, M, x) {
        return a * x % M;
    }
\end{lstlisting}
\vspace{-6mm}
\caption{\texttt{remhash} function}
\label{fig:remhash}
\vspace{-4mm}    
\end{wrapfigure}

%% file: code-figure/incorrect-property-under.tex
\begin{wrapfigure}{r}{0.37\textwidth}
\vspace{-8mm}
    \centering
    \begin{lstlisting}[ tabsize=3, 
    basicstyle= \tt \footnotesize, 
    keywordstyle=\color{black}\bfseries, 
    commentstyle=\color{gray}, 
    xleftmargin=-2.5em, 
    escapeinside=``,
    language = C,
    morekeywords = {Language, bool},
    emph={a, M, y}, emphstyle=\color{dred},  
    emph={[2]x}, emphstyle={[2]\color{dblue}}, 
    numbers = none
    ]
    I1: y == 0
    I2: -M <= a /\ a < M /\ a == y
    I3: -M <= y /\ y < M 
        /\ -M < a /\ a < M
        /\ a != 0 /\ isPrime(M)
    
\end{lstlisting}
\vspace{-4mm}
\caption{Synthesized \lunderpropertiesp{\langunder}}
\label{fig:remhash-under-prop}
\vspace{-2mm}    
\end{wrapfigure}

%% file: code-figure/group-2-2-2.tex
\begin{figure}[t!]
    \centering
    
    \begin{subfigure}{0.28\textwidth}
        \centering
        \input{code-figure/query-wpp}
        \vspace{-10pt}
        \caption{$\queryremwpp$}
        \label{fig:remhash-wpp-query}
    \end{subfigure}
    \hspace{1em}\vrule\hspace{0.5em}
    \begin{subfigure}{0.31\textwidth}
        \centering
        \input{code-figure/dsl-wpp}
        \caption{$\langunder_{\wppre}$}
        \label{fig:remhash-wpp-grammar}
    \end{subfigure}
    \hspace{1em}\vrule\hspace{0.5em}
    \begin{subfigure}{0.28\textwidth}
        \centering
        \input{code-figure/incorrect-property-wpp}
        \caption{\lunderpropertiesp{\langunder_\wppre}}
        \label{fig:remhash-wpp-prop}
    \end{subfigure}
    \caption{(a) A query $\queryremwpp$ that allows identifying \textit{weakest possible preconditions} that cause a bug in the \remhash function (i.e., outputing a negative number).
    (b) A DSL $\langunder_{\wppre}$ for expressing weakest possible preconditions.
    (c) \lunderpropertiesp{\langunder_\wppre} synthesized by \name.}
\end{figure}

%% file: code-figure/query-wpp.tex
    \begin{lstlisting}[ tabsize=3, 
    basicstyle= \tt \footnotesize, 
    keywordstyle=\color{black}\bfseries, 
    commentstyle=\color{gray}, 
    xleftmargin=-2.5em, 
    escapeinside=``,
    language = C,
    morekeywords = {Variables, Query, exist },
    emph={a, M, x}, emphstyle=\color{dred},  
    emph={[2]y}, emphstyle={[2]\color{dblue}}, 
    numbers = none
    ]
    Variables { 
        int a, M, x; 
        exist int y; 
    }
    Query { 
        y = remhash(a, M, x);
        y < 0;
    }
\end{lstlisting}

%% file: code-figure/dsl-wpp.tex
\begin{lstlisting}[ tabsize=3, 
    basicstyle= \tt \footnotesize, 
    keywordstyle=\color{black}\bfseries, 
    commentstyle=\color{gray}, 
    xleftmargin=-2.5em, 
    escapeinside=``,
    language = C,
    morekeywords = {Language, bool},
    emph={a, M, x}, emphstyle=\color{dred},  
    emph={[2]y}, emphstyle={[2]\color{dblue}}, 
    numbers = none
    ]
    Language { 
        C -> /\[AP,0..6];
        AP -> I {<=|<|==|!=} I
            | isPrime(M) 
            | !isPrime(M)
        I -> 0 | a | x | M | -M
    }
\end{lstlisting}

%% file: code-figure/incorrect-property-wpp.tex
    \begin{lstlisting}[ tabsize=3, 
    basicstyle= \tt \footnotesize, 
    keywordstyle=\color{black}\bfseries, 
    commentstyle=\color{gray}, 
    xleftmargin=-2.5em, 
    escapeinside=``,
    language = C,
    morekeywords = {Language, bool},
    emph={a, M, x}, emphstyle=\color{dred},  
    emph={[2]y}, emphstyle={[2]\color{dblue}}, 
    numbers = none
    ]
    I1: -M < x /\ x < 0
        /\ 0 < a /\ a < M
        /\ isPrime(M)
    I2: 0 < x /\ x < M
        /\ -M < a /\ a < 0
        /\ isPrime(M)
\end{lstlisting}

%% file: wraptable.tex
\begin{wraptable}{r}{0.40\textwidth}
\vspace{-12pt}
\caption{
Applications 1 to 3. 
LoC is the number of lines of \sketch code used to write the semantics of programs and operators.
$|\exists|$ is the size of the domain of the existentially quantified variables.  
%
\#P and T(s) are the number of properties and synthesis time for both \loverproperties and \lunderproperties.
Incorrectness reasoning does not require synthesizing \loverproperties.
}
\label{tab:benchmarks-main}

\vspace{-10pt}
{\footnotesize
\setlength{\tabcolsep}{2pt}
\renewcommand{\arraystretch}{0.95}
\begin{tabular}{ccrrrrrr} 
\toprule[.1em]
\multicolumn{2}{c}{\multirow{2}{*}[-0.4ex]{Problem}}
& \multicolumn{1}{c}{\multirow{2}{*}{LoC}}
& \multicolumn{1}{c}{\multirow{2}{*}{$|\exists|$}}
 & \multicolumn{2}{c}{$\lang$-cons.} & \multicolumn{2}{c}{$\lang$-impl.}\\
    \cmidrule{5-8}
      & & & & \#P & T(s) & \#P & T(s) \\
\midrule[.1em]

\parbox[t]{2mm}{\multirow{14}{*}{\rotatebox[origin=c]{90}{Incorrectness}}}  
     & \remwupo    & 10 & 32     & /  & /  & 3  & 27.05  \\
     & \remwpp    & 14 & 32  & /  & /  & 2  & 26.27  \\
     & \incwupo{1}    & 13 & 256     & /  & /  & 3  & 1.10  \\
     & \incwupo{2}    & 28 & 64  & /  & /  & 9  & 2.15  \\
     & \incwupo{3}    & 16 & 2    & /  & /  & 2  & 0.41  \\
     & \incwpp{1}     & 12 & 256   & /  & /  & 2  & 0.42  \\
     & \incwpp{2}     & 33 & \numf{4096} & /  & /  & 7  & 6.01  \\
     & \incwpp{3}     & 21 & 2    & /  & /  & 1  & 0.15  \\
     & \arith{1wupo} & 8  & 32   & /  & /  & 2  & 4.03  \\
     & \arith{2wupo} & 8  & 64 & /  & /  & 3  & 2.96  \\
     & \arith{1wpp}  & 8  & 4    & /  & /  & 2  & 3.76  \\
     & \arith{2wpp}  & 8  & 8    & /  & /  & 5  & 62.52  \\
     & \wppHash{}     & 87 & 16   & /  & /  & 3  & 15.94 \\
     & \coin          & 23 & \numf{1024} & /  & /  & 1  & 20.33 \\

\cmidrule{1-8}

 \parbox[t]{2mm}{\multirow{8}{*}{\rotatebox[origin=c]{90}{Concurrency}}}    
    & \philosopher  & 79  & ${\sim}10^{6}$ & 4 & 11.85 & 3 & 6.15 \\
    & \race{1}      & 82  & \numf{64}                & 1 & 0.71 & 3 & 0.85 \\
    & \race{2}      & 86  & \numf{1024}              & 1 & 2.08 & 3 & 2.06 \\
    & \race{3}      & 88  & \numf{4096}              & 1 & 18.79 & 5 & 4.69 \\
    & \resource{1}  & 81  & \numf{16}                & 4 & 2.73 & 4 & 1.67 \\
    & \resource{2}  & 85  & \numf{256}               & 4 & 5.08 & 4 & 5.37 \\
    & \resource{3}  & 114 & ${\sim}10^{6}$            & 4 & 51.83 & 4 & 38.31 \\
    & \resource{4}  & 96  & ${\sim}10^{5}$             & 6 & 145.61 & 6 & 81.41 \\
    & \obdet  & 52  & \numf{1057}              & 3 & 2.11 & / & / \\

\cmidrule{1-8}

 \parbox[t]{2mm}{\multirow{5}{*}{\rotatebox[origin=c]{90}{Game}}}  
    & \numgame{1}  & 47 & 8 & 10 & 16.13 & 19 & 13.99 \\
    & \numgame{2}  & 47 & 32 & 25 & 21.32 & 15 &  8.54  \\
    & \rggame      & 29 & 32 & 2 & 0.38 & 2 &  0.49 \\
    & \nimgame{2}  & 59 & ${\sim}10^{22}$& 2 & 99.23 & 4 &  22.19  \\
    & \tempgame    & 34 & 120 & 10 & 125.24 & 10 &  74.65  \\

\bottomrule[.1em]
\bottomrule[.1em]
\end{tabular}
}
\vspace{-16pt}
\end{wraptable}

%% file: code-figure/arit1-grammar.tex
\begin{wrapfigure}{r}{0.27\textwidth}  
\centering
\vspace{-5mm}
\begin{lstlisting}[ tabsize=3, 
    basicstyle= \tt \footnotesize, 
    keywordstyle=\color{black}\bfseries, 
    commentstyle=\color{gray}, 
    xleftmargin=-2.5em, 
    escapeinside=``,
    morekeywords = {Language, bool},
    emph={a, b, x}, emphstyle=\color{dred},  
    numbers = none
    ]
    C -> /\[AP,0..5];
    AP -> N {<=|<|==|!=} N
    N ->  N' | -N'
    N' -> 0 | a | b | x 
\end{lstlisting}
\vspace{-2mm}
\caption{The DSL for \arith{1}}
\label{fig:arit1-grammar}
\vspace{-5mm}
\end{wrapfigure}

%% file: code-figure/group-7-3.tex
\begin{figure}[t!]
    \centering
    \begin{subfigure}{0.28\textwidth}
        \centering
        \input{code-figure/philosopher-grammar}
        \caption{The DSLs $\langover$ (rooted at \texttt{O}) and $\langunder$ (rooted at \texttt{U}) 
        }
        \label{fig:philosopher-grammar}
    \end{subfigure}
    \hspace{1em}\vrule\hspace{0.5em}
    \begin{subfigure}{0.3\textwidth}
        \centering
        \input{code-figure/philosopher-property-over}
        \caption{Synthesized \loverpropertiesp{\langover}}
        \label{fig:philosopher-property-over}
    \end{subfigure}
    \hspace{1em}\vrule\hspace{0.5em}
    \begin{subfigure}{0.28\textwidth}
        \centering
        \input{code-figure/philosopher-property-under}
        \caption{Synthesized \lunderpropertiesp{\langunder}}
        \label{fig:philosopher-property-under}
    \end{subfigure}
    \caption{The DSLs and synthesized properties for the \philosopher benchmark.}
\end{figure}

%% file: code-figure/philosopher-grammar.tex
\begin{lstlisting}[ tabsize=3, 
    basicstyle= \tt \footnotesize, 
    keywordstyle=\color{black}\bfseries, 
    commentstyle=\color{gray}, 
    xleftmargin=-2.5em, 
    escapeinside=``,
    morekeywords = {Language, bool},
    emph={o1, oN, dl}, emphstyle=\color{dred},  
    numbers = none
    ]
    O -> /\[AP,0..5] => D
    U -> /\[AP,0..5] /\ D
    AP -> I == 'L' | I == 'R'
    I -> o1 | ... | oN
    D -> dl | !dl
\end{lstlisting}

%% file: code-figure/philosopher-property-over.tex
\begin{lstlisting}[ tabsize=3, 
    basicstyle= \tt \footnotesize, 
    keywordstyle=\color{black}\bfseries, 
    commentstyle=\color{gray}, 
    xleftmargin=-2.5em, 
    escapeinside=``,
    morekeywords = {Language, bool},
    emph={o1, o2, o3, dl}, emphstyle=\color{dred},  
    emph={[2]y}, emphstyle={[2]\color{dblue}}, 
    numbers = none
    ]
    C1: (o1 == L /\ o2 == R) 
        => !dl
    C2: (o2 == L /\ o3 == R) 
        => !dl
    C3: (o3 == L /\ o1 == R) 
        => !dl
\end{lstlisting}

%% file: code-figure/philosopher-property-under.tex
\begin{lstlisting}[ tabsize=3, 
    basicstyle= \tt \footnotesize, 
    keywordstyle=\color{black}\bfseries, 
    commentstyle=\color{gray}, 
    xleftmargin=-2.5em, 
    escapeinside=``,
    morekeywords = {Language, bool},
    emph={o1, o2, o3, dl}, emphstyle=\color{dred},  
    emph={[2]y}, emphstyle={[2]\color{dblue}}, 
    numbers = none
    ]
    I1: o1 == L /\ o2 == L 
        /\ o3 == L /\ dl
    I2: o1 == R /\ o2 == R 
        /\ o3 == R /\ dl
    I3: !dl
\end{lstlisting}

%% file: 8related_work.tex
\section{Related Work}
\label{se:related-work}


\mypar{Abstract Interpretation}
Many static program-analysis and verification techniques represent large program state spaces symbolically as predicates, using \emph{abstract interpretation}~\cite{DBLP:conf/popl/CousotC77}.
While the majority of works on abstract interpretation has been focused on over-approximation, it can also be used to describe under-approximations of the program behavior~\cite{grumberg2005proof, schmidt2007calculus, bruni2022repair}.
In particular, the best \lconjunction synthesis problem is an instance of \emph{strongest-consequence problem}~\cite{DBLP:conf/vmcai/RepsT16}.
Given a formula $\phil$ in logic $\lang_1$ (with interpretation $\interp{\cdot}_1$), the goal of strongest-consequence problem is to determine the strongest formula $\psi$ that is expressible in a different logic $\lang_2$ (with interpretation $\interp{\cdot}_2$) such that $\interp{\phil}_1 \subseteq \interp{\psi}_2$.
One existing technique to solve this problem identifies a chain of weaker implicants until one becomes a consequence of $\phil$~\cite{DBLP:conf/vmcai/RepsSY04},
whereas other techniques take the opposite direction, identifying a chain of stronger consequences~\cite{DBLP:conf/cav/ThakurR12, DBLP:conf/sas/ThakurER12}.
Unlike these approaches, our framework \framework (like \spyro~\cite{park2023specification}) supports a customizable DSL, avoiding requirements to perform operations on elements within the DSL $\lang$, such as join~\cite{DBLP:conf/cav/ThakurR12}.
%
The ability to modify the DSL is what makes the \framework framework applicable to many domains.

\mypar{Best $\lang$-term Synthesis} 
The idea of synthesizing a ``best'' term from a customizable DSL $\lang$ was first proposed by \citet{DBLP:journals/pacmpl/KalitaMDRR22}, where the goal was to synthesize a most-precise abstract transformer for a given abstract domain.
\citet{park2023specification} generalized the idea and introduced the setting to define and solve the problem of synthesizing best \lconjunctions.
In these work, the ``best'' term should be \rone \textit{sound}: it is a valid approximation to the best transformer in \citet{DBLP:journals/pacmpl/KalitaMDRR22} or the semantics of query in \citet{park2023specification}, and \rtwo \textit{precise}: it is minimal w.r.t. a preorder defined on $\lang$.

The \framework framework takes a step further: it further generalizes the queries to allow existential quantifiers and introduces the problem of synthesizing weakest \lunderproperties and best \ldisjunctions.
Logically, the \framework framework subsumes both \spyro and the work by \citeauthor{DBLP:journals/pacmpl/KalitaMDRR22}.

At the algorithmic level, the tools solving the above problems all use two kinds of examples for synthesis, where one is treated as hard constraints to guarantee soundness and the other one is treated as soft constraints to guarantee precision.
\spyro improved the algorithm by \citeauthor{DBLP:journals/pacmpl/KalitaMDRR22} by introducing the idea of freezing examples, thus avoiding the need for a synthesizer with hard and soft constraints.
The CEGQI algorithm we present \Cref{se:cegqi} is a new approach that is not present in the aforementioned works as none of them supports existential quantifiers in their queries.

\mypar{Program Logic}
Hoare~\cite{hoare69axiomatic} and incorrectness logic~\cite{Peter2019Incorrectness,edsko2011reverse} can reason about program properties through preconditions and postconditions. 
If one treats the DSL $\lang$ as an assertion language, the problems of computing \textit{strongest postcondition}~\cite{dijkstra1990predicate} and \textit{weakest liberal precondition}~\cite{dijkstra1976discipline} in Hoare logic, and \textit{weakest under-approximation postcondition} and \textit{weakest possible precondition}~\cite{hoare78properties} in incorrectness logic, can be expressed within the \framework framework 
(see \Cref{app:relation-to-program-logics}).

One key distinction between our approach and the one used in automating computing the above operations in program logics is that in the \framework framework, one can specify what DSL $\lang$ they want their properties to be expressed in.
In contrast, the properties produced automatically for, e.g., weakest possible preconditions in incorrectness logic, are the results of syntactic rewrites that often result in complex properties with potentially many quantifiers.

\mypar{Invariant inference}
Many data-driven, CEGIS-style algorithms can infer program invariants---e.g., Elrond~\cite{DBLP:journals/pacmpl/ZhouDDJ21},  abductive inference \cite{dillig2012automated}, ICE-learning \cite{DBLP:conf/cav/0001LMN14}, LoopInvGen \cite{DBLP:conf/pldi/PadhiSM16}, Hanoi~\cite{hanoi}, and Data-Driven CHC Solving \cite{DBLP:conf/pldi/ZhuMJ18}.
Dynamic techniques like Daikon~\cite{DBLP:journals/tse/ErnstCGN01,DBLP:journals/scp/ErnstPGMPTX07}, QuickSpec~\cite{smallbone2017quick} and Precis \cite{DBLP:journals/pacmpl/AstorgaSDWMX21} can also synthesize invariants through program traces or random tests.
The \framework framework differs from the above works in three key ways: \rone The language $\lang$ is customizable and is not limited to a set of predefined predicates, and thus the \framework framework can be used in a domain-agnostic way (as showcased by the many applications presented in \Cref{se:evaluation}); \rtwo the \framework framework supports both over-approximated and under-approximated reasoning, and \rthree the properties synthesized by \framework are \textit{provably sound strongest \loverproperties and sound weakest \lunderproperties}.

\mypar{Quantifier Elimination}
Many algorithms \cite{Dillig2013abductive, Dillig2011container, Calvanese2020interpolant, Komuravelli2014modelcheck} are built on abductive inference, specifically, approximate quantifier elimination. \citet{gulwani2008cover} defined overapproximate existential quantifier elimination as a ``cover operation'', where the goal is,  given a formula $\exists V. \phi$, to find a quantifier-free formula $\varphi$ such that $(\exists V. \phi) \Rightarrow \varphi$. If $\varphi$ is restricted to be in a DSL $\lang$, the cover problem corresponds to synthesizing \loverproperties for queries with an existential quantifier in \framework framework. 
Some algorithms \cite{Dillig2011container, Komuravelli2014modelcheck} also define underapproximate existential quantifier elimination, which corresponds to synthesis of \lunderproperties. 
Other approaches are limited to specific theories or require nontrivial, theory-specific primitives~\cite{Frazan2018Strategy, Nikolaj2015quantifier}.
\framework framework differs from above works because it allows custom DSLs that express the target quantifier-free formulas and thus is not restricted to any fixed theory.

\mypar{Logic-Based Learning} 
The angelic and demonic behaviors of nondeterminism correspond to the concepts of brave entailment (i.e., entailment from \emph{some} answer set) and cautious entailment (i.e., entailment from \emph{every} answer set) in logic-based learning \cite{Law2019entailment}.
However, the angelic notion in \framework does not precisely correspond to cautious entailment. While angelic specifications represent an under-approximation of possible behavior, cautious entailment does not necessarily imply an under-approximation. This distinction mirrors the difference between forward reasoning in incorrectness logic and backward reasoning in Hoare’s possible correctness, 
as discussed in \Cref{app:relation-to-program-logics}.

\mypar{Under-approximation}
The \framework framework could potentially be combined with existing compositional under-approximate reasoning techniques, such as incorrectness logic~\cite{Peter2019Incorrectness} or compositional symbolic execution~\cite{godefroid2007compositional}.
An inherited limitation of syntax-directed under-approximate reasoning is the inability to effectively reason about statements or procedures involving constraints beyond the scope of the theory $\mathcal{T}$ assumed by the under-approximate reasoning framework.
We expect one could synthesize weakest \lunderproperties to approximate such constraints into summaries that are expressible in the theory $\mathcal{T}$ assumed by under-approximation frameworks. 

%% file: 9conclusion.tex
\section{Conclusion}
\label{se:conclusion}


This paper presented \framework, a general framework for synthesizing over- and under-approximated specifications of both deterministic and nondeterministic programs, thus enabling broad applications---e.g., describing sources of bugs in concurrent code and finding winning strategies in two-player games.
The paper also presents general procedures for solving \framework problems using simple synthesis primitives that do not involve complex quantifier alternations.
%

Currently, our tool \name is implemented on top of the \sketch synthesizer, which results in some limitations.
First, synthesized formulas are only sound for inputs up to a given bound. 
Such an issue could be addressed by combining our approach with an off-the-shelf verifier; however, we are not aware of verifiers that can reason about \lunderproperties---i.e., under-approximated specifications.
Our work provides a motivation for building such verifiers.
Second, \sketch limits us from exploring applications that involve inputs of unbounded length---e.g., reasoning about infinite traces, LTL formulas, and reactive systems.
Our work thus opens an opportunity for the research community: by improving efficiency and providing stronger soundness guarantees for the primitives used to solve \framework problems, researchers can tackle the many applications supported by the framework.

\subsection*{Acknowledgement}

Supported, in part, by a Jacobs Faculty Scholarship, and by NSF grants CCF-2422214, CCF-2506134, CCF-2446711.
Any opinions and conclusions or recommendations expressed in this publication are those of the authors, and do not necessarily reflect the views of the sponsoring entities.

\subsection*{Data-Availability Statement}

Our implementation of \name that instantiates the \framework framework is built on top of the \sketch synthesizer. 
We provide a comprehensive Docker image on Zenodo that contains the source code and binary of \name, all the necessary dependencies, and scripts and datasets used in the experiments described in \Cref{se:evaluation}~\cite{loudzenodo}.

%% file: 11appendix-hoare.tex
\section{Relation to Program Logics}
\label{app:relation-to-program-logics}

In this section, we discuss how our problem formulation relates to the type of reasoning program logics like Hoare~\cite{hoare69axiomatic} and incorrectness logic~\cite{Peter2019Incorrectness} can do with respect preconditions/presumptions and postconditions/results.

\subsection{Relation to Hoare Logic}
\label{se:rel-hoare-app}

We start by considering Hoare logic and its ability to reason about correctness properties of programs.
A Hoare triple $\hoaretriple{P}{s}{Q}$ consists of a precondition $P$, a statement $s$, and a postcondition $Q$, and it has the following meaning:
if the precondition $P$ holds before executing $s$ and $s$ terminates, then the postcondition $Q$ holds upon termination. 
In other words, the postcondition \textit{over-approximates} the set of possible behaviors the program can result in.

We assume the semantics of a program $s$ is given by a relation $\interp{s}(\progstate, \progstate')$, which holds true if $s$ on input state $\progstate$ can terminate with an output state $\progstate'$.
%
The meaning of the triple $\hoaretriple{P}{s}{Q}$ can then be formalized as follows:
\begin{equation}
\label{eq:hoare-validity-app}
\forall \progstate, \progstate'.\, P(\progstate) \wedge \interp{s}(\progstate, \progstate') \Rightarrow Q(\progstate')
\end{equation}

\mypar{Backward Reasoning: Weakest Liberal Precondition}
Weakest precondition operations can be formalized as predicate transformers that assign a unique (in a sense, most general) precondition $P$ to each program $s$ and postcondition $Q$.
Given a program $s$ and a postcondition $Q$, the \emph{weakest liberal precondition} $\wlpre(s, Q)$ represents the weakest predicate $P$ such that the triple $\hoaretriple{P}{s}{Q}$ holds~\cite{dijkstra1976discipline}. 
If we view $\wlpre(s, Q)$ as a backward predicate transformer, it reformulates the problem of verifying the triple $\hoaretriple{P}{s}{Q}$ to the problem of checking a first-order formula $P \Rightarrow \wlpre(s, Q)$.\footnote{
Dijkstra's original weakest precondition requires that whenever the precondition $P$ holds before the execution of $s$, the execution of $s$ is guaranteed to terminate~\cite{dijkstra1975guarded}.
While our over-approximation framework can elegantly capture the notion of weakest liberal precondintion, reasoning about Dijkstra's weakest precondition and total correctness is problematic because semantics encoding presented in this section describe \textit{possible} end states $\progstate'$ from an initial state $\progstate$, without addressing whether some executions may not terminate.
}

The problem of computing such a predicate transformer, and in particular the weakest one expressible as a \textit{disjunction} of predicates in the DSL $\lang$, can be phrased in our framework.
The need to lift to \textit{disjunction} of predicates arises because there may exist multiple incomparable predicates in $\lang$ that satisfy the triple $\hoaretriple{P}{s}{Q}$ and cannot be further weakened (while still satisfying $\hoaretriple{P}{s}{Q}$).
Thus, we define $\wlprelang(s, Q)$ as follows:
if $P$ is a weakest predicate in $\lang$ such that $\hoaretriple{P}{s}{Q}$ holds true, then it must also hold true that $P \Rightarrow \wlprelang(s, Q)$.
\begin{definition}[\lwlp]
\label{defn:lwlp-app}
Given a program $s$, a postcondition $Q$ and a DSL $\lang$, 
the \emph{\lwlp} $\wlprelang(s, Q)$ is the (possibly infinite) disjunction $\bigvee_i P_i$ of all predicates $P_i \in \lang$ such that 
\rone $\hoaretriple{P_i}{s}{Q}$ holds true, and
\rtwo no $P \in \lang$ is strictly implied by $P_i$ while $\hoaretriple{P}{s}{Q}$ holds true.
\end{definition}

Note that if the DSL $\lang$ is expressive enough, the \lwlp $\wlprelang(s, Q)$ will be equivalent to the weakest liberal precondition $\wlpre(s, Q)$---i.e., the DSL would be what in Hoare logic is called an expressive enough assertion language.

Because \lwlp is a weakest formula one may be tempted to compute it via weakest \lunderproperties.
However, this approach would not yield the desired result because an \lunderproperty only proves the existence of an execution satisfying the postcondition $Q$, but it does not prove whether \emph{every} execution satisfies the postcondition $Q$.

For example, consider a simple DSL that can only express the properties true and false. 
%
Given a nondeterministic program $\x = \ast$ (i.e., one that non-deterministically assigns an integer to $\x$) and a postcondition $Q(\ecolor{\x}) := \ecolor{\x} = 0$, the property ``true'' would be an \lunderproperty of the query $\exists \ecolor{\x}.\, (\ecolor{\x} = \ast \wedge \ecolor{\x} = 0)$ because there always \emph{exist} an execution that results in $\x$ being zero, while the formula ``false'' is a valid \lwlp for the query $\exists \ecolor{\x}.\, \ecolor{\x} = \ast$ and the postcondition $\ecolor{\x} = 0$ because there is no precondition that \emph{ensures} $\x$ will be zero after the execution.
As we will see in \Cref{se:rel-incorrectness} \lunderproperties can be used to compute the \emph{weakest possible precondition}, a different backward predicate transformer used in reverse Hoare logic and incorrectness logic.

We next show that by \emph{negating} every predicate (precondition, postcondition, and language or properties), we can capture \lwlp via \loverproperties.
%
%
First, observe that by appropriately negating the pre- and postcondition, \Cref{eq:hoare-validity-app} can be rewritten as follows:
\begin{equation}
\label{eq:hoare-wlp-cond-app}
\forall \ucolor{\progstate}.\, [\exists \ecolor{\progstate'}.\, \neg Q(\ecolor{\progstate'}) \wedge \interp{s}(\ucolor{\progstate}, \ecolor{\progstate'})] \Rightarrow \neg P(\ucolor{\progstate})
\end{equation}
Intuitively, \Cref{eq:hoare-validity-app} states that every state $\progstate'$ that violates $Q$ must come from a state $\progstate$ that violates $P$.
Thus, we can introduce a DSL for negated formulae $\langneg = \{ \neg \phil \mid \phil \in \lang \}$ and have that any predicate $P$ satisfies conditions \rone and \rtwo of \Cref{defn:lspo-app} if and only if $\neg P$ is a strongest \loverpropertyp{\langneg} of the query:
\begin{equation}
\label{eq:hoare-wlp-query-app}
\exists \ecolor{\progstate'}.\, \neg Q(\ecolor{\progstate'}) \wedge \interp{s}(\ucolor{\progstate}, \ecolor{\progstate'})
\end{equation}
By combining \Cref{thm:bestset} and De Morgan's laws with the above observation, we can prove the equivalence of $\wlprelang(s, Q)$ and a best \lconjunctionp{\langneg} for the query $\eqref{eq:hoare-wlp-query-app}$.
%
%
\begin{theorem}
The \lwlp $\wlprelang(s, Q)$ is semantically equivalent to the negation of a best \lconjunctionp{\langneg} of query 
$\query := \exists \ecolor{\progstate'}.\, \neg Q(\ecolor{\progstate'}) \wedge \interp{s}(\ucolor{\progstate}, \ecolor{\progstate'})$. 
\end{theorem}
For example, to compute \lwlpparam{\langover} for the program $\y = \modhasharg{\vara}{\varmodulus}{\x}$, a postcondition $\y = 0$ and a DSL $\langover$ defined in \Cref{se:example-reachability-over}, we introduce a DSL for negated formulae $\langoverneg = \{ \neg \phil \mid \phil \in \langover\}$. Then the problem of computing $\wlprelangp{\langover}(\y = \modhasharg{\vara}{\varmodulus}{\x}, \y = 0)$ is encoded as the problem of synthesizing a best \lconjunctionp{\langoverneg} for the query:
\begin{equation}
\label{eq:hoare-wlp-query-ex}
\exists \ecolor{\y}.\, \neg (\ecolor{\y} = 0) \wedge \ecolor{\y} = \modhasharg{\ucolor{\vara}}{\ucolor{\varmodulus}}{\ucolor{\x}}
\end{equation}

The following set of \loverpropertiesp{\langoverneg} forms a best \lconjunctionp{\langoverneg} for the query \eqref{eq:hoare-wlp-query-ex}:
\begin{equation}
\label{eq:hoare-wlp-result-ex}
\begin{array}{c@{\hspace{8.0ex}}c@{\hspace{8.0ex}}c@{\hspace{8.0ex}}c}
\ucolor{\vara} \neq 0 & \ucolor{\x} \neq 0 &
\ucolor{\vara} \neq \ucolor{\varmodulus} & \ucolor{\x} \neq \ucolor{\varmodulus}
\end{array}
\end{equation}
By negating the formulas in \Cref{eq:hoare-wlp-result-ex} we get the following disjunction---i.e., the \lwlp:
\begin{equation*}
\label{eq:hoare-wlp-result-ex-formula}
\ucolor{\vara} = 0 \lor \ucolor{\x} = 0 \lor
\ucolor{\vara} = \ucolor{\varmodulus} \lor \ucolor{\x} = \ucolor{\varmodulus}
\end{equation*}

\mypar{Forward Reasoning: Strongest Postcondition}
Strongest postcondition predicate transformers can be thought of being the dual of the weakest precondition ones.
Given a program $s$ and a precondition $P$, the \emph{strongest postcondition} $\spost(s, P)$ represents the strongest predicate $Q$ such that the triple $\hoaretriple{P}{s}{Q}$ holds~\cite{dijkstra1990predicate}. 
If we view $\spost(s, Q)$ as a forward predicate transformer, we can reformulate the problem of verifying the triple $\hoaretriple{P}{s}{Q}$ as the problem of checking whether a first-order formula $\spost(s, P) \Rightarrow Q$ holds.

The problem of computing such a predicate transformer, and in particular the strongest one expressible as a \textit{conjunction} of predicates in the DSL $\lang$, can be phrased in our framework.
Similar to the case of the weakest liberal precondition, we define $\spostlang(s, P)$ as follows:
if $Q$ is a strongest predicate in $\lang$ such that $\hoaretriple{P}{s}{Q}$ holds true, then it must also hold true that $\spostlang(s, P) \Rightarrow Q$---i.e., the \lspo $\spostlang(s, P)$ is stronger than any strongest postcondition $Q$ in $\lang$.
%
\begin{definition}[\lspo]
\label{defn:lspo-app}
Given a program $s$, a precondition $P$ and a DSL $\lang$, 
the \emph{\lspo} $\spostlang(s, P)$ is the (possibly infinite) conjunction $\bigwedge_i Q_i$ of all predicates $Q_i \in \lang$ such that 
\rone $\hoaretriple{P}{s}{Q_i}$ holds true, and
\rtwo no $Q \in \lang$ strictly implies $Q_i$ while $\hoaretriple{P}{s}{Q}$ holds true.
\end{definition}
%

The problem of obtaining the \lspo $\spostlang(s, P)$ can also be encoded as synthesizing strongest \loverproperties, this time without a need of negating formulas.
Observe that relocating the quantifier $\forall \progstate$ into the implicant rewrites \Cref{eq:hoare-validity-app} as follows:
\begin{equation}
\label{eq:hoare-spost-cond-app}
\forall \ucolor{x'}.\, [\exists \ecolor{\progstate}.\, P(\ecolor{\progstate}) \wedge \interp{s}(\ecolor{\progstate}, \ucolor{\progstate'})] \Rightarrow Q(\ucolor{\progstate'})
\end{equation}
From \Cref{eq:hoare-spost-cond-app} we have  that a predicate $Q \in \lang$ satisfies conditions \rone and \rtwo of \Cref{defn:lspo-app} if and only if $Q$ is a strongest \loverproperty of the following query:
\begin{equation}
\label{eq:hoare-sp-query-app}
\exists \ecolor{\progstate}.\, P(\ecolor{\progstate}) \wedge \interp{s}(\ecolor{\progstate}, \ucolor{\progstate'})
\end{equation}
%
This observation yields the following theorem:
\begin{theorem}
The \lspo $\spostlang(s, P)$ is semantically equivalent to the best \lconjunction of the query 
$\query := \exists \ecolor{\progstate}.\, P(\ecolor{\progstate}) \wedge \interp{s}(\ecolor{\progstate}, \ucolor{\progstate'})$.
\end{theorem}
\Cref{se:example-reachability-over} illustrated how to obtain a \lspoparam{\langover} $\spostlangp{\langover}(\y = \modhasharg{\vara}{\varmodulus}{\x}, \top)$ using over-appximated reasoning in our framework.

\subsection{Relation to Reverse Hoare Logic and Incorrectness Logic}
\label{se:rel-incorrectness-app}

Reverse Hoare logic \cite{edsko2011reverse} and incorrectness logic \cite{Peter2019Incorrectness} can both be thought of being the dual of Hoare logic---they under-approximate (instead of overapproximate) the set of possible behaviors a program can result in.
The under-approximating logics are semantically equivalent despite having been designed with different goals in mind: reverse Hoare logic was designed to reason about the correctness of nondeterministic programs, whereas incorrectness logic was designed to identify the presence of bugs in programs.
Our evaluation in \Cref{se:evaluation} investigates both of these applications.

A \textit{incorrectness triple} $\incortriple{P}{s}{Q}$ consists of a presumption $P$, a statement $s$, and a result $Q$, and it has the following meaning: every final state satisfying $Q$ is reachable by executing program $s$ starting from \textit{some} state that satisfies presumption $P$. 
In other words, the predicate $Q$  \textit{under-approximates} the set of possible behaviors the program $s$ can result in when executed on inputs satisfying $P$:
\begin{equation}
\label{eq:incorrectness-validity-app}
\forall \progstate'.\, Q(\progstate') \Rightarrow \exists \progstate.\, [P(\progstate) \wedge \interp{s}(\progstate, \progstate')]
\end{equation}


\mypar{Forward Reasoning: Weakest Under-approximate Postcondition}
Weakest postcondition operations can be formalized as predicate transformers that assign a unique precondition $Q$ to each program $s$ and precondition $P$.
Given a program $s$ and a presumption $P$, the \emph{weakest under-approximate postcondition} $\wpost(s, P)$ represents the weakest predicate $Q$ such that the triple $\incortriple{P}{s}{Q}$ holds. 
If we view $\wpost(s, P)$ as a forward predicate transformer, we can reformulate the problem of verifying the triple $\incortriple{P}{s}{Q}$ as the problem of checking whether a first-order formula $Q \Rightarrow \wpost(s, P)$ holds.

The problem of computing such a predicate transformer, and in particular the weakest one expressible as a \textit{disjunction} of predicates in the DSL $\lang$, can be phrased in our framework.
We define $\wpostlang(s, P)$ as follows:
if $Q$ is a weakest predicate in $\lang$ such that $\incortriple{P}{s}{Q}$ holds true, then it must also hold true that $Q \Rightarrow \wpostlang(s, P)$.
\begin{definition}[\lwpo]
\label{defn:lwpo-app}
Given a program $s$, a presumption $P$ and a DSL $\lang$, 
the \emph{\lwpo} $\wpostlang(s, P)$ is the (possibly infinite) disjunction $\bigvee Q_i$ of all predicates $Q_i \in \lang$ such that 
\rone $\incortriple{P}{s}{Q_i}$ holds true, and
\rtwo no $Q \in \lang$ is strictly implied by $Q_i$ while $\incortriple{P}{s}{Q}$ holds true.
\end{definition}
Following \Cref{eq:incorrectness-validity-app}, the problem of obtaining  the \lwpo $\wpostlang(s, P)$ can be directly encoded as synthesizing weakest \lunderproperties. 
A predicate $Q \in \lang$ satisfies conditions \rone and \rtwo of \Cref{defn:lwpo-app}
if and only if $Q$ is a weakest \lunderproperty of the following query:
\begin{equation}
\label{eq:incor-wp-query-app}
\exists \ecolor{\progstate}.\, P(\ecolor{\progstate}) \wedge \interp{s}(\ecolor{\progstate}, \ucolor{\progstate'})
\end{equation}
This observation yields the following theorem:
\begin{theorem}
The \lwpo $\wpostlang(s, P)$ is semantically equivalent to the best \ldisjunction of the query $\query := \exists \ecolor{\progstate}.\, P(\ecolor{\progstate}) \wedge \interp{s}(\ecolor{\progstate}, \ucolor{\progstate'})$. 
\end{theorem}


\mypar{Backward Reasoning: Weakest Possible Precondition} 
While forward predicate transformers for incorrectness logic behave well---i.e., given a presumption $P$ and a program $s$, one can always assign the weakest result $Q$ such that $\incortriple{P}{s}{Q}$ holds true---backward predicate transformers for incorrectness logic do not always exist! 
This problem arises because valid presumptions may not exist in incorrectness logic.
For example, there is no predicate $P$ making the triple $\incortriple{P}{\y = \modhasharg{\vara}{\varmodulus}{\x}}{\y = -1}$ true because no values of $\vara$, $\varmodulus$ and $\x$ satisfies $\modhasharg{\vara}{\varmodulus}{\x} = -1$.
To address this shortcoming and still take advantage of some form of backward reasoning in incorrectness logic, \citet{Peter2019Incorrectness} suggests using the weakest possible precondition $\wppre(s, Q)$, which is predicate transformer described by \citet{hoare78properties} for what he referred as ``possible correctness''.
Intuitively, $\wppre(s, Q)$ captures the set of initial states from which it is \textit{possible} to execute $s$ and terminate in a state that satisfies $Q$.
 
\citeauthor{Peter2019Incorrectness} then proposes to use a two phase approach to derive valid incorrectness triples as follows.
Starting with a postcondition $Q$, one first computes the weakest possible precondition $P=\wppre(s, Q)$ and then applies forward reasoning and computes the weakest under-approximate postcondition $Q' = \wpost(s, P)$ to obtain a valid incorrectness triple $\incortriple{P}{s}{Q'}$.
Since we already showed how to capture the weakest under-approximate postcondition in our framework, we now show how to capture the weakest possible precondition. 

A predicate $P$ is called a \textit{possible precondition} of predicate $Q$ for program $s$, if every input state satisfying $P$ has a run of the program $s$ that terminates to an end state satisfying $Q$.
That is, $\wppre(s, Q)$ is the weakest $P$ satisfying
\begin{equation}
\label{eq:hoare-wpp-cond-app}
\forall \ucolor{\progstate}.\, P(\ucolor{\progstate}) \Rightarrow [\exists \ecolor{\progstate'}.\, Q(\ecolor{\progstate'}) \wedge \interp{s}(\ucolor{\progstate}, \ecolor{\progstate'})]
\end{equation}
Note that $P = \wppre(s, Q)$ does not form neither a Hoare nor an incorrectness triple with the program $s$ and the postcondition $Q$.
The postcondition $Q$ is not a valid over-approximation of possible final states because there could be other executions that do not satisfy $Q$---e.g., if the program is nondeterministic.
The postcondition $Q$ is not a valid under-approximation either because there might be states satisfying $Q$ that are not reachable---e.g., if the postcondition $Q$ is $\tru$ and the program $s$ is $\y = \modhasharg{\vara}{\varmodulus}{\x}$, then any negative value of $\y$ satisfies the postcondition $Q$, but no values of $\vara$, $\varmodulus$ and $\x$ can yield a negative output.
As proposed by O'Hearn, we can remedy this issue by computing a new postcondition $Q' = \wpost(s, P)$ using the weakest under-approximate postcondition operator to obtain a valid incorrectness triple $\incortriple{P}{s}{Q'}$

The problem of computing such a predicate transformer, and in particular the weakest one expressible as a \textit{disjunction} of predicates in the DSL $\lang$, can be phrased in our framework.
\begin{definition}[\lwpp]
\label{defn:lwpp-app}
Given a program $s$, a postcondition $Q$ and a DSL $\lang$, 
the \emph{\lwpp} $\wpprelang(s, Q)$ is the (possibly infinite) disjunction $\bigvee_i P_i$ of all predicates $P_i \in \lang$ such that 
\rone \Cref{eq:hoare-wpp-cond-app} holds true, and
\rtwo no $Q \in \lang$ is strictly implied by $Q_i$ while \Cref{eq:hoare-wpp-cond-app} holds true.
\end{definition}

From \Cref{eq:hoare-wpp-cond-app}, the problem of obtaining \lwpp $\wpprelang(s, Q)$ can be directly encoded as synthesizing weakest \lunderproperties. 
A predicate $P \in \lang$ satisfies conditions \rone and \rtwo of \Cref{defn:lwpp-app}
if and only if $P$ is a weakest \lunderproperty of the following query:
\begin{equation}
\label{eq:incor-wpp-query-app}
\exists \ecolor{\progstate'}.\, Q(\ecolor{\progstate'}) \wedge \interp{s}(\ucolor{\progstate}, \ecolor{\progstate'})
\end{equation}
This observation yields the following theorem:
\begin{theorem}
The \lwpp $\wpprelang(s, Q)$ is semantically equivalent to the best \ldisjunction of the query $\query := \exists \ecolor{\progstate'}.\, Q(\ecolor{\progstate'}) \wedge \interp{s}(\ucolor{\progstate}, \ecolor{\progstate'})$. 
\end{theorem}

%% file: 12appendix-alg.tex
\section{Synthesizing a Best \lconjunction and \ldisjunction}
\label{app:alg-allproperties}

In this section we present detailed algorithms to synthesize a best \lconjunction(\Cref{app:alg-conjunction}) and \ldisjunction(\Cref{app:alg-disjunction}).

\subsection{Synthesizing a best \lconjunction}
\label{app:alg-conjunction}

\begin{algorithm}[htbp] {\it
\caption{$\synthoverproperties(\query)$}
\label{alg:SynthesizeAllOverProperties}
\DontPrintSemicolon
$\phiand \gets \top$ \\ 
$\Pi \gets \emptyset$ \\ 
\label{Li:InitializePhiToEmpty}
$\eplus \gets \emptyset$ \\ 
\label{Li:InitEminusMayEmpty}
\While{$\tru$}{
    \label{Li:BeginWhileLoop}
    $\phil, \eplus, \eminus \gets \synthoverproperty(\query,  \phiand, \tru, \eplus, \emptyset)$  \label{Li:CallSynthPropertyOne} \\
    \If{$\eminus = \emptyset$}
    {
        \label{Li:BeginIf}
        $\negex \gets$ \textsc{IsSat($\phiand \wedge \neg \phil $)} 
        \label{Li:CheckForImprovement} \\
        \eIf{$\negex \neq \bot$} {
            $\eminus \gets \{ \negex \}$ \\
        }{
            \Return $\Pi$ 
            \label{Li:SynthPropertiesReturn}
        }
    }
    $\phil,\eplus, \_$ $\gets$ $\synthoverproperty(\query, \tru, \phil, \eplus, \eminus$)  \label{Li:CallSynthPropertyTwo} \\
    $\Pi \gets \Pi \cup \{ \phil \}$ 
    \label{Li:UpdatePhi} \\
    $\phiand \gets \phiand \wedge \phil$ 
    \label{Li:UpdateConjunction} \\
}
}\end{algorithm}

The algorithm \synthoverproperties iteratively synthesizes incomparable strongest \loverproperties.
At each iteration, \synthoverproperties keeps track of the conjunction of synthesized strongest \loverproperties $\phiand$, along the set of positive examples $\eplus$ that have been observed so far.

Each iteration \synthoverproperties calls \synthoverproperty to try to synthesize a strongest \loverproperty for $\query$ with respect to $\phiand$ (line \ref{Li:CallSynthPropertyOne}). 
A property $\phil$ returned by \synthoverproperty is checked whether it rejects some example that was not rejected by $\phiand$ (lines \ref{Li:BeginIf} and \ref{Li:CheckForImprovement}).

If $\phil$ does not reject any example that was not already rejected by $\phiand$, the formula $\phiand$ is a best \lconjunction, and thus \synthoverproperties returns the set of synthesized \loverproperties (line \ref{Li:SynthPropertiesReturn}).

If $\phil$ rejects some example that was not rejected by $\phiand$, 
\synthoverproperties needs to further strengthen $\phil$ to a strongest \loverproperty for $\query$ with respect to examples that might already be rejected by $\phiand$.
Without this step the returned \loverproperty may be imprecise for examples that were not considered by \synthoverproperty because they were outside of $\phiand$.
To achieve this further strengthening, \synthoverproperties makes another call to \synthoverproperty with the example sets $\eplus$ and $\eminus$ returned by the previous call to \synthoverproperty together with $\phil$ and $\phiinit := \phil$, but with $\phiand := \tru$ (line \ref{Li:CallSynthPropertyTwo}).

\subsection{Synthesizing a best \ldisjunction}
\label{app:alg-disjunction}

\begin{algorithm}[t] {\it
\caption{$\synthunderproperties(\query)$}
\label{alg:SynthesizeAllUnderProperties}
\DontPrintSemicolon
$\diffcolor{\phior} \gets \diffcolor{\bot}$ \\ 
$\Pi \gets \emptyset$ \\ 
$\diffcolor{\eminus} \gets \emptyset$ \\ 
\While{$\tru$}{
    $\phil, \eplus, \eminus \gets \diffcolor{\synthunderproperty}(\query,  \diffcolor{\phior}, \diffcolor{\fls}, \diffcolor{\emptyset}, \diffcolor{\eminus})$ \\ 
    \If{$\diffcolor{\eplus} = \emptyset$}
    {
        $\diffcolor{\posex} \gets$ \textsc{IsSat($\diffcolor{\neg \phior \wedge \phil}$)} \\ 
        \eIf{$\diffcolor{\posex} \neq \bot$} {
            $\diffcolor{\eplus} \gets \{ \diffcolor{\posex} \}$ \\
        }{
            \Return $\Pi$ 
        }
    }
    $\phil,\diffcolor{\_}, \diffcolor{\eminus}$ $\gets$ $\diffcolor{\synthunderproperties}(\query, \diffcolor{\fls}, \phil, \eplus, \eminus$) \\
    $\Pi \gets \Pi \cup \{ \phil \}$ \\ 
    $\diffcolor{\phior} \gets \diffcolor{\phior \lor}\; \phil$ \\ 
}
}
\end{algorithm}

Because \synthunderproperties solves the dual problem of the one solved by \synthoverproperties, the two algorithms share the same structure.
Due to the duality
\rone the roles of positive and negative examples are inverted;
\rtwo $\tru$ is replaced by $\fls$; and
\rthree at each iteration a weakest \lunderproperty is synthesized by \synthunderproperty instead of \synthoverproperty.
These changes are highlighted in \diffcolor{violet} in Algorithm~\ref{alg:SynthesizeAllUnderProperties}.

%% file: 13appendix-proof.tex
\section{Proofs}
\label{app:proof}

\finitecompleteness*
\begin{proof}
    Because of the duality between \loverproperties and \lunderproperties, we only prove for \synthoverproperty and \synthoverproperties here (the other case is identical).
    We first prove the termination of \synthoverproperty.

    If the example domain is finite, there are only finitely many semantically inequivalent \lproperties, so the case of finite example domain can be reduced to the case of finite DSL.

    For the case of a finite DSL, the theorem can be proved by first proving \checkimpl (\Cref{Li:CallCheckSoundness}) does not return $\top$ (i.e., $\phil$ is not a consequence) only finitely many times, and then proving \cpover (\Cref{Li:CallCheckPrecision}) does not return $\top$ (i.e., $\phil$ is a consequence but not a strongest one) only finitely many times between successive instances where \checkimpl returns $\top$.

    Every time \checkimpl does not return $\top$, it returns a positive example $\posex$ that is rejected by the current \lproperty $\phil$. However, since $\posex$ is added to $\eplus$, the properties in the subsequent iterations must accept $\posex$ so that $\phil$ will not occur again as an argument of \checkimpl. Therefore, we only meet finitely many properties that are not consequences.

    We can now move to the case where an \lproperty $\phil$ reaches \cpover at line \ref{Li:CallCheckPrecision}.
    In this case, there are two possibilities:
    \rone \cp returns a property $\phil'$ that is not a consequence (and a negative example $\negex_1$);
    \rtwo \cp returns a property $\phil'$ that is a consequence (and a negative example $\negex_2$).
    For case \rone, \checkimpl in the next iteration will not return $\top$.
    For case \rtwo, $\negex_2$ will be merged to $\eminus$ so that the properties in the subsequent iterations must reject $\negex_2$, whereas $\phil$ accepts it. Consequently, $\phil$ will not occur again as an argument of \cpover.
    Therefore, \cpover does not return $\top$ finitely many times between each time \checkimpl does not return $\top$.

    Turning to \synthoverproperties,
    since it calls \synthproperty to grow a set of incomparable \loverproperties, the fact that each call on \synthproperty terminates guarantees that \synthoverproperties always terminates.
\end{proof}

%% file: 10appendix-eval.tex
\section{Evaluation Details}
\label{App:eval}



\subsection{\typeone~ Test Set: Mining Under-Approximated Specifications}
\label{App:eval-specmining}

The goal of this section is to evaluate \name's general capability to mine \textit{under-approximated specifications} for deterministic programs. 
\citet{park2023specification} mined \loverproperties for programs used in the synthesis literature. 
Their benchmark set consists of 45 programs paired with corresponding grammars that their tool \spyro uses to mine \loverproperties. 
Note that the queries for all the \spyro benchmarks \textit{do not} contain existential quantifiers.

We successfully recomputed \loverproperties produced by \spyro for all benchmarks.
However, we found that \name's implementation was on average 2.2x faster (geometric mean) than \spyro at computing \loverproperties for deterministic programs.
Although \name and \spyro implement the same algorithm for computing \loverproperties for deterministic programs, \name's implementation of \cp forces the generated example to be negative (using an assertion), whereas \spyro's implementation generates an example and then checks later if it is negative. 
The implementation avoids spurious calls to \cp and is therefore faster.

We focus on the new capability of \name to compute \lunderproperties.
To compute \lunderproperties for the deterministic benchmarks used to evaluate \spyro, we modified the top-level production of each DSL used in the \spyro evaluation to consider conjunctions instead of disjunctions---i.e., if the top-level production was of the form $P \to AP \lor AP \lor \cdots$, we replaced it with $P \to AP \land AP \land \cdots$. 
We denote the grammar before modification as $\langover$ and after modification as $\langunder$. 

Of the 45 benchmarks used when evaluating \spyro, we only consider 35 benchmarks for which the semantics of operations were expressed using \sketch.
The 10 benchmarks we don't consider are simple ones for which the semantics are expressed directly using SMT formulas instead of \sketch; \name does not support SMT semantics yet.
The \numSpec benchmarks consist of: \numSpecSyGuS programs from the CLIA track of the \sygus competition \cite{https://doi.org/10.48550/arxiv.1904.07146}, 22 programs manipulating data structures taken from the \synquid synthesizer~\cite{polikarpova2016program}, and 6 imperative programs designed by the authors of \spyro.
For the \sygus problems and imperative programs, the DSLs can express arithmetic comparison between variables---e.g. for the query $o = \iaabs{}(x)$, one would get the \loverpropertyp{\langover} $-x \ge x \Rightarrow x = -o$. 
For \synquid problems, the DSLs support common functions over the data structures appearing in each specific benchmark ---e.g., length for lists. 
For the query $l_{out} = \lreverse (l_{in})$ one would get $\size(l_{out}) = \size(l_{in})$ as one of the \loverpropertiesp{\langover}.

\name could synthesize properties for 35/35 benchmarks, and guaranteed that all of them were best \lunderpropertiesp{\langunder}. \Cref{tab:benchmarks-spyro} shows the evaluation details of a few selected benchmarks.
\input{table-mining}

%
%
Overall, \name could solve each of the \numSpec benchmarks in under 7 minutes (the largest grammar contained $1.48\cdot 10^{13}$ properties). Next, we analyze the results for each subcategory separately.

\mypar{\sygus Benchmarks} 
%
For the 7 \sygus benchmarks, we found that the synthesized \lunderpropertiesp{\langunder}  exactly coincides with the semantics of the given queries. 
When we replicated the experiment by~\citet{park2023specification}, we could only synthesize \loverpropertiesp{\langover} for 4/7 benchmarks (like \spyro, \name failed on \maxfour, \maxfive, and \arraythree), but for these 4 benchmarks the synthesized properties also coincided with the semantics of the given queries.
The result shows that the grammar $\langunder$ and $\langover$ of the benchmarks are expressive enough to synthesize \textit{exact} approximations---i.e., the properties are semantically equivalent to the query. 
Take $o = \maxtwo(x_1, x_2)$ as an example: the synthesized \loverpropertiesp{\langover} are $\{o = x_1 \lor o = x_2, x_2 = x_1 \lor o > x_1 \lor o > x_2\}$, and the synthesized \lunderpropertiesp{\langunder} are $\{o = x_1 \land x_2 < o, o = x_2 \land x_1 < o\}$, which can be proved to be equivalent.
Due to the equivalence of the \lunderpropertiesp{\langunder} and \loverpropertiesp{\langover}, we further compared the efficiency of computing over-approximations and under-approximations for these benchmarks to assess which formulas were easier to compute. 
Although the size of grammar for \loverpropertiesp{\langover} and \lunderpropertiesp{\langunder} are the same (since we only replace $\lor$ by $\land$ at the top layer), synthesizing \lunderpropertiesp{\langunder} was on average 1.6x faster (geometric mean) than synthesizing \loverpropertiesp{\langover} for the 4 benchmarks on which both synthesis processes terminated. 
Furthermore, while \spyro (and also \name) failed to synthesize \loverpropertiesp{\langover} for $\maxfour$, $\maxfive$, and $\arraythree$ within the given timeout, \name successfully synthesized \lunderpropertiesp{\langunder} in less than 5 minutes for each of these benchmarks.

We compared the \lunderpropertiesp{\langunder} to \loverpropertiesp{\langover} to see how they differed.
For the query $o = \maxthree(x_1, x_2, x_3)$, for example, the synthesized \loverpropertiesp{\langover} consist of 5 conjuncts 
\begin{equation}
\label{eq:max3-over}
\begin{array}{c@{\hspace{4.0ex}}c@{\hspace{4.0ex}}c}
x_2 < o \lor x_1 < o \lor x_2 = x_1 &
x_3 \le o &
o = x_2 \lor o = x_1 \lor x_1 < x_3  \\
o = x_2 \lor o = x_3 \lor x_3 < x_1 &
x_2 < x_3 \lor o = x_2 \lor x_2 < x_1 
\end{array}
\end{equation}
while the synthesized \lunderpropertiesp{\langunder} consist of 3 disjuncts
\begin{equation}
\label{eq:max3-under}
\begin{array}{c@{\hspace{4.0ex}}c@{\hspace{4.0ex}}c}
x_1 = o \land x_3 \le x_1 \land x_2 \le x_1&
x_2 = o \land x_3 \le x_2 \land x_1 \le x_2&
x_3 = o \land x_1 \le x_3 \land x_2 \le x_3
\end{array}
\end{equation}

%
The \loverpropertiesp{\langover} in Eq. \ref{eq:max3-over} can be thought of as a declarative specification that \maxthree ~function much to satisfy, and in fact similar to the specifications provided in \sygus benchmarks. 
Instead, the \lunderpropertiesp{\langunder} in Eq. \ref{eq:max3-under} captures the paths of the program and the output they produce, which one would obtain via symbolic execution. 
%

To summarize, under-approximation computed semantically equivalent specifications for \sygus benchmarks, but it was faster than over-approximation.
This improvement could be attributed to the fact that fewer \lunderpropertiesp{\langunder} could capture the semantics of a program when compared to \loverpropertiesp{\langover}---e.g. 3 \lunderpropertiesp{\langunder} vs 5 \loverpropertiesp{\langover} for \maxthree.

\mypar{Imperative Program Benchmarks}
These \numSpecImp benchmarks ask \name to find linear or nonlinear relations to capture the semantics of simple imperative programs. \name synthesized \lunderpropertiesp{\langunder} that exactly capture the program semantics except for \ianonlinsum{1}, for which the DSL can only describe linear relations, but the program semantics can only be captured by a quadratic relation.


\mypar{\synquid Benchmarks} 
%
Many of the \lunderproperties synthesized by \name for this benchmark category were not particularly useful. 
Specifically, the DSLs are not expressive enough for computing useful \lunderproperties, mostly yielding trivial properties involving cases in which the data structure is empty. 
%
For example, for the query $l_{out}=\lreverse(l_{in})$, \name only synthesized the \lunderpropertyp{\langunder} $\{l_{out} = l_{in} \land \size(l) \le 1\}$, which only describes the behavior of the function $\lreverse$ on lists of length less than 1. 
While the original DSLs were useful for computing consequences (e.g., $\size(l_{in}) = \size(l_{out})$), their dual versions are too weak to reason about implicants---one would instead need to talk about more specific position information of elements in lists.
The same problem holds for other data structures.

\textbf{Finding:} 
\name can mine \loverproperties and \lunderproperties that help understand program behaviors.
%
%
If the DSL is not expressive enough, \name cannot produce useful \lunderproperties.

\subsection{\typetwo~ Test Set: Simple nondeterministic programs}
\label{App:eval-non-deter}

In this section, we evaluate \name's ability to reason about nondeterministic programs.
Nondeterminism can be modeled using existential quantifiers, a key innovation of the \framework framework.
Consider a program that takes as input an integer $x$ and then adds a nondeterministically chosen positive number to it. 
Such a program can be modeled using the query \inlinef{$\exists \ecolor{h}.~\ucolor{o} = \ucolor{x} + \ecolor{h}$} where $\ecolor{h}$ is a positive number. 
A consequence of this query is the property \inlinef{$\ucolor{o} \ge \ucolor{x}$}, which holds for all possible values of $\ecolor{h}$. 
An implicant of the same query is the property \inlinef{$\ucolor{o} = \ucolor{x} + 17$}, which holds when $\ecolor{h} = 17$. 

From these two example properties (which can be synthesized by \name when given proper DSLs), we observe that for nondeterministic programs, consequences hold for \textit{every possible} nondeterministic choice (the demonic perspective of nondeterminism), whereas implicants hold for \textit{at least one} nondeterministic choice (the angelic perspective). 
To model a nondeterministic program, we introduce an array of nondeterministic values $\ecolor{H}$ into the query as an existentially quantified variable. Whenever the program execution reaches a non-deterministic command (e.g. \texttt{nondet()} is called), the command takes the next value of the array $\ecolor{H}$.



\subsubsection{Benchmark Selection and Quantitative Results} 
We collected 12 benchmarks involving nondeterminism: \rone we created 4 nondeterministic sorting algorithms where the goal is to synthesize properties that characterize when the algorithm works or does not work as intended;
\rtwo we collected  the 4 nondeterministic  recursive programs for which the goal is to synthesize a polynomial invariant by~\citet{Chatterjee2020Polynomial} (all other benchmarks by~\citeauthor{Chatterjee2020Polynomial} are deterministic), and
\rthree we collected 4 SV-COMP~\cite{SVCOMP24} benchmarks in the bitvector category where nondeterministic values are used to model unknown inputs and parameters that determine control flow in the program (all others SV-COMP benchmarks are deterministic or unsuitable for specification synthesis).

\input{table-nondeter}

For each category of benchmarks, we acted as the users of \framework and designed the DSLs following such a methodology: we included constants, free variables that appear in the query, basic arithmetic operations between them, and also several atomic functions and predicates that were related to the problem or we expected the resulting properties to contain. We describe the details of the DSLs in \Cref{se:eval-non-deter-qual}.

\name synthesizes \loverpropertiesp{\langover} and \lunderpropertiesp{\langunder} for 11/12 benchmarks. For those solved benchmarks, \name takes less than 400 seconds to synthesize \loverpropertiesp{\langover} and less than 140 seconds to synthesize \lunderpropertiesp{\langunder}. 
Evaluation details are shown in \Cref{tab:benchmarks-nondeter}.

Because these benchmarks involve existential quantifiers, they require \Cref{alg:CEGIS-Loop,alg:CEGQI-Loop-Over}, the CEGQI algorithms presented in \Cref{se:cegqi}.
\Cref{alg:CEGIS-Loop,alg:CEGQI-Loop-Over} avoided exploring a large portion of the space of values for existentially quantified variables---e.g., for the benchmark \wppBubble{4} each call to \Cref{alg:CEGIS-Loop,alg:CEGQI-Loop-Over} terminated with at most 13 instances $\ecolor{h}$ in $H$, instead of considering all $3^8 = \num{6561}$ nondeterministic instances.
When we reuse the instance $\ecolor{h}$ generated across all calls to \checkex, the total number was only 17.
%
This result inspired us to create a version of our algorithm that caches instances and reuses them across different calls to \checkex.
Caching and reusing instances results, on average, in a \cachespd speedup.


\subsubsection{Qualitative Evaluation}
\label{se:eval-non-deter-qual}
We discuss each benchmark in detail.

\mypar{Random Shuffle} 
The \shuff{\text{$n$}} 
($n\in\{3,4,5\}$) benchmarks model a random shuffle function that, given an array $l = [a_1, \cdots, a_{n}]$, for each element $a_i$, nondeterministically picks a value $j<i]$ and swaps $a_i$ and $a_j$. 
When encoding the problem as query $\exists \ecolor{h}.~ \ucolor{l_{out}} = \shuff{}(\ucolor{l_{in}}, \ecolor{h})$ and providing the DSL $\lang$ in \Cref{eq:shuffle-grammar}, \name synthesized \inlinef{$\sort(\ucolor{l_{in}}) = \sort(\ucolor{l_{out}})$} as both \loverproperty and \lunderproperty, which essentially states that the output should be a permutation of the input, and any permutation can be obtained as the output. 

\begin{equation}
\label{eq:shuffle-grammar}
\begin{array}{rcl}
    S &:= & L = L \\
     L & := & \ucolor{l_{in}} \mid \ucolor{l_{out}} \mid \sort(\ucolor{l_{in}}) \mid \sort(\ucolor{l_{out}}) 
\end{array}
\end{equation}

\mypar{Nondeterministic Sorting}
The \wppBubble{\text{$l$}} ($l\in\{3,4\}$) benchmarks model a program that, given an array $[\ucolor{a_1}, \cdots, \ucolor{a_{l}}]$, repeatedly nondeterministically swaps two neighboring elements at most $\ucolor{n}$ time and returns whether the final array is sorted (stored in Boolean variable $\ucolor{ok}$).
When encoding the problem as the query $\exists \ecolor{H}.~\ucolor{ok} = bubble(\ucolor{a_1},\cdots, \ucolor{a_l}, \ecolor{H})$ and supplying the DSLs $\langover$ (rooted at $B_\langover$) and $\langunder$ (rooted at $B_\langunder$) in \Cref{eq:bubble-grammar},
for \wppBubble{3} \name computes 7 \lunderpropertiesp{\langunder}, including
\inlinef{$\ucolor{n} {\ge} 1 \land \ucolor{a_1} {\le} \ucolor{a_3} \land \ucolor{ok}$}, which states that if the first and third elements are already in the right order, we can make the entire array sorted using one swap. 
\name also synthesizes 5 \loverpropertiesp{\langover}. For example, the \loverpropertyp{\langover} \inlinef{$(\ucolor{n} < 3 \land \ucolor{a_1} > \ucolor{a_2} \land \ucolor{a_2} > \ucolor{a_3}) \Rightarrow \neg \ucolor{ok} $} tells us that, if the array is descending, we cannot make it sorted using fewer than $3$ swaps.

\begin{equation}
\label{eq:bubble-grammar}
\begin{array}{rcl}
    B_{\langover} &:= & G \Rightarrow D \qquad \qquad B_{\langunder} \ \ :=\ \ G \land D \\
     G & := & \top \mid \nonap \mid \nonap \land \nonap \mid \cdots \mid \nonap \land \nonap \land \nonap \land \nonap \land \nonap\\
     \nonap & := & \nonint ~\{\le \mid < \mid =\mid \neq \}~ \nonint \mid \ucolor{n} \{<\mid \ge\} \{0\mid 1\mid \cdots \mid l^2\}  \\
     \nonint & := & \ucolor{a_1} \mid \cdots \mid \ucolor{a_{l}}   \\
     D & := & \ucolor{ok} \mid \neg \ucolor{ok} \\
\end{array}
\end{equation}

The benchmarks \wppSwap{3} and \wppSwap{4} consider a similar problem to \wppBubble{$l$}, but allow swaps of arbitrary elements in the array instead of only neighboring ones.

\mypar{Polynomial invariants~\cite{Chatterjee2020Polynomial}} 
This category contains 4 benchmarks: \rsum, \rsquaresum, \rcubicsum, and \mergesort. The DSLs  contain polynomials of a certain degree.

The  \rsum~benchmark models a program that nondeterminisitcally adds a subset of the numbers from $1$ to $n$, i.e., $\ucolor{s} = \sum_{k=1}^{\ucolor{n}} \texttt{nondet()} \cdot k$, which is encoded as a query \inlinef{$\exists \ecolor{H}.~\ucolor{s} = rsum(\ucolor{n}, \ecolor{H})$}.
When given a DSL that can describe quadratic functions over $\ucolor{n}$, \name produces $\ucolor{n}^2 + \ucolor{n} \ge 2\ucolor{s} \ge 0$ as both the only \loverpropertyp{\langover} and the only \lunderpropertyp{\langunder}. 
The synthesized formula tells us the summation is not greater than $(\ucolor{n}^2 + \ucolor{n})/2$ and all natural numbers that are not greater than $(\ucolor{n}^2 + \ucolor{n})/2$ can be obtained as the result.
Similarly, the benchmarks \rsquaresum~and \rcubicsum~model $\sum_{k=1}^n \texttt{nondet()} \cdot k^2$ and $\sum_{k=1}^n \texttt{nondet()} \cdot k^3$.
\name synthesizes \inlinef{$2\ucolor{n}^3 + 3\ucolor{n}^2 + \ucolor{n} + 2 \ge 6\ucolor{s}\ge 0$} and \inlinef{$2\ucolor{n}^4 + 2\ucolor{n} + 4 > 4\ucolor{s}\ge 0$} as the only \loverpropertyp{\langover} for \rsquaresum~ and \rcubicsum, respectively. 
Unlike for \rsum, not all natural numbers in the range induced by the synthesized \loverpropertyp{\langover} can be obtained, so \name will not yield the same \lunderpropertyp{\langunder}. Instead, for each benchmark \name synthesizes two \lunderpropertiesp{\langunder} that describe particular nondeterministic choices, e.g.,  $\{\ucolor{s} = 0, \ucolor{s} = \ucolor{n}^2\}$ for \rsquaresum.

\mergesort takes an array $\ecolor{A}$ with nondeterministic values as input and computes the number of inverse pairs of the subarray $\ecolor{A}_{\ucolor{s}..\ucolor{e}}$ using merge sort. \name should
synthesize an \loverpropertyp{\langover} like $\ucolor{n} \le (\ucolor{e}-\ucolor{s})(\ucolor{e}-\ucolor{s}+1)/2$ that explains the maximal number of inverse pairs, but due to the nested recursion/loops, \name fails to return any properties within the time limit.

\mypar{Benchmarks from SV-COMP~\cite{SVCOMP24}} These 4 benchmarks model programs that add a constant to each input variable repeatedly. E.g., \jain{2} models the following program:
\begin{lstlisting}[language=C, tabsize=3, 
    basicstyle= \tt \footnotesize, 
    keywordstyle=\color{purple}\bfseries, 
    commentstyle=\color{gray}, 
    xleftmargin=0em, 
    escapeinside=``,
    numbers = none,
    numbersep = 1pt,
    ]
    x = 1; y = 1; 
    while(nondet()) {
        x = x + 2 * nondet(); 
        y = y + 2 * nondet(); }
\end{lstlisting}
The remaining \jain{} benchmarks differ from \jain{2} in the number of variables, the initial values, and the values that are added.
We designed DSLs with basic arithmetic operators including modulo, from which \name can discover properties about the final values of the variables. 
For \jain{2}, \name synthesizes $\{\ucolor{y} \bmod 2 = 1, (\ucolor{x}+\ucolor{y}) \bmod 2 = 0 \}$ as \loverpropertiesp{\langover} and $\{\ucolor{x} \bmod 2 = 1 \land \ucolor{y} \bmod 2 = 1\}$ as \lunderpropertiesp{\langunder}. 
Note that the \loverpropertiesp{\langover} consists of two properties and \lunderpropertiesp{\langunder} only has a single one; however, these properties are equivalent. 
For the \jain{} benchmarks, synthesis of \lunderpropertiesp{\langunder} takes on average 1.38x longer than synthesis of \loverpropertiesp{\langover}; for the former, \name has to discover the entire property at once.


\subsection{Application 1: Incorrectness Reasoning}
\label{App:eval-incorrectness}

\paragraph{Expressible amounts in two kinds of coins}
The \coin benchmarks non-deterministically model the possible dollar amounts that can be represented using coins of two values.

\citet{Peter2019Incorrectness} provides syntax-direct rules to find the weakest under-approximate postcondition and the weakest possible precondition, but this approach has to explicitly unroll loops and introduce existential quantifiers to deal with assignments and nondeterminism, thus resulting in predicates that might be hard to dispatch to a constraint solver.
Using \framework, one can instead use the DSL $\lang$ to customize what properties they are interested in obtaining in the \lwpo and \lwpp.

For example, the following function \exname{coin} takes two integers $a$ and $b$ that are co-prime (shown in \exname{presumes} part) as input, then nondeterministically chooses two non-negative integers $x$ and $y$, and finally returns $ax + by$. 
The more intuitive interpretation of the program is that it represents all amounts that can be expressed using only coins of value $a$ and $b$.
\begin{lstlisting}[language=C, tabsize=3, 
    basicstyle= \tt \footnotesize, 
    keywordstyle=\color{purple}\bfseries, 
    commentstyle=\color{gray}, 
    xleftmargin=0em, 
    escapeinside=``,
    numbers = left,
    numbersep = 1pt,
    ]
    int coin(int a, int b) 
    /*  presumes: [gcd(a,b)==1]
        achieve1: [gcd(a,b)==1 /\ exists x>=0, y>=0. r==a*x+b*y] 
        achieve2: [gcd(a,b)==1 /\ r==a]
        achieve3: [gcd(a,b)==1 /\ r>a*b-a-b]  */
    {
        int x = nondet(); assume(x >= 0);
        int y = nondet(); assume(y >= 0);
        return a * x + b * y;
    }
\end{lstlisting}
The predicates \exname{achieve1}, \exname{achieve2}, and \exname{achieve3} are all valid under-approximation postconditions. 
\exname{achieve1} is the one obtained using the derivation rules by \citet{Peter2019Incorrectness}: It is precise but has an existential quantifier and multiplication, which make it hard to check in later reasoning. 
\exname{achieve2} could be obtained by a dynamic symbolic execution approach~\cite{Cadar2013symexec} that concretizes $x = 1$ and $y = 0$; this postcondition is valid but less precise than \exname{achieve1}.

%
The flexibility of \framework allows us to modify the DSL and not be tied to any specific rule derivation technique. 
Consider for example a situation in which we are not interested in the actual relation captured by the program, but are just interested in identifying a lower bound (or upper bound) above (or below) which all program outcomes can be effectively produced---i.e., an under-approximation of the output range of the function. 
For example, one may care about after some nondeterministic perturbation of the initial state, which states within a certain distance could all be possible results.
%

In terms of \exname{coin}, one can look for such a lower bound using \name, by supplying the following DSL that has the bounding predicate $\ucolor{r} > N_0$ at the top level and such that nonterminal $N_0$ can derive a quadratic expression containing $\ucolor{a}$ and $\ucolor{b}$:
\begin{equation}
\label{eq:coin-grammar}
\begin{array}{rcl}
     C & := & \ucolor{r} > N_0  \\
     N_0 & := & N_1 \mid N_1 + N_1 \mid N_1 - N_1 \\
     N_1 & := & \nonint \mid \nonint + \nonint \mid \nonint - \nonint \\
     \nonint & := & \ucolor{a} \mid \ucolor{b} \mid \ucolor{ab} \mid \ucolor{a}^2  \mid \ucolor{b}^2 \mid 1 \mid 0
\end{array}
\end{equation}

For the DSL in \Cref{eq:coin-grammar}, \name will synthesize the \lwpo \exname{achieve3}, i.e., \inlinef{$\gcd(a,b) = 1 \land r > ab-a-b$}, which states that all numbers greater than $ab - a - b$ can be produced by \exname{coin} assuming $\textrm{gcd}(a, b) = 1$. 
%
Compared to \exname{achieve1} and \exname{achieve2}, \exname{achieve3} guarantees a certain degree of precision and meanwhile meets the needs to obtain a lower bound.

\paragraph{Conditions lead to a hash collision}
The \wppHash{} benchmark models the condition of hash collision after applying a parametric hash function to a set of integers.
The parametric hash function is defined as $f[\ecolor{a}] (x) = \ecolor{a}x \bmod \ucolor{M}$, where $\ecolor{a} \in \{1,\cdots, M-1\}$ is the parameter we can instantiate the function with. 
We also have a set of integers $\ucolor{S}$ to which we want to apply the hash function. 
%

To identify what kind of set $\ucolor{S}$ will possibly lead to a hash collision, we compute the \lwpp $\wpprelang(\ecolor{S_o} = \exname{map}(f[\ecolor{a}], \ucolor{S}), \exname{size}(\ucolor{S}) > \exname{size}(\ecolor{S_o}))$.
%
The following DSL $\lang$ is intended to capture the relation between $\ucolor{S}$ and $\ucolor{M}$ in the \lwpp:
\begin{equation}
\label{eq:wpp-hash-under}    
\begin{array}{rcl}
    D  & := &  \bot \mid AP \mid AP \land AP \mid AP \land AP \land AP\\
    AP &:= & \exname{isPrime}(\ucolor{M}) \mid \neg\exname{isPrime}(\ucolor{M}) \mid N \{\le \mid < \mid = \mid \neq \} N  \\
    N & := & \exname{size}(\ucolor{S}) \mid \exname{modsize}(\ucolor{S}, \ucolor{M}) \mid \ucolor{M} \mid 0 \mid 1 \\
\end{array}
\end{equation}
The function $\exname{modsize}(\ucolor{S}, \ucolor{M})$ computes the size of the set obtained by taking the $\ucolor{M}$-modulus for each number in $\ucolor{S}$. 

Using the DSL from \Cref{eq:wpp-hash-under}, \name synthesizes the following \lwpp:
\begin{equation}
\begin{array}{c@{\hspace{4.0ex}}c}
    \exname{size}(\ucolor{S}) > \exname{modsize}(\ucolor{S}, \ucolor{M}) &
    \exname{size}(\ucolor{S}) \ge \ucolor{M} \land \neg \isPrime(\ucolor{M})  
\end{array}
\end{equation}

The first \lunderproperty is a valid possible precondition since it implies that there are at least two integers in $\ucolor{S}$ that are congruent w.r.t. $\ucolor{M}$, which will be hashed to the same value.

The second \lunderproperty is also a valid one as twofold: \rone if the size of $\ucolor{S}$ is larger than $\ucolor{M}$, there exists two integers in $\ucolor{S}$ are congruent and thus collide. \rtwo if the size of $\ucolor{S}$ is equal to $\ucolor{M}$, and meanwhile $\ucolor{M}$ is not prime, there always exists a bad parameter $\ecolor{a}$ that is not coprime with $\ucolor{M}$ and therefore can cause collision. 
%

By following the syntax-directed rules proposed by~\citet{hoare78properties} to compute the weakest possible precondition, we would get the predicates \inlinef{$\exists \ecolor{a}.~\exname{size}(S) = \exname{size}(\exname{map}(f[\ecolor{a}], \ucolor{S}))$}, which still contains a quantifier and is effectively the same as the original query we were asking---i.e., it does not help us understand the program behavior. 


Furthermore, if we are interested in what set $\ucolor{S}$ will possibly \textit{not} lead to a hash collision, we could set the postcondition as $\exname{size}(\ucolor{S}) = \exname{size}(\ecolor{S_o})$,
and \name will synthesize the \lwpp in same DSL from \Cref{eq:wpp-hash-under} as follows:
\begin{equation}
\begin{array}{c}
    \exname{size}(\ucolor{S}) = \exname{modsize}(\ucolor{S}, \ucolor{M}) 
\end{array}
\end{equation}
which means no two integers in $\ucolor{S}$ are congruent w.r.t $\ucolor{M}$. 
This is a valid possible precondition since when it is satisfied, there is always \textit{some} parameter $\ecolor{a}$ (e.g., 1) such that $\gcd(\ecolor{a}, \ucolor{M}) = 1$, and thus one can prevent hash collision.

\subsection{Application 2: Reasoning about Concurrent Programs}
\label{App:eval-concurrency}

\paragraph{Describing Sources of Deadlock}
The \philosopher~benchmark encodes the dining-philosophers problem where $N$ ``philosophers'' are sitting around a table, and between each pair of philosophers is a single fork (and thus, $N$ total). 
Each philosopher alternatively thinks and eats. 
To eat, a philosopher needs two forks, both the one on the left and the one on the right. 
When finishing eating and back to thinking, they will put down both forks. 
We say there is a deadlock when all philosophers want to eat but cannot get both forks because every philosopher is holding a single fork.

In this example, the so-called \textit{circular wait} is a necessary condition for deadlock, in which there exists a circular chain of threads such that each thread holds resources that are being requested by the next thread in the chain. 
Circular wait can be prevented by adjusting the order in which each thread requests resources.
In terms of \philosopher, each philosopher could either take the left fork first or the right fork first. 

We show \name can be used to understand what execution orders affect whether a deadlock happens.
To do so, we model this problem as query $\exists \ecolor{h}.~\ucolor{dl} = schedule(\ucolor{o_1}, \cdots, \ucolor{o_N}, \ecolor{h})$, where $\ucolor{o_i}\in \{L,R\}$ indicates which fork the philosopher $i$ always takes first. 
We then supply the following two DSLs (the one rooted at nonterminal $B_{\langover}$ is for \loverpropertiesp{\langover}, whereas the one rooted at 
$B_{\langunder}$ is for \lunderpropertiesp{\langunder}):
\begin{equation}
\label{eq:philosopher-grammar-app}
\begin{array}{rcl}
     B_{\langover} &:= & G \Rightarrow D \qquad \qquad B_{\langunder} \ \ :=\ \ G \land D \\
     G & := & \top \mid \nonap \mid \nonap \land \nonap \mid \cdots \mid \nonap \land \nonap \land \nonap \land \nonap \land \nonap\\
     \nonap & := & O = L \mid O = R \\
     O & := & \ucolor{o_1} \mid \cdots \mid \ucolor{o_N} \\
     R & := & \ucolor{dl} \mid \neg \ucolor{dl}
\end{array}
\end{equation}
For the case involving three threads/philosophers ($N=3$), \name synthesizes the following \loverpropertiesp{\langover}, which informally state that deadlock can be prevented by having two of the threads disagree on their fork choice:
\begin{equation}
\begin{array}{c@{\hspace{4.0ex}}c@{\hspace{4.0ex}}c}
    (\ucolor{o_0} = L \land \ucolor{o_2} = R) \Rightarrow \neg \ucolor{dl} &
    (\ucolor{o_2} = L \land \ucolor{o_1} = R) \Rightarrow \neg \ucolor{dl} &
    (\ucolor{o_1} = L \land \ucolor{o_0} = R) \Rightarrow \neg \ucolor{dl} 
\end{array}
\end{equation}

For the same $N$, \name also synthesizes the following \lunderpropertiesp{\langunder}, which 
exactly characterize the two cases in which a deadlock can happen (first two properties) and also capture that there exists an execution that does not lead to a deadlock (last property).
\begin{equation}
\begin{array}{c@{\hspace{4.0ex}}c@{\hspace{4.0ex}}c}
    \ucolor{o_0} = L \land \ucolor{o_1} = L \land \ucolor{o_2} = L \land \ucolor{dl} &
    \ucolor{o_0} = R \land \ucolor{o_1} = R \land \ucolor{o_2} = R \land \ucolor{dl} &
    \neg \ucolor{dl}
\end{array}
\end{equation}



The 4 \resource{} benchmarks focus on how the amount of resources affects the deadlock. Consider a simple resource allocator that contains $M$ types of resources $R_{1},\cdots R_{M}$, and initially has $n_{i}$ units of resource $R_i$.
The allocator receives $T$ threads, each containing a list of resources the thread needs and in what order, and at each step, it needs to decide which next resource of each thread should be allocated.
Once all the resources in the list are allocated, the thread completes its job and releases them altogether. 
However, if a request cannot be fulfilled due to the lack of resources of that type, the thread waits. 
We say the allocator is in a deadlock when multiple threads are waiting and no progress can be made.

We show how \name tells us what resource amounts never lead to a deadlock and what resource amounts possibly cause a deadlock. The benchmark \resource{2} is a case where $T = 2$ and $M=2$, and where thread-1 requests resources $[R_1, R_2, R_1, R_2]$ and thread-2 requests resources $[R_2, R_1, R_2, R_1]$. 
In \framework, a completely nondeterministic allocator that could allocate any resource to any waiting thread could be modeled as the query \inlinef{$\exists \ecolor{h}.~ \ucolor{dl} = schedule(\ucolor{n_1}, \ucolor{n_2}, \ecolor{h})$}, where $\ucolor{dl}$ is a Boolean variable that indicates whether deadlock happens, and the $\ecolor{h}$ is an array used to model the sequence of nondeterministic choices during scheduling. 
%
\Cref{fig:allocator} shows the actual \sketch program used to define the semantics of the scheduler.

We supply \name with a similar DSL to \Cref{eq:philosopher-grammar-app}, with replace the production rule for $AP$ by \inlinef{$AP := N \{<|\le|=\} N$}, where $N$ can derive every $\ucolor{n_i}$ and integer constants. \name synthesizes the \loverpropertiesp{\langover} in \Cref{eq:resource2-over-result} that tells us about resource amounts for which a deadlock must or must not happen. For example, the third \loverpropertyp{\langover} states that ``if there are more than 2 units of $R_1$ and more than 3 units of $R_2$, \textit{no scheduling order} can lead to a deadlock.'' 
\begin{equation}
\label{eq:resource2-over-result}
\begin{array}{c@{\hspace{4.0ex}}c@{\hspace{4.0ex}}c@{\hspace{4.0ex}}c}
    \ucolor{n_1} \le 1 \Rightarrow \ucolor{dl} &
    \ucolor{n_2} \le 1 \Rightarrow \ucolor{dl} &
    (\ucolor{n_1} \ge 2 \land \ucolor{n_2} \ge 3) \Rightarrow \neg\ucolor{dl} &
    (\ucolor{n_1} \ge 3 \land \ucolor{n_2} \ge 2) \Rightarrow \neg\ucolor{dl} 
\end{array}
\end{equation}
For the same problem, \name synthesizes the \lunderpropertiesp{\langunder} in \ref{eq:resource2-over-result} that tells us about resource amounts for which a deadlock may or may not happen. For example, the third \lunderpropertyp{\langunder} states that ``if both types of resources are available in a quantity no more than $3$, \textit{there exists a scheduling order} that leads to a deadlock.'' 
\begin{equation}
\label{eq:resource2-under-result}
\begin{array}{c@{\hspace{4.0ex}}c@{\hspace{4.0ex}}c@{\hspace{4.0ex}}c}
    \ucolor{n_1} \le 2 \land \ucolor{dl} &
    \ucolor{n_2} \le 2 \land \ucolor{dl} &
    \ucolor{n_1} \le 3 \land \ucolor{n_2} \le 3 \land \ucolor{dl} &
    \ucolor{n_1} \ge 2 \land \ucolor{n_2} \ge 2 \land \neg \ucolor{dl} 
\end{array}
\end{equation}

All \resource{} benchmarks are instances of the problem above, of which $(T, M, \text{length of the request list})$ are $(2,2,2)$, $(2,2,4)$, $(3,2,4)$, and $(2, 3, 8)$. Even for the hardest instance of $(2, 3, 8)$, \name synthesized the best \loverpropertiesp{\langover} within 160 seconds (from $4.37 \cdot 10^{12}$ properties), and the best \lunderpropertiesp{\langunder} within 60 seconds (from $2.75 \cdot 10^{11}$ properties)

\begin{figure}[!thbp]
\begin{lstlisting}[language=C++, tabsize=3, 
    basicstyle= \tt \footnotesize, 
    keywordstyle=\color{purple}\bfseries, 
    commentstyle=\color{gray}, 
    xleftmargin=0em, 
    escapeinside=``,
    numbers = left,
    numbersep = 1pt,
    ]
    int[] th1, th2; 
    bool schedule(int n1, int n2, int[] h) {
        int pc1 = 0, pc2 = 0; 
        for(int i = 0; i < len(th1) + len(thread2); i++) {
            bool ready1 = pc1 < len(th1) && resource of type th1[pc1] available
            bool ready2 = pc2 < len(th2) && resource of type th2[pc2] available
            int allocto;
            if(h[i] == 1) { // prioritize thread 1
                if(ready1) allocto = 1;
                else if(ready2) allocto = 2;
                else return true; 
                // the next resource of both threads is not enough, deadlock.
            }
            else if(h[i] == 2) { // prioritize thread 2
                if(ready2) allocto = 2;
                else if(ready1) allocto = 1;
                else return true; 
                // the next resources of both threads are not enough, deadlock.
            }
            allocate the resource to th_allocto;
            pc_allocto++;
            if (pc_allocto == len(th_allocto))
                release the resources it occupied; 
        }
        return false; // both threads finished, no deadlock.
    }
\end{lstlisting}
\vspace{-6mm}
\caption{Resource allocator for two threads}
\label{fig:allocator}
\vspace{-6mm}
\end{figure}

\paragraph{Preventing Race Conditions}
The 3 \race{} benchmarks are about describing possible race conditions in concurrent programs. 
In each of them, there are 2 threads that access and modify a shared variable using the methods \texttt{set()} and \texttt{get()}. For example, \Cref{fig:race2} shows the code of 2 threads in benchmark \race{2}.

\begin{wrapfigure}{r}{0.4\textwidth}
\vspace{-2mm}
    \centering
    \begin{multicols}{2}
\begin{lstlisting}[ tabsize=3, 
    basicstyle= \tt \footnotesize, 
    keywordstyle=\color{purple}\bfseries, 
    commentstyle=\color{gray}, 
    xleftmargin=-0.5em, 
    escapeinside=``,
    morekeywords = {get, set},
    numbers = none
    ]
   // Thread 1
   0: t <- get()
   1: t <- t + 1
   2: t <- get()
   3: t <- t + 1
   4: set(t)
\end{lstlisting}
%
\columnbreak
\begin{lstlisting}[ tabsize=3, 
    basicstyle= \tt \footnotesize, 
    keywordstyle=\color{purple}\bfseries, 
    commentstyle=\color{gray}, 
    xleftmargin=-0.5em, 
    escapeinside=``,
    morekeywords = {get, set},
    numbers = none
    ]
  // Thread 2
  0: t <- get()
  1: t <- t - 1
  2: set(t)
  3: t <- t - 1
  4: set(t)
\end{lstlisting}
%
\end{multicols}
\vspace{-6mm}
    \caption{Two threads in \race{2}}
    \label{fig:race2}
\vspace{-4mm}    
\end{wrapfigure}

When there is no possible context switching between the 2 threads in \race{2}, if the initial value of the variable \texttt{t} is 0, its final value should be $-1$ (which we call the expected result). 
However, context switching can cause different interleaving of the threads to produce values different than $-1$---i.e., there exists a data race. 
Such a data race is typically prevented by introducing critical sections, in which instructions must be executed atomically. 

We show how \name can be used to identify the minimum part of the code that should be made atomic for the code to be race free.
To model the problem in \name, we introduce two variables $\atomcons{1}$ and $\atomcons{2}$, which will be used to capture what lines in each thread should be executed atomically.
%
The predicate $\atomp(\atomcons{i}, l, r)$ holds true if the instructions from line $l$ to line $r$ of thread-$i$ should be executed atomically. 

Now we can model whether a race happens as a query $\exists \ecolor{h}. ~\ucolor{race} = schedule(\atomcons{1}, \atomcons{2}, \ecolor{h})$, where \ucolor{race} captures whether a race condition can happen.
We provide the following two DSLs (the one rooted at nonterminal $B_{\langover}$ is for \loverpropertiesp{\langover}, whereas the one rooted at 
$B_{\langunder}$ is for \lunderpropertiesp{\langunder}):
\begin{equation}
\label{eq:race-grammar}
\begin{array}{rcl}
     B_{\langover} &:= & G \Rightarrow D \qquad \qquad B_{\langunder} \ \ :=\ \ G \land D \\
     G & := & \top \mid \nonap \mid \nonap \land \nonap \mid \cdots \mid \nonap \land \nonap \land \nonap \land \nonap \land \nonap\\
     \nonap & := & \atomp(AC, I, I) \mid \neg \atomp(AC, I, I) \\
     I & := & \text{all line numbers} \\
     AC & := & \atomcons{1} \mid \atomcons{2} \\
     D & := & \ucolor{race} \mid \neg \ucolor{race}
\end{array}
\end{equation}
The predicate $\atomp(\atomcons{i}, l, r)$ is implemented as a conjunction $\bigwedge_{k=l}^{r-1} \textit{noSwitch}(\atomcons{i}, k)$, where the predicate $\textit{noSwitch}(\atomcons{i}, k)$ is true if in thread-1, the instruction $k+1$ must be immediately executed after instruction $k$.
This implementation makes it so that the predicate $\atomp(\atomcons{i}, l_1, r_1)$ implies $\atomp(\atomcons{i}, l_2, r_2)$ when $l_1 \le l_2$ and $r_1 \ge r_2$. 
Since \name looks for a tightest properties, such an implementation lets \name reason about what is the smallest needed atomic execution.
%


For \race{2}, \name synthesizes the following \loverpropertiesp{\langover}, which informally states that setting line 2 to line 4 of thread-1 and line 0 to line 4 of thread-2 as critical sections can prevent data race:
\begin{equation}
\begin{array}{c}
    (\atomp(\atomcons{1}, 2, 4) \land \atomp(\atomcons{2}, 0, 4)) \Rightarrow \neg \ucolor{race} 
\end{array}
\end{equation}
We can observe that the effect of the first two instructions of thread-1 is overwritten by the third instruction, so it is unnecessary to include them in a critical section.

For \race{2}, \name to synthesizes the following \lunderpropertiesp{\langunder}, which states that setting line 2 to line 4 of thread-1 and whole thread-2 as critical sections is in fact \textit{necessary}!
\begin{equation}
\begin{array}{c@{\hspace{4.0ex}}c@{\hspace{4.0ex}}c}
    \neg \atomp(\atomcons{1}, 2, 4) \land \ucolor{race} &
    \neg \atomp(\atomcons{2}, 0, 4) \land \ucolor{race} &
    \neg \ucolor{race}
\end{array}
\end{equation}

This example shows how over- and under-approximated reasoning can be cleverly combined in \framework to understand necessary and sufficient interventions in preventing data races.

\subsection{Application 3: Solving two-player games}
\label{App:eval-game}

\paragraph{Definition of safety games}
A safety game consists of a game graph $G = \langle (S, E), (S_1, S_2)\rangle$ and a safety objective $F$. In the graph, $S$ is a set of states partitioned into \playerA~states $S_1$ and \playerB~states $S_2$, $E\subseteq S \times S$ is a set of edges in which each edge connects a state in $S_1$ and a state in $S_2$. 
The safety objective $F\subseteq S$ is a set of safe states. 
\playerA's goal is to remain in safe states, while \playerB's goal is to visit unsafe states at least once.
%

We assume both players play the game according to a finite-state \textit{memoryless} strategy that is independent of the action history and depends only on the current states. 
For safety games, there always exists a memoryless winning strategy. 

\paragraph{More benchmarks}
\label{App:eval-game-benchmarks}

The \nimgame{2}~game is played with 2 heaps of pebbles with number $\ucolor{n_1}$ and $\ucolor{n_2}$. On each turn, a player must remove at least one pebble and may remove any number of pebbles if they all come from the same heap. The goal of the game is to be the player to remove the last pebble. One synthesized must strategy is \Cref{eq:nim2-over-result}.
\begin{equation}
\label{eq:nim2-over-result}
\begin{aligned}
    &\forall i,j.  (i < j \Rightarrow \strategyA[i, j].\exname{heap} = 2)
    \land  \forall i,j. (i < j \Rightarrow \strategyA[i, j].\exname{num} = j - i) \\
    \land &~\forall i,j. (i > j \Rightarrow \strategyA[i, j].\exname{heap} = 1)  \land  \forall i,j. (i > j \Rightarrow \strategyA[i, j].\exname{num} = i - j)  \land \ucolor{n_1} != \ucolor{n_2} \Rightarrow \ucolor{w} = T
\end{aligned}
\end{equation}
where \exname{heap} denotes the heap from which pebbles are taken when the numbers of pebbles in two heaps are $i$ and $j$,
and \exname{num} denotes the number of pebbles taken. 
\Cref{eq:nim2-over-result} essentially states that one can win if the initial two heaps of stones are different, and always keep them the same after taking.

\begin{wrapfigure}{r}{0.35\textwidth}
\vspace{-6mm}
    \centering
    \begin{lstlisting}[ tabsize=3, 
    basicstyle= \tt \footnotesize, 
    keywordstyle=\color{purple}\bfseries, 
    commentstyle=\color{gray}, 
    xleftmargin=-2.5em, 
    escapeinside=``,
    language = C,
    morekeywords = {assert},
    numbers = none
    ]
    temp = 20.5; isOn = 0;
    while(*) {
        assert(20 <= temp <= 25);
        isOn ++;
        if (isOn == ??) {
            temp += ??; isOn = 0;
        }
        temp -= (temp - 19) / 10;
    }
\end{lstlisting}
    \label{fig:thermostat}
\vspace{-4mm}    
\end{wrapfigure}
The \tempgame~game models a controller for a thermostat shown in the right.  We consider strategies that set \texttt{isOn} as 1 at the $k$-th of every $n$ times. \name synthesizes 9 must strategies (and also 9 equivalent may strategies since there is no adversary in \tempgame~game), e.g., \inlinef{$2 \le \strategyA.n \le 7 \land \strategyA.k = 1 \Rightarrow \ucolor{w}$}, which states that the thermostat can keep the temperature in $[20, 25]$ by increasing two degree in the first second of every $n(2\le n \le 7)$ seconds. 

The \numgame{} benchmarks consider games played over a one-dimensional grid---i.e., an integer.
Each game is a 4-tuple $(v, A_1, A_2, S)$, where $v$  is an initial integer value, $A_1 = \{f_1, f_2,\cdots, f_n\}$ and $A_2 = \{g_1, g_2, \cdots, g_m\}$ (such that $f_i, g_i \in \mathbb{Z} \to \mathbb{Z}$) are the actions set that \playerA~and \playerB~use to manipulate the integer (e.g., increments, decrements, etc.), and $S \subseteq \mathbb{Z}$ is a the set of integers \playerA~ wants to stay in to win the game.

The problems \numgame{1}=$(v, \{nop, -1\}, \{\times 2, +1\}, [0, 4])$ and \numgame{2}=$(\{nop, -1\}, \{nop, +1\}, [0, 4])$ are two instances of the game above, where the initial integer value $v$ is left unspecified.

We want to use \name to understand the relationship between the value of $v$ and winning strategies, and thus supply \name with the following DSL:
\begin{equation}
\label{eq:num1-must-grammar}
\begin{array}{rcl}
    B_{must}  &:= &  G \Rightarrow R   \qquad \qquad 
    B_{may} \ \ := \ \ G \land R \\
     G & := & \top \mid \nonap \mid \nonap \land \nonap \mid \cdots \mid \nonap \land \nonap \land \nonap \land \nonap \land \nonap\\
     \nonap & := & \decision(\strategyA, S_1, A_1) \mid \ucolor{v} = \{0 | 1 |2 |3 |4\}\\
     S_1 &:= & 0 \mid 1 \mid 2 \mid 3 \mid 4 \\
     A_1 &:= & nop \mid -1 \\
     R & := & \ucolor{w} = \{T\mid F\}
\end{array}
\end{equation}
Using the DSL rooted at $B_{must}$, for \numgame{1} \name synthesized the \loverpropertiesp{\langover} that consists of 
\begin{equation}
\label{eq:num1-over-result}
(\ucolor{v} = 0 \land \strategyA[0] = nop \land \strategyA[2] = -1) \Rightarrow \ucolor{w} = T
\end{equation}
as well as other 9 \loverpropertiesp{\langover} in 15 seconds. 
The property states that if the initial value is $0$, \playerA can remain in the range $[0,5]$ by performing the action $nop$ at state 0 and the action $-1$ at state $2$.
Note that the strategy does not need to be defined at any of the other infinitely many states.

Using the DSL rooted at $B_{may}$, for \numgame{1} \name synthesized the \lunderpropertiesp{\langunder} that consists of 
\begin{equation}
\label{eq:num1-under-result}
 \strategyA[1] = nop \land \strategyA[2] = nop \land \ucolor{w} = F
\end{equation} as well as other 16 \lunderpropertiesp{\langunder} in 10 seconds. 
The property states that by performing the action $nop$ at states 0 and 2, there exists a \playerB strategy that makes  \playerA lose the game.


%% file: table-mining.tex
\begin{table}[tp]
\caption{
Evaluation of \name for mining specifications. 
$|\langunder|$ is the size of grammar.
LoC is the number of lines of \sketch code used to write the semantics of programs and operators.
\#P is the number of \lunderproperties synthesized by \name.
Num and T(s) are the number of times each primitive was called and the total time it consumed, respectively.
%
%
For \synquid benchmarks, we only show 5 with the longest running time.
\label{tab:benchmarks-spyro}}
{\footnotesize
\setlength{\tabcolsep}{2pt}
\begin{tabular}{ccccrrrrrrrr} 
\toprule[.1em]
\multicolumn{2}{c}{\multirow{2}{*}[-0.4ex]{Problem}}
& \multicolumn{1}{c}{\multirow{2}{*}{$|\langunder|$}}
& \multicolumn{1}{c}{\multirow{2}{*}{LoC}}
& \multicolumn{1}{c}{\multirow{2}{*}{\#P}}
& \multicolumn{2}{c}{\checkimpl} & \multicolumn{2}{c}{\cpunder} & \multicolumn{2}{c}{\synth} &   Total. \\
\cmidrule{6-12}
& & & & & \hspace{2mm}Num & T(s) & \hspace{1mm}Num & T(s) & Num & T(s)  &  T(s) \\
\midrule[.1em]
\parbox[t]{2mm}{\multirow{7}{*}{\rotatebox[origin=c]{90}{\sygus}}}
    & \maxtwo & $ 1.60\cdot 10^{5}$ & 7 & 2 & 13 & 0.44 & 10 & 0.82 & 5 & 0.12  & 1.51 \\
    & \maxthree & $ 8.94\cdot 10^{5}$ & 15 & 3 & 30 & 0.53 & 15 & 1.66 & 18 & 0.73  & 2.92 \\
     & \maxfour & $ 5.10\cdot 10^{8}$ & 31 & 4 & 89 & 2.02 & 34 & 18.38 & 59 & 10.90  &  31.32 \\   
     & \maxfive & $ 4.72\cdot 10^{1}$ & 63 & 5 & 189 & 5.33 & 56 & 136.77 & 138 & 150.14  &  292.26 \\
     & \diff & $ 5.68\cdot 10^{7}$ & 7 & 2 & 28 & 0.45 & 11 & 7.45 & 19 & 0.93  &  8.84 \\
     & \arraytwo & $ 1.01\cdot 10^{6}$ & 9 & 3 & 34 & 0.50 & 14 & 0.99 & 23 & 1.21  & 2.71 \\
     & \arraythree & $ 7.28\cdot 10^{8}$ & 11 & 4 & 86 & 1.39 & 27 & 14.78 & 63 & 20.55  &  36.84 \\
\cmidrule{1-12}
  \parbox[t]{2mm}{\multirow{6}{*}{\rotatebox[origin=c]{90}{Arithmetic}}}
     & \iaabs{-1} & $ 1.01\cdot 10^{6}$ & 7 & 2 & 17 & 0.27 & 10 & 0.35 & 9 & 0.20 & 0.83 \\
     & \iaabs{-2} & $ 2.90\cdot 10^{10}$ & 7 & 2 & 29 & 0.45 & 11 & 4.29 & 20 & 1.17 & 5.92 \\
     & \ialinsum{-1} & $ 3.01\cdot 10^{6}$ & 9 & 2 & 7 & 0.11 & 4 & 0.07 & 4 & 0.06  & 0.25 \\
     & \ialinsum{-2} & $ 2.90\cdot 10^{10}$ & 9 & 2 & 41 & 0.66 & 17 & 3.64 & 27 & 2.87 &  7.18 \\
     & \ianonlinsum{-1} & $ 3.01\cdot 10^{6}$ & 9 & 2 & 22 & 0.36 & 12 & 0.67 & 12 & 0.35 &  1.39 \\
     & \ianonlinsum{-2} & $ 1.48\cdot 10^{13}$ & 9 & 2 & 70 & 9.07 & 16 & 378.68 & 56 & 95.40 & 483.17 \\
\cmidrule{1-12}
 \parbox[t]{2mm}{\multirow{4}{*}{\rotatebox[origin=c]{90}{\synquid}}} 
  & \lappend & $ 4.97\cdot 10^{8}$ & 75 & 2 & 82 & 3.50 & 22 & 58.07 & 62 & 40.84 & 102.56 \\
  & \ldelete & $ 3.01\cdot 10^{6}$ & 83 & 2 & 51 & 1.86 & 19 & 254.67 & 34 & 20.81 & 277.39 \\
  & \ldeleteall & $ 3.01\cdot 10^{6}$ & 82 & 2 & 48 & 2.12 & 16 & 39.26 & 34 & 19.64 & 61.17 \\
  & \ldrop & $ 4.75\cdot 10^{8}$ & 75 & 2 & 79 & 2.42 & 22 & 37.67 & 59 & 39.04 & 79.16 \\
  & \qenqueue & $ 6.00\cdot 10^{5}$ & 147 & 1 & 8 & 17.98 & 5 & 48.44 & 4 & 0.56 & 67.02 \\
\bottomrule[.1em]
\bottomrule[.1em]
\end{tabular}
}
\end{table}

%% file: table-nondeter.tex
\begin{wraptable}{r}{0.41\textwidth}
\vspace{-10pt}
\caption{
Evaluation results of \typetwo~test set. 
%
$|\exists|$ is the size of the domain of the existentially quantified variables.  
%
\#P and T(s) are the number of properties and synthesis time for both \loverproperties and \lunderproperties (- denotes timeouts).
%
}
\label{tab:benchmarks-nondeter}

\vspace{-5pt}
{\footnotesize
\setlength{\tabcolsep}{2pt}
\begin{tabular}{ccrrrrrr} 
\toprule[.1em]
\multicolumn{2}{c}{\multirow{2}{*}[-0.4ex]{Problem}}
& \multicolumn{1}{c}{\multirow{2}{*}{LoC}}
& \multicolumn{1}{c}{\multirow{2}{*}{$|\exists|$}}
 & \multicolumn{2}{c}{$\lang$-cons.} & \multicolumn{2}{c}{$\lang$-impl.}\\
    \cmidrule{5-8}
      & & & & \#P & T(s) & \#P & T(s) \\
\midrule[.1em]
 \parbox[t]{2mm}{\multirow{15}{*}{\rotatebox[origin=c]{90}{Nondeterminism}}}  
      & \shuff{3}   & 33 & $4$   & 1  & 0.14   & 1  & 1.11   \\
     & \shuff{4}   & 33 & $27$   & 1 & 0.30 & 1 & 1.17  \\
    & \shuff{5}   & 33 & $256$   & 1  & 33.19   & 1  & 15.26   \\
     & \rsum         & 28 & ${\sim}10^{5}$            & 1  & 5.51   & 1  & 10.29  \\
     & \rsquaresum   & 28 & ${\sim}10^{5}$            & 1  & 11.25  & 2  & 4.53   \\
     & \rcubicsum    & 28 & ${\sim}10^{5}$            & 1  & 4.07   & 2  & 4.80   \\
     & \mergesort    & 31 & ${\sim}10^{7}$  & -  & -      & -  & -      \\
     & \jain{1}      & 27 & ${\sim}10^{12}$ & 1  & 0.91   & 1  & 0.63   \\
     & \jain{2}      & 29 & ${\sim}10^{12}$ & 2  & 26.94  & 1  & 27.19  \\
     & \jain{4}      & 34 & ${\sim}10^{12}$ & 3  & 17.34  & 1  & 37.10  \\
     & \jain{6}      & 34 & ${\sim}10^{12}$ & 3  & 161.80 & 1  & 131.91 \\
     & \wppBubble{3} & 33 & \numf{64}            & 6  & 9.63   & 8  & 7.34   \\
     & \wppBubble{4} & 33 & \numf{6561}          & 17 & 392.00 & 22 & 136.56 \\
     & \wppSwap{3}   & 33 & ${\sim}10^{7}$   & 7  & 9.30   & 6  & 7.27   \\
     & \wppSwap{4}   & 33 & ${\sim}10^{9}$   & 24 & 164.61 & 19 & 69.77  \\
\bottomrule[.1em]
\bottomrule[.1em]
\end{tabular}
}
\vspace{-5pt}
\end{wraptable}